\documentclass[reqno,oneside,11pt]{amsart}

\usepackage{amssymb,graphicx,braket,amsmath,amsfonts,amsthm}
\usepackage{hyperref}
\usepackage[utf8]{inputenc}
\usepackage{braket}
\usepackage{color}

\usepackage{stackengine,scalerel}
\newcommand\textaltcolon{\ensurestackMath{\stackon[0.5ex]{\circ}{\circ}}}
\newcommand\altcolon{\savestack\Tmp{\raisebox{-2.7pt}{$\textaltcolon$}}%
  \dp\Tmpcontent=\dimexpr\dp\Tmpcontent-2.7pt\relax%
  \mathrel{\scalerel*{\Tmp}{:}}}

\usepackage{array,caption,diagbox}

\theoremstyle{plain}
\newtheorem{theorem}{Theorem}[section]
\theoremstyle{plain}
\newtheorem{proposition}[theorem]{Proposition}
\newtheorem{remark}[theorem]{Remark}

\newtheorem{example}[theorem]{Example}
\newtheorem{lemma}[theorem]{Lemma}

\theoremstyle{definition}
\newtheorem{definition}[theorem]{Definition}

\numberwithin{equation}{section}

\let\frontmatter\relax
\def\mainmatter{\def\baselinestretch{1.1}\normalfont}

\usepackage[margin=1in]{geometry}

\makeatletter
\renewcommand{\section}{\@startsection
{section}
{1}
{\z@}
{-\baselineskip}
{0.8\baselineskip}
{\centering\scshape\large}} 

\renewcommand{\subsection}{\@startsection
{subsection}
{2}
{\z@}
{-0.8\baselineskip}
{0.5\baselineskip}
{\normalfont \bf \normalsize}} 

\renewcommand{\subsubsection}{\@startsection
{subsubsection}
{3}
{\z@}
{-0.8\baselineskip}
{0.5\baselineskip}
{\normalfont \it \normalsize}} 
\makeatother

\let\emptyset\varnothing

\usepackage{mathtools}

\newcommand{\eprint}[1]{\href{http://arxiv.org/abs/#1}{\texttt{arXiv\string:\allowbreak#1}}}

\def\baselinestretch{1.1}

\def\ket#1{\left|\,#1\,\right\rangle}
\def\baselinestretch{1.1}
\def\bea{\begin{eqnarray*}}
\def\eea{\end{eqnarray*}}

\definecolor{darkgreen}{rgb}{0.1, 0.8, 0.1}

\begin{document}
\frontmatter

\begin{flushright}
CALT-2019-025
\end{flushright}

\bigskip

\title{Super Quantum Airy Structures}

\author{Vincent Bouchard}
\address{Department of Mathematical \& Statistical Sciences,
University of Alberta, 632 CAB\\
Edmonton, Alberta, Canada T6G 2G1}
\email{vincent.bouchard@ualberta.ca}

\author{Pawe\l  \ Ciosmak}
\address{Faculty of Mathematics, Informatics and Mechanics, University of Warsaw, ul. Banacha
2, 02-097 Warsaw, Poland}
\email{p.ciosmak@mimuw.edu.pl}

\author{Leszek Hadasz}
\address{M. Smoluchowski Institute of Physics, Jagiellonian University, ul. \L ojasiewicza 11,
30-348 Krak\'ow, Poland}
\email{leszek.hadasz@uj.edu.pl}

\author{Kento Osuga}
\address{School of Mathematics and Statistics, University of Sheffield, The Hicks Building,
Hounsfield Road, Sheffield, S3 7RH, United Kingdom}
\email{K.Osuga@sheffield.ac.uk}

\author{B\l a\.zej Ruba}
\address{M. Smoluchowski Institute of Physics, Jagiellonian University, ul. \L ojasiewicza 11,
30-348 Krak\'ow, Poland}
\email{blazej.ruba@doctoral.uj.edu.pl}

\author{Piotr Su\l kowski}
\address{Faculty of Physics, University of Warsaw, ul. Pasteura 5, 02-093 Warsaw, Poland \&
Walter Burke Institute for Theoretical Physics, California Institute of Technology, Pasadena, CA 91125, USA}
\email{psulkows@fuw.edu.pl}



\begin{abstract}

We introduce super quantum Airy structures, which provide a supersymmetric generalization of quantum Airy structures. We prove that to a given super quantum Airy structure one can assign a unique set of free energies, which satisfy a supersymmetric generalization of the topological recursion. We reveal and discuss various properties of these supersymmetric structures, in particular their gauge transformations, classical limit, peculiar role of fermionic variables, and graphical representation of recursion relations. Furthermore, we present various examples of super quantum Airy structures, both finite-dimensional -- which include well known superalgebras and super Frobenius algebras, and whose classification scheme we also discuss -- as well as infinite-dimensional, that arise in the realm of vertex operator super algebras. 

\end{abstract}

\maketitle
\tableofcontents
\mainmatter

\section{Introduction}

In recent years we have learnt that solutions to a plethora of problems in physics and mathematics that involve some form of quantization arise from a universal system of recursive equations, referred to as the Chekhov-Eynard-Orantin topological recursion \cite{CE,EO,EO2}. The topological recursion was originally discovered in the realm of matrix models as a way of solving loop equations, which enables computation of the free energy to all orders in the large $N$ expansion, based on the information encoded in the spectral curve of the matrix model \cite{CE}. Soon after, Eynard and Orantin realized that the topological recursion can be formulated independently of matrix models, as a tool that assigns symplectic invariants to a large class of algebraic curves (which play the role of spectral curves in the context of matrix models) \cite{EO,EO2}. More recently, a more general and abstract reformulation of the topological recursion was provided in the form of quantum Airy structures \cite{KS,ABCD,Hadasz:2019vak,Ruba:2019} (see also the lecture notes \cite{Notes by Gaetan}). 

To date, the study of the topological recursion proceeded along two (of course interrelated) main lines. First, various generalizations thereof have been discovered, such as those mentioned above, as well as its $\beta$-deformed version \cite{Chekhov:2009mm}, the formulation for curves with higher ramifications and the global version \cite{BHLMR,BE}, the blobbed version \cite{Borot:2015hna}, the geometric recursion \cite{GeometricRecursion}, etc. Second, a lot of effort has been put in identifying various systems and problems whose solutions are captured by the topological recursion. This led to the simplification or better understanding of the structure of previous solutions, or to new solutions of those problems. The topological recursion turned out to play such a role in areas such as topological strings and Gromov-Witten theory \cite{Bouchard:2007ys,Bouchard:2011ya,DBOSS,EO3,FLZ}, the theory of quantum curves \cite{Gukov:2011qp,Mulase:2012tm,Bouchard:2016obz}, knot theory \cite{Dijkgraaf:2010ur,Borot:2012cw}, Hurwitz theory \cite{Bouchard:2007hi,BSLM,Eynard:2009xe,Do:2012udk,Alexandrov:2018ncq}, just to name a few.

The main aim of this paper is to follow the first line of research mentioned above, and to introduce a supersymmetric generalization of quantum Airy structures, which we call super quantum Airy structures. After defining a super quantum Airy structure, as in the original (non-supersymmetric) case we assign to it free energies, and prove their existence and uniqueness. We write down recursion relations satisfied by these free energies: these recursion relations generalize the original (non-supersymmetric) topological recursion. Among various interesting features of super quantum Airy structures, we reveal a peculiar role played by fermionic variables. We provide several finite- and infinite-dimensional examples of super quantum Airy structures, and hint how their classification could be conducted. 
Furthermore, we expect that super quantum Airy structures should have interesting applications in various contexts, which we briefly mention in what follows\footnote{For the subsequent progress in this context see \cite{STR}}. 

In view of the matrix model origin of the (non-supersymmetric) topological recursion and its reformulation in terms of quantum Airy structures, one may expect similar connections in the supersymmetric context. However, such relations are obscure, even though corresponding supersymmetric structures in matrix models are known. Indeed, supersymmetric generalizations of matrix models, referred to supereigenvalue models, have been introduced and discussed some time ago \cite{AlvarezGaume:1991jd,Becker:1992rk}, and also more recently \cite{Bouchard:2018anp,Ciosmak:2017ofd,Ciosmak:2016wpx,Ciosmak:2017omd,Osuga:2019uqc}. By construction, loop equations for such supereigenvalue models can be rewritten in the form of super-Virasoro constraints. This generalizes the reformulation of loop equations in terms of Virasoro constraints in the non-supersymmetric case, and thus one might hope that super-Virasoro constraints for supereigenvalue models lead immediately to supersymmetric topological recursion. However, in \cite{Bouchard:2018anp,Osuga:2019uqc} it was shown that such a generalization is not automatic: one can indeed write down a recursive system that determines the partition function of a supereigenvalue model, but it is augmented by an auxiliary equation, which does not have a simple interpretation. In a sense, this makes the super quantum Airy structures that we introduce here even more interesting, and revealing their meaning in the context of supereigenvalue models is an important task. 

A similar situation arises in the context of quantum curves. On the one hand, for large classes of spectral curves it was shown that the topological recursion can be used to reconstruct the wave-function associated to a given classical curve, and at the same time the corresponding quantum curve \cite{Gukov:2011qp,Mulase:2012tm,Bouchard:2016obz}. On the other hand, various types of supersymmetric quantum curves (also called super quantum curves) have been constructed recently in the formalism of supereigenvalue models, and reinterpreted from the conformal field theory point of view \cite{Ciosmak:2017ofd,Ciosmak:2016wpx,Ciosmak:2017omd}. Therefore one might expect that super quantum Airy structures provide a more general framework to construct super quantum curves, and to develop their theory further. We hope to address this issue in future work. 

Super quantum Airy structures may also have interesting connections with various problems in enumerative geometry: we postpone this analysis for future work. They may be related to enumerative problems involving odd cohomology classes, such as Gromov-Witten theory of non-singular target curves. For instance, in \cite{Okounkov:2003rf} Okounkov and Pandharipande show that these invariants are encoded in a set of bosonic and fermionic operators forming a representation of a super-Virasoro algebra. These operators could potentially be related to super quantum Airy structures. 

Super quantum Airy structures could also play a role in the theory surrounding Mirzakhani's recursion relations for the Weil-Petersson volumes of the moduli spaces of Riemann surfaces \cite{Mirzakhani}. In the non-supersymmetric case, it was shown that Mirzakhani's recursion relations can be transformed into the form of the topological recursion \cite{Eynard:2007fi,Mulase:2006baa}, which plays a fundamental role in the connection between Jackiw-Teitelboim (JT) gravity and matrix models \cite{SSS}. Very recently, Stanford and Witten generalized this fascinating story to the supersymmetric realm \cite{Stanford:2019vob}. In the process, they found a generalization of Mirzakhani's recursion relations to the volumes of the moduli spaces of super Riemann surfaces. Those new recursion relations could presumably be related to the supersymmetric topological recursion that we present here, and the corresponding super quantum Airy structures. 

We discuss many more open problems and potential applications of super quantum Airy structures in section \ref{sec-conclusion}, which we believe deserve further investigation.

\subsection{Outline}

The plan of this paper is as follows. In section \ref{sec-sqas} we define super quantum Airy structures, and prove the existence and uniqueness of the corresponding free energies. We also show that these free energies satisfy a recursion relation, which is a supersymmetric version of the bosonic topological recursion, and we provide its graphical interpretation. We also introduce and discuss gauge transformations and the classical limit of super quantum Airy structures. 

In section \ref{sec-finite} we present several finite-dimensional examples of super quantum Airy structures. We also discuss a classification scheme for such finite-dimensional structures. We illustrate this discussion by presenting the $\mathfrak{osp}(1|2)$ example, and conclude this section by constructing super quantum Airy structures from super Frobenius algebras. 

In section \ref{sec:VOSA} we construct examples of infinite-dimensional super quantum Airy structures. The construction follows along the lines of \cite{BBCCN,Milanov}. We construct our examples as untwisted and $\mathbb{Z}_2$-twisted representations for the free boson-fermion vertex operator super algebra  (VOSA). In the process, we also generalize slightly the bosonic construction of \cite{BBCCN}, by considering larger families of subalgebras of the algebra of modes to construct our quantum Airy structures (from the point of view of vertex operator algebras, this should be related to the construction of Whittaker modules for the (super-)Virasoro algebra). As a byproduct, we obtain a realization of the ``topological recursion without branched covers'' of \cite{ABCD} in terms of untwisted representations for the free boson vertex operator algebra. 

We conclude the paper with section \ref{sec-conclusion}, where we list several open problems and potential applications of super quantum Airy structures. Finally, for completeness we provide a computational proof of existence of the free energies associated to a super quantum Airy structure in appendix \ref{proof of existence}.


\subsection*{Acknowledgements}

We thank Ga\"etan Borot, Nitin Chidambaram, Thomas Creutzig, and Motohico Mulase for inspiring discussions. This work was supported by the ERC Starting Grant no. 335739 ``Quantum fields and knot homologies'' funded by the European Research Council under the European Union's Seventh Framework Programme, and the TEAM programme of the Foundation for Polish Science co-financed by the European Union under the European Regional Development Fund (POIR.04.04.00-00-5C55/17-00). V.B. and K.O. acknowledge the support of the Natural Sciences and Engineering Research Council of Canada. The work of KO is also supported in part by the Engineering and Physical Sciences Research Council under grant agreement ref. EP/S003657/1. The work of P.C. is also supported by the NCN Preludium grant no.  
2016/23/N/ST1/01250 ``Quantum curves and Schr\"odinger equations in matrix models'' funded by National Science Centre in Poland. The work of B.R. was supported by the Faculty of Physics, Astronomy and Applied Computer Science grant MSN 2019 (N17/MNS/000040) for young scientists and PhD students.


\section{Super Quantum Airy Structures}     \label{sec-sqas}

In this section we define super quantum Airy structures. We associate to them a unique free energy, which can be calculated recursively. We study their classical limits, and explain how super quantum Airy structures can be obtained as quantizations of super classical Airy structures. We also provide a graphical interpretation for the recursive computation of the free energy associated to a super quantum Airy structure.

\subsection{Definition of Super Quantum Airy Structures}

Quantum Airy structures were introduced in \cite{KS,ABCD} (see also \cite{Notes by Gaetan}) as an abstract framework underlying the Chekhov-Eynard-Orantin topological recursion \cite{CE,EO,EO2}. Just as the Chekhov-Eynard-Orantin topological recursion can be generalized to spectral curves with arbitrary ramification \cite{BHLMR,BE}, quantum Airy structures admit a natural generalization as higher quantum Airy structures: those were studied in \cite{BBCCN}.\footnote{In the nomenclature used in this paper, \emph{higher quantum Airy structures} would be simply \emph{quantum Airy structures}, while the original \emph{quantum Airy structures} would be particular quantum Airy structures that are \emph{quadratic}.} In this section, we propose a further generalization, by introducing fermionic degrees of freedom.

\subsubsection{Background and Notation}

The starting point of quantum Airy structures is a vector space $V$, with dimension either finite or countably infinite\footnote{In fact this restriction is not essential, but it simplifies the notation a bit. It is satisfied in all interesting examples explored until now.}, over ${\mathbb K} = {\mathbb R}$ or ${\mathbb C}$. We introduce fermionic degrees of freedom by considering instead a super vector space $V$: that is, a $\mathbb{Z}_2$-graded vector space $V$. We denote the even and odd subspaces of $V$  by $V_0$ and $V_1$, and elements of $V_0$ and $V_1$ will be called respectively even and odd. We write  $|v| = \alpha$ for the parity of homogeneous elements $v \in V_{\alpha}$, $\alpha=0,1$.

We will only use bases composed of homogeneous elements. For a basis $\{x^i\}_{i \in I}$ in $V$, we denote by $\{ y_i \}_{i \in I}$ the corresponding dual set in $V^* = {\rm Hom}_{\mathbb K}(V,{\mathbb K})$. Here, $I = \{1, 2, \ldots \}$ is a (possibly countably infinite) index set. We abbreviate $|i| = |x^i|$.

\begin{remark}
In order to distinguish the ${\mathbb Z}_2$-grading on $V$ from a $\mathbb Z$-grading of different structures that will appear in what follows, we denote the latter ones with a superscript. For instance, for any $\mathbb Z$-graded superalgebra $A$ we have:
\begin{subequations}
\begin{gather}
A = \bigoplus_{n \in \mathbb Z} A^n, \\
A^n \cdot A^m  \subseteq A^{n+m}.
\end{gather}
\end{subequations}
We also define $A^{\leq n} = \bigoplus_{m \leq n} A^m$ and $A^{\geq n} = \bigoplus_{m \geq n} A^m$. For any $a \in A$ we write $a = \sum_n a^n$ with $a^n \in A^n$. Moreover we set $a^{\leq n} = \sum_{m \leq n} a^m$ and $a^{\geq n} = \sum_{m \geq n} a^m$.
\end{remark}

We define ${\mathbb K}[V, \hbar]$ as the space of all polynomials in $x^i$ and $\hbar$. Similarly, we let $\mathcal W_{\hbar}(V)$ be the space of differential operators which can be written as a sum of finitely many terms of the form
\begin{equation}
\hbar^{m+k} x^{i_1} \ldots x^{i_n} c_{i_1\ldots i_n}^{j_1 \ldots j_m} \partial_{j_1}\ldots\partial_{j_m},
\label{eq:Weyl_terms}
\end{equation}
where for each fixed $j_1,...,j_m$ there exist finitely many $i_1,...,i_n$ such that $c_{i_n...i_1}^{j_m ... j_1} \neq 0$.\footnote{Here and henceforth we use the Einstein summation convention, in which repeated indices are summed over.}  This condition does not depend on the choice of basis. It is equivalent to the statement that the set of coefficients $c$ represents an element of $\mathrm{Hom}(V^{\otimes m}, V^{\otimes n})$.

We introduce a $\mathbb Z$-grading on $\mathcal W_{\hbar}(V)$ by declaring an expression of the form (\ref{eq:Weyl_terms}) to be homogeneous of degree $n+m+2k$. In particular we have
\begin{equation}
\text{deg}(x^i) = 1, \qquad \text{deg}(\hbar \partial_i) = 1, \qquad \text{deg}(\hbar) = 2.\label{degree}
\end{equation}
The same rule defines a grading on the space ${\mathbb K}[V, \hbar]$.

So far we have defined the space $\mathcal W_{\hbar}(V)$ of differential operators acting on ${\mathbb K}[V, \hbar]$. An element of $\mathcal W_{\hbar}(V)$ may be written as a sum of finitely many terms of the form \eqref{eq:Weyl_terms}. We will also need to consider formal series in variables $\hbar, x^i$, and act on these with operators which are infinite sums of terms of the form \eqref{eq:Weyl_terms}. In order to define this precisely, we will introduce a~topology on ${\mathbb K}[V, \hbar]$.

Let $R$ be a ring and $I \subseteq R$ an ideal. We define on $R$ the $I$-adic topology  by declaring the collection $\{ I^n \}_{n=0}^{\infty}$ to be the neighbourhood basis of zero. With this topology $R$ is a topological ring\footnote{This means that $R$ is equipped with a topology such that addition and multiplication, regarded as maps $R \times R \to R$ (with $R \times R$ given the product topology), are continuous.}. It satisfies the Hausdorff axiom if and only if $\bigcap_{n=0}^{\infty} I^n = \{ 0 \}$.

Using this notion we equip ${\mathbb K}[V,\hbar]$ with a ${\mathbb K}[V,\hbar]^{\geq 1}$-adic topology and denote its completion with respect to this topology by ${\mathbb K}[[V, \hbar]]$.

\begin{example}
To see explicitly what this means, let's take the simplest example: $V \cong \mathbb K$. Consider the algebra ${\mathbb K}[V]$ of polynomials in $x$ together with the ideal $I = {\mathbb K}[V]^{\geq 1}$ of polynomials with vanishing constant term. The difference of two polynomials $p,q \in {\mathbb K}[V]$ belongs to $I^{n+1}$ if $p(x)$ and $q(x)$ agree up to the term of the order $x^{n}$. The space ${\mathbb K}[[V]]$
obtained by completing ${\mathbb K}[V]$ with respect to the $I$-adic topology contains, besides polynomials, also formal power series in $x$.
\end{example}

$\mathcal W_{\hbar}(V)$ can now be regarded as a space of operators on ${\mathbb K}[[V, \hbar]]$ and endowed with the topology of pointwise convergence. This means that a generalized sequence $T_{\alpha}$ in $\mathcal W_{\hbar}(V)$ converges to an element $T$ if and only if for every $f \in {\mathbb K}[[V, \hbar]]$ we have $T_{\alpha} f \to Tf$ in $\mathbb K[[V, \hbar]]$. We claim that with this topology, multiplication on $\mathcal W_{\hbar}(V)$ is continuous. Indeed, let $T_{\alpha}$ and $S_{\alpha}$ be two generalized sequences in $\mathcal W_{\hbar}(V)$, with $T_{\alpha} \to T$ and $S_{\alpha} \to S$. Choose $f \in \mathbb K[[V, \hbar]]$. We have
\begin{equation}
(S_{\alpha} T_{\alpha} - ST) f = S_{\alpha}(T - T_{\alpha}) f + (S_{\alpha} - S) Tf.
\end{equation}
By definition, $(S_{\alpha} -S) T f \to 0$. Now given an $n \in \mathbb N$, there exists $\alpha_0$ such that $(T-T_{\alpha})f$ belongs to $\mathbb K[[V, \hbar]]^{\geq n}$ for $\alpha \geq \alpha_0$. Then also $S_{\alpha}(T-T_{\alpha})f \in \mathbb K[[V, \hbar]]^{\geq n}$ for $\alpha \geq \alpha_0$, so $(S_{\alpha} T_{\alpha} - ST) f \to 0$. Since $f$ was arbitrary, we obtain $S_{\alpha} T_{\alpha} \to S T$. Continuity of multiplication on $\mathcal W_{\hbar}(V)$ guarantees that it extends uniquely to a continuous multiplication on the completion, which will be denoted by $\widehat{\mathcal W_{\hbar}}(V)$. Explicitly, one has $TS = \lim_{n \to \infty} T^{\leq n} S^{\leq n} = \sum_{n=0}^{\infty} \sum_{k=0}^n T^k S^{n-k}$ for $T, S \in \widehat{\mathcal W_{\hbar}}(V)$.

\subsubsection{Super Quantum Airy Structures}

We are now ready to define super quantum Airy structures. To this end, let us introduce a little more notation. Let $V$ be a super vector space. As above, we choose a basis $\{x^i \}_{i \in I}$ in $V$, with $I = \{1,2,\ldots \}$, and the corresponding dual set $\{y_i\}_{i \in I}$ in $V^*$. Let $\widetilde V = V \oplus {\mathbb K}^{0|1}$. We let $x^0$ be a~basis for ${\mathbb K}^{0|1}$. In other words, $\widetilde V$ has one more fermionic dimension than $V$.

To clearly distinguish between $V$ and $\widetilde V$, we denote indices which take value in the set $\{0,1,2,\ldots\}$ by small letters from the beginning of the alphabet,
i.e.\ $a,b,c,d = 0,1,2,\ldots$, while the indices denoted by letters from the ``middle'' of the alphabet, $i,j,k,$  take values in the set $I=\{1,2,\ldots\}.$ Consequently,
we denote by $\{x^a\}_{a \geq 0}$ a~basis in $\tilde V = {\mathbb K}^{0|1} \oplus V$.

\begin{definition}\label{d:SQAS}
Let $V$ be a super vector space, and $\widetilde V = V \oplus {\mathbb K}^{0|1}$.  A \emph{super quantum Airy structure} is a pair $(V,L)$, with $L: \ V^* \to \widehat{\mathcal W_{\hbar}}(\widetilde V)$ an even (that is, grade-preserving) continuous linear operator such that:
\begin{enumerate}
\item The left ideal $\mathfrak L \subseteq \widehat{\mathcal W_{\hbar}}(V)$ generated by $L(V^*)$ is involutive, i.e. $[\mathfrak L, \mathfrak L] \subseteq \hbar \cdot \mathfrak L$,
\item $L(y_i)^{\leq 1} = \hbar \partial_i$.
\end{enumerate}
Here $[\cdot, \cdot]$ is the super-commutator, which for two homogeneous elements $v, v' \in V^*$ is given by
\begin{equation}
[L(v), L(v')] = L(v) L(v') - (-1)^{|v| |v'|} L(v) L(v').
\end{equation}
\end{definition}

It is convenient to abbreviate $L_i =  L(y_i)$. The condition (1) guarantees that
\begin{equation}
[L_i, L_j] = \hbar f_{ij}^k L_k
\end{equation}
for some elements $f^k_{ij} \in \widehat{\mathcal W_{\hbar}}(V)$. Continuity of $L$ implies that for fixed $k$ there are only finitely many $i,j$ such that $f_{ij}^k \neq 0$. It also follows from continuity of $L$ that all $L(\phi)$ are determined by $L_i$
\begin{equation}
L(\phi) = \lim_{J \in \mathcal P_{\mathrm{fin}}(I)} \sum_{j \in J} \phi(x^j) L_j,
\label{eq:L_continuity}
\end{equation}
where $\mathcal P_{\mathrm{fin}}(I)$ is the set of all finite subsets of $I$, ordered by inclusion.

\begin{remark}
We note that in the definition of super quantum Airy structures, the domain is the super vector space $V^*$, while the image consists of differential operators acting on the space ${\mathbb K}[[{\widetilde V}, \hbar]]$, where $\widetilde V = V \oplus {\mathbb K}^{0|1}$. In other words, the linear operator $L$ acts on a space with one more fermionic dimension than $V^*$. That is, the operators $\{L_i \}_{i \geq 1}$ can depend on the extra fermionic variable $x^0$, but there is no operator $L_0$.
This possibility is a peculiar feature of super quantum Airy structures, which is not present for traditional quantum Airy structures. It turns out to be crucial in many interesting examples of super quantum Airy structures.
\end{remark}

\begin{remark}
We remark that in the particular case where $L(V^*) \subset  \widehat{\mathcal W_{\hbar}}(V) \subset  \widehat{\mathcal W_{\hbar}}(\widetilde V)$ (i.e. the $L_i$ do not depend on the extra fermionic variable $x^0$ in $\widetilde V$), and $V = V_0$ ($V$ is an even vector space with no fermionic dimension), Definition \ref{d:SQAS} becomes the standard definition of higher Airy structures of \cite{BBCCN}. We will call such super quantum Airy structures \emph{bosonic}.
\end{remark}

The definition \ref{d:SQAS} is rather general, as it allows operators that are infinite sums of terms of the form \eqref{eq:Weyl_terms}, and also differential operators of infinite order. Nonetheless, in practice we will mostly consider finite order differential operators.

\begin{definition}
Let $(V,L)$ be a super quantum Airy structure. If $L(\phi) \subseteq \widehat{ \mathcal W_{\hbar}}(\widetilde V)^{\leq n}$ for some positive integer $n$, we call the smallest such $n$ the \emph{order} of $L$. If $n=2$, we say that the super quantum Airy structure is \emph{quadratic}. If there is no such $n$, we say that the super quantum Airy structure has \emph{infinite order}.
\end{definition}

We will also be interested in the particular case of super quantum Airy structures where the stronger requirement that the $L_i$ themselves span a Lie superalgebra is satisfied, instead of the milder constraint (1) in Definition \ref{d:SQAS}.

\begin{definition}
We say that a super quantum Airy structure $(V,L)$ is \emph{subalgebraic} if the stronger requirement that $[L(V^*), L(V^*)] \subseteq \hbar \cdot  L(V^*)$ is satisfied. In other words, the $L_i$ span a Lie superalgebra.
\end{definition}

Condition $[L(V^*), L(V^*)] \subseteq \hbar \cdot  L(V^*)$ is equivalent to existence of $f_{ij}^k \in \mathbb K$ such that
\begin{equation}
[L_i,L_j] = \hbar f_{ij}^k L_k.
\end{equation}
Continuity of $L$ implies that for fixed $k$ there are finitely many $i, j$, such that $f_{ij}^k \neq 0$. This means that if we endow $V$ with the discrete topology, and $V^*$ with the topology of pointwise convergence, then the expression $[y_i,y_j] = f_{ij}^k y_k$ extends uniquely to a continuous Lie bracket on $V^*$.

In some cases we will also be interested in restricting to super quantum Airy structures such that the $L_i$ can be written as finite sums of terms of the form \eqref{eq:Weyl_terms}. That is, we want to consider the particular case where the $L_i$ are in $\mathcal W_{\hbar}(\widetilde V)$.

\begin{lemma} \label{lem:order}
Let $(V,L)$ be a super quantum Airy structure with $L(V^*) \subseteq \mathcal W_{\hbar}(\widetilde{V})$. Then there exists $n \in \mathbb N$ such that $L(V^*) \subseteq \mathcal W_{\hbar}(\widetilde V)^{\leq n}$ for some $n$.
\end{lemma}
\begin{proof}
Each $L_i$ is an element of $\mathcal W_{\hbar}(\widetilde V)$, so it is actually in $\mathcal W_{\hbar}(\widetilde V)^{\leq n_i}$ for some $n_i$. Suppose that $n_i$ are not uniformly bounded.
Then for $\phi \in V^*$ such that $\phi(x^i)=1$ for all $i$ the right-hand side of (\ref{eq:L_continuity}) is manifestly divergent. Contradiction. Hence there is an $n \in \mathbb N$ such that $L_i \in \mathcal W_{\hbar}(\widetilde V)^{\leq n}$ for each $i$. Since $\mathcal W_{\hbar}(\widetilde V)^{\leq n}$ is closed in $\mathcal W_{\hbar}(\widetilde V)$, the result follows from the formula (\ref{eq:L_continuity}).
\end{proof}

\begin{remark}
With all these definitions, we recover the original definition of quantum Airy structures in \cite{KS,ABCD} as the particular case of a super quantum Airy structure that is bosonic (there is no fermionic variable), quadratic ($L_i  \in \mathcal W_{\hbar}(V)^{\leq 2}$), and subalgebraic ($[L(V^*), L(V^*)] \subseteq \hbar \cdot  L(V^*)$) . Note that we do not require that $V$ is finite-dimensional.
\end{remark}

\subsubsection{Free Energy and Partition Function}
\label{s:free_energy}

Perhaps the most important reason to study quantum Airy structures is that one can assign to them a unique free energy, which for special choice of $L_i$ turns out to be a generating function for some interesting enumerative invariants \cite{KS,ABCD,Notes by Gaetan, BBCCN}. In this section we generalize this construction to super quantum Airy structures.

To every super quantum Airy structure we assign a partition function $Z(x) = {\rm e}^{\frac{F(x)}{\hbar}}$, with free energy $F(x) \in \left( \widetilde V \cdot {\mathbb K}[[\widetilde V, \hbar]] \right)_0^{\geq 3}$, which is defined, loosely speaking, as a solution to the equation
\begin{equation}
L_i \cdot Z(x) =0.
\label{eq:partition_diff_eq}
\end{equation}
The requirement that $F(x) \in \left( \widetilde V \cdot {\mathbb K}[[\widetilde V, \hbar]] \right)_0^{\geq 3}$ means that:
\begin{itemize}
\item $F(x)$ is a formal power series in the variables $x^a$ and $\hbar$ with no term independent of the $x^a$;
\item $F(x)$ is even with respect to the $\mathbb{Z}_2$-grading;
\item $F(x)$ only has terms of degree $\geq 3$ with respect to the $\mathbb{Z}$-grading.
\end{itemize}
Following \cite{KS} we will demonstrate that there is a unique such $Z(x)$.
To this end it is useful to reformulate (\ref{eq:partition_diff_eq}) slightly.

Every $F \in \left( \widetilde V \cdot {\mathbb K}[[\widetilde V, \hbar]] \right)_0^{\geq 3}$ induces a continuous automorphism
\begin{equation}
\widehat{\mathcal W_{\hbar}}(\widetilde V) \ni D \mapsto  \psi_F(D) =\sum_{n=0}^{\infty} \frac{1}{n!} \left( - \frac{1}{\hbar} [F,\,\cdot\,] \right)^n(D) \in \widehat{\mathcal W_{\hbar}}(\widetilde V).
\end{equation}
In particular acting on topological generators\footnote{For a topological ring $R$, subset $S \subseteq R$ is said to be a set of topological generators
if the smallest closed subring of $R$ containing $S$ coincides with $R$.} of $\widehat{\mathcal W_{\hbar}}(\widetilde V)$ one gets:
\begin{subequations}
\begin{gather}
\psi_F(\hbar) = \hbar, \\
\psi_F(x^a) = x^a, \\
\psi_F(\hbar \partial_a) = \hbar \partial_a + \partial_a F.
\end{gather}
\end{subequations}
Since $\left( - \frac{1}{\hbar} [F,\,\cdot\,] \right)^n(D) \in \widehat{\mathcal  W_{\hbar}}(\widetilde V)^{\geq n}$, the series in the definition of $\psi_F(D)$ converges.
Linearity of $\psi_F$ is obvious. Continuity follows from the fact that the subspaces $\widehat{\mathcal W_{\hbar}}(\widetilde V)^{\geq n}$ are $\psi_F$-invariant.
$\psi_F(D_1 D_2)= \psi_F(D_1) \psi_F(D_2)$ is a consequence of $- \frac{1}{\hbar} [F,\,\cdot\,]$ being a derivation and $\psi_F^{-1}$ is given explicitly as $\psi_{-F},$
what proves that $\psi_F$ is indeed an automorphism. We thus say that $F$ is the free energy associated to the super quantum Airy structure $(V,L)$ if
\begin{equation}
\forall \phi \in V^*, \qquad \psi_F(L(\phi)) \cdot 1 =0.\label{psi_F}
\end{equation}

We can now formulate the main result of this section:

\begin{theorem} \label{thm:free_energy}
Every super quantum Airy structure admits a unique free energy.
\end{theorem}
\begin{proof}
Pick some $F \in \left( \widetilde V \cdot {\mathbb K}[[\widetilde V, \hbar]] \right)_0^{\geq 3}$ and let $E(\phi) = \psi_F(L(\phi))) \cdot 1$. Then $E(\phi)^{\leq 1}=0$. We make an inductive hypothesis that $F^{\leq n}$ may be chosen in a unique way so that $E(\phi)^{\leq n-1}=0$. Acting with $\psi_F(L_i) = \hbar \partial_i + \partial_i F + \psi_F(L_i^{\geq 2})$ on $1$ we obtain
\begin{equation}
\left( E(y_i) \right)^n = \partial_i \left( F^{n+1} \right) + H_i^n,
\end{equation}
where $H_i^n$ is a function of $F^{\leq n}$. Continuity of $L$ and $\psi_F$ guarantees that for fixed $n$ there are only finitely many $i$ such that $H_i^n \neq 0$. Evaluating $\psi_F \left([L_i,L_j]- \hbar f_{ij}^k L_k \right) \cdot 1$ we get
\begin{equation}
\partial_i H_j^n - (-1)^{|i||j|} \partial_j H_i^n=0,
\end{equation}
so the equation $\left( E(y_i) \right)^n=0$ can be solved for $F^{n+1}$ as
\begin{equation}
F^{n+1} = - (1+x^i \partial_i )^{-1} x^j H_j^n,
\label{eq:free_energy_solution}
\end{equation}
up to the addition of an arbitrary integration constant in the extra fermionic variable $x^0$ in $\widetilde V$. But $F$ is required to be even, and hence this integration constant must vanish. Therefore, the solution is unique, and hence $F^{n+1}$ is uniquely determined by the condition $E(\phi)^{\leq n}=0$.
\end{proof}

\begin{remark}
One could ask whether we cannot enlarge $\widetilde V$ further. For instance, one could consider $L: \ V^* \to \widehat{ \mathcal W_{\hbar}}(V \oplus X)$ and $F \in {\mathbb K}[[V \oplus X ,\hbar]]$, for more general $X$. In general, the proof of Theorem \ref{thm:free_energy} would then guarantee existence of the free energy, but not uniqueness. Existence and uniqueness are obtained only when $\dim X = 0|1$, since in this case the requirement that $F$ is even is sufficient to guarantee that $\partial_i F^{n+1}+H_i^n=0$ has a unique solution.

In particular, as is clear from this argument, there is no such freedom of enlarging $V$ for bosonic quantum Airy structures. The possibility of having an extra coordinate $x^0$ is a purely fermionic phenomenon.
\end{remark}

\subsubsection{A Little More Structure}

\label{s:ringR}

In the previous section we showed that we can associate a unique free energy $F \in \left( \widetilde V \cdot {\mathbb K}[[\widetilde V, \hbar]] \right)_0^{\geq 3}$ to every super quantum Airy structure, by requiring that $\psi_F(L(\phi)) \cdot 1 =0 \label{psi_F}$ for all $ \phi \in V^*$. However, we started the section by saying that the free energy was defined such that the partition function $Z(x) = {\rm e}^{\frac{F(x)}{\hbar}}$ is a solution to the system of equations $L_i \cdot Z(x) = 0$. Let us now explore the connection between the two statements more precisely. To this end, we now define a convenient ring of series.

\begin{definition}
Let $R$ be the super ${\mathbb K}$-vector space of formal series of the form
\begin{equation}
f = \sum_{a = - \infty}^{\infty} \sum_{b=0}^{\infty} \hbar^a f_{a,b}(x),
\end{equation}
where $f_{a,b}(x)$ is a polynomial of degree $b$ in the variables $x^i$ and $f_{a,b}=0$ if $3a+b<0$. Note that we allow both positive and negative powers of $\hbar$ here. Each term $\hbar^a f_{a,b}(x)$ is declared to be homogeneous of degree $2a+b$, as is consistent with our $\mathbb{Z}$-grading.\footnote{This is not a typo: we use both combinations $2a+b$ and $3a+b$ in our considerations.}
\end{definition}

It is easy to check that the condition $3a+b \geq 0$ combined with $b \geq 0$ entail that $2a +b \geq 0$, so the degree of each term is always non-negative. For fixed $d=2a+b$ we have inequalities $2a \leq d$ and $a \geq -d$, so $f$ may be rewritten as
\begin{equation}
f = \sum_{d=0}^{\infty} \sum_{a = -d}^{2a \leq d} \hbar^a f_{a,d-2a}(x) = \sum_{d=0}^{\infty} f^d,
\end{equation}
where we introduced the homogeneous components $f^d$ of $f$. Notice that each $f^d$ is a polynomial in $x, \hbar$ and $\hbar^{-1}$ with $f^0$ being a constant. Therefore:

\begin{lemma}
$R$ is a supercommutative ${\mathbb K}$-algebra, with a product given by
\begin{equation}
fg = \sum_{d=0}^{\infty} \sum_{k=0}^d f^k g^{d-k}.
\end{equation}
\end{lemma}

\begin{remark}
Since we are considering series which include both positive and negative powers of $\hbar,$ the condition $f_{a,b}=0$ for $3a+b <0$ is needed to make the  multiplication in $R$ well defined.
Let us remark that a weaker condition $f_{a,b}=0$ for $2a +b <0$ would also work here, but it is insufficient to make sense of gauge transformations which will be discussed in Section \ref{ss:cl:q}.
\end{remark}

We can go further:

\begin{lemma}
$R$ is a local superring with maximal ideal $\mathfrak m = \{ f \in R | \ f^0=0 \}$.
\end{lemma}

\begin{proof}
Since the degree is non-negative, elements of $\mathfrak m$ cannot be invertible.
On the other hand, for any $\epsilon \in \mathfrak m$ and $u \in {\mathbb K}^{\times}$, the inverse of the element $u+ \epsilon$ is given explicitly by
\begin{equation}
\frac{1}{u+ \epsilon} = u^{-1} \sum_{k=0}^{\infty} (- u^{-1}\epsilon)^k.
\end{equation}
The infinite sum on the right hand side is well-defined because only the first $d$ terms contribute to the homogeneous component of degree $d$. The identity $(u+ \epsilon) \cdot u^{-1}  \sum_{k=0}^{\infty} (- u^{-1} \epsilon)^k =1$ is then quite obvious, and we conclude that all elements of $R \setminus \mathfrak m$ are invertible.
\end{proof}

It is easy to see that $R$ equipped with the $\mathfrak m$-adic topology is a complete Hausdorff space. In particular, for any $f \in \mathfrak m$ and $g \in {\mathbb K}[[t]]$ we have an element $g(f) \in R$.
The most important for us (perhaps except for polynomials) examples of this construction are
\begin{subequations}
\begin{gather}
\exp (f ) = \sum_{k=0}^{\infty} \frac{f^k}{k!}, \\
\log (1+f) = - \sum_{k=1}^{\infty} \frac{(-f)^k}{k}.
\end{gather}
\end{subequations}
We note that $\log ( \exp (f) ) = f$ and $\exp (\log(1+f))=1+f$.

Since elements of each degree in $R$ are polynomials, we can act on them by elements of the completed Weyl algebra $\widehat{\mathcal W_{\hbar}}(\widetilde V)$. After projecting to homogeneous terms, this reduces to the computation of a finite sums of terms in which a differential operator of finite degree acts on a polynomial. We conclude that:
\begin{lemma}
 $R$, equipped with the $\mathfrak m$-adic topology,  is a topological $\widehat{\mathcal W_{\hbar}}(\widetilde V)$-module.
 \end{lemma}

Now let us go back to the free energy and partition function associated to a super quantum Airy structure. Following \cite{KS}, we defined the free energy associated to a super quantum Airy structure as an automorphism $\psi_F=\exp \left( \frac{1}{\hbar} [F,\,\cdot\,] \right)$ with $F \in \left( \widetilde V \cdot {\mathbb K}[[\widetilde V, \hbar]] \right)_0^{\geq 3}$, and such that for all $i$ the operator $\psi_F(L_i)$ annihilates $1.$
Observe now that $F$ (and hence also $Z  = {\rm e}^{\frac{F}{\hbar}}$) is a well-defined element of the ring $R$. Moreover $Z^0=1$, so $Z$ is invertible. Simple manipulation with the involved series shows that the identity\footnote{After projecting to terms of given total degree we always get finite sums, so this is a purely combinatorial problem. It boils down to identity $[A,\,\cdot\,]^n(B) = \sum_{k=0}^n \binom{n}{k} A^k B (-A)^{n-k} $ for even $A$.}
\begin{equation}
Z^{-1} \left( L_i \cdot Z \right) = \psi_F(L_i) \cdot 1
\end{equation}
indeed holds, as expected.
Therefore $Z \in R$ is annihilated by all $L_i$, as claimed originally.

An explicit formula for the partition function in terms of the free energy reads
\begin{equation}
Z = 1 + \sum_{k=1}^{\infty} \sum_{a=-k}^{2a \leq k-1} \hbar^a \sum_{\{ (g_{\alpha}, n_{\alpha}) \}_{\alpha \in A}} \prod_{\alpha \in A} F_{g_{\alpha},n_{\alpha}},
\end{equation}
where the last sum is taken over all finite sets of pairs $(g,n) \in \mathbb N^2$ such that $n \geq 1$, $\sum_{\alpha \in A} (g_{\alpha}-1)=a$, and $\sum_{\alpha \in A} [2(g_{\alpha}-1)+n_{\alpha}]=k$. This sum is always finite.

Two remarks are in order.
\begin{remark}
In the definition of the ring $R$, we allow both positive and negative powers of $\hbar$. One may then revisit the proof of existence and uniqueness of the free energy (Theorem \ref{thm:free_energy}) with slightly more general assumptions. Instead of requiring from the start that $F  \in \left( \widetilde V \cdot {\mathbb K}[[\widetilde V, \hbar]] \right)_0^{\geq 3}$, one could consider $F  \in R^{\geq 3}$, with the requirement that $F=0$ for $x=0$. The difference here is that we allow terms with negative powers of $\hbar$, as long as they are accompanied with sufficiently many powers of $x^a$ so that the degree of each term is $\geq 3$. Then, following the same steps as in the proof of Theorem \ref{thm:free_energy}, one sees that existence and uniqueness of the free energy associated to a super quantum Airy structure is still true, and hence it must belong to $\left( \widetilde V \cdot {\mathbb K}[[\widetilde V, \hbar]] \right)_0^{\geq 3}$. In other words, the lack of negative powers of $\hbar$ in $F$ is a result, rather than an assumption.
\end{remark}

\begin{remark}
One may ask why we did not prove existence and uniqueness of the partition function $Z$ directly by solving the system of equations $L_i Z = 0$ subject to the condition that $Z$ evaluated at $x=0$ is equal to $1$. This is certainly possible: the proof technique is exactly the same. Then one may define $F = \hbar \log Z$. In this approach however it is not clear to us how to prove directly that the free energy does not contain negative powers of $\hbar$. (Of course, this must still be true, since it is the same unique free energy as the one obtained above.)
\end{remark}


\subsection{Recursive System}\label{sec:computational}

For simplicity, in this section we focus on super quantum Airy structures that are quadratic and subalgebraic, and thus can be regarded as supersymmetric analogs of the original quantum Airy structures of \cite{KS,ABCD}. In the spirit of \cite{ABCD}, we derive explicit conditions for the coefficients of the operators $L_i$ such that the $L_i$ form a super quantum Airy structure. We also compute a recursive system for the coefficients of the free energy uniquely associated to a super quantum Airy structure.

\subsubsection{Constraints on the Coefficients}

As in the previous section, we choose a basis $\{x^i\}_{i \geq 1}$ for the super vector space $V$ and denote by $x^0$ a basis vector for ${\mathbb K}^{0|1}.$ We let $\widetilde V = {\mathbb K}^{0|1} \oplus V $, and denote by $\{x^a\}_{a \geq 0}$ a basis in $\tilde V$. We denote indices which take value in the set $\{0,1,2,\ldots\}$ by small letters from the beginning of the alphabet,
i.e.\ $a,b,c,d = 0,1,2,\ldots$, while indices denoted by the letters from the ``middle'' of the alphabet, $i,j,k,$  take values in the set $\{1,2,\ldots\}.$

Super quantum Airy structures were introduced in Definition \ref{d:SQAS}. To construct a quadratic super quantum Airy structure, we need to find an even continuous linear operator $L:V^*\rightarrow \widehat{\mathcal{W}_{\hbar}}(\widetilde{V})^{\leq 2}$  such that:
\begin{itemize}
\item for each $i \geq 1$ we have $L_i = L(y_i) = \hbar \partial_i + L_i^2$ where $L_i^2 \in  \widehat{\mathcal{W}_{\hbar}}(\widetilde{V})^2$;
\item $[L_i, L_j] = \hbar f_{ij}^k L_k$ for some structure constants $f_{ij}^k \in \mathbb{K}$.
\end{itemize}
Explicitly, we can write
\begin{align}
\label{quadratic:Ls}
L_i
&= \hbar\partial_i - \frac12 A_{iab}x^ax^b - \hbar B_{ia}^b x^a\partial_b - \frac12\hbar^2 C_i^{ab}\partial_a\partial_b - \hbar D_i \\
&=: L_i^{\leq 1} + L_i^{A} + L_i^{B} + L_i^{C} + L_i^{D},
\end{align}
with the coefficients $A_{iab}, B^b_{ia}, C_i^{ab}, D_i \in \mathbb{K}$.
As we have already remarked, continuity of $L$ imposes that,
if $\widetilde V$ is infinite-dimensional, for fixed $i$ only finitely many $A_{iab}$ are non-zero, and for fixed $i$ and $b$ only finitely many $B_{ia}^b$ are non vanishing.
Clearly, we may assume the symmetry conditions
\begin{equation}
A_{iab} = (-1)^{|a||b|} A_{iba}, \hskip 10 mm C_i^{ab} = (-1)^{|a| |b|} C_i^{ba}.
\label{eq:tensor_symmetries}
\end{equation}

Since $L:V^*\rightarrow \widehat{\mathcal{W}_{\hbar}}(\widetilde{V})^{\leq 2}$ is assumed to be even, we can think of the coefficients $A_{iab}, B^b_{ia}, C_i^{ab}$, and  $D_i $ as components of even tensors
\begin{equation}
A\in V \otimes  \widetilde V \otimes \widetilde{V}, \qquad B\in\text{Hom}(\widetilde{V},V \otimes \widetilde{V}),\qquad
C\in\text{Hom}(\widetilde{V}\otimes\widetilde{V},V),\qquad
D\in V.\label{tensorABCD}
\end{equation}

In the spirit of \cite{ABCD}, we can reformulate the Lie superalgebra requirement $[L_i, L_j ] = \hbar f_{ij}^k L_k$, as a~set of constraints on the tensors $A,B,C,D$. This is to be compared with Lemma 2.2 of \cite{ABCD}.

\begin{lemma}\label{lem:SAS}
The differential operators $L_i$ in \eqref{quadratic:Ls} form a super quantum Airy structure if and only if the following conditions are satisfied
\begin{subequations}
\begin{align}
A_{jia}=&(-1)^{|i||j|}A_{ija},\label{A} \\
f_{ij}^k=&(-1)^{|i||j|}B_{ij}^k-B_{ji}^k,\\
0=&(-1)^{|i||j|}B_{ij}^0-B_{ji}^0,\label{f}
\end{align}
\end{subequations}
and
\begin{subequations}
\begin{align}
B_{ia}^cA_{jcb}+(-1)^{|a||b|}B_{ib}^cA_{jca}+(-1)^{|i||j|}B_{ij}^kA_{kab}&=(-1)^{|i||j|}(i\leftrightarrow j),\label{BA}\\
B_{ia}^cB_{jc}^b+(-1)^{|a||b|}C_i^{bc}A_{jca}+(-1)^{|i||j|}B_{ij}^kB_{ka}^b&=(-1)^{|i||j|}(i\leftrightarrow j),\label{BB-CA}\\
C_i^{ac}B_{jc}^b+(-1)^{|a||b|}C_i^{bc}B_{jc}^a+(-1)^{|i||j|}B_{ij}^kC_k^{ab}&=(-1)^{|i||j|}(i\leftrightarrow j),\label{CB}\\
\frac{1}{2}C_i^{ba}A_{jab}+(-1)^{|i||j|}B_{ij}^kD_k&=(-1)^{|i||j|}(i\leftrightarrow j).\label{CA-BD}
\end{align}
\end{subequations}
\end{lemma}
\begin{proof}
These conditions are very similar to those of Lemma 2.2 in \cite{ABCD}, with appropriate signs, and range of indices to take into account the extra fermionic variable. The proof is also a straightforward computation. We simply expand the super-commutator $[L_i, L_j]$ and collect terms with respect to $x^a$ and $\partial_a$. Then by comparing with $\hbar f_{ij}^k L_k$, we obtain the set of constraints.
\end{proof}

\subsubsection{Topological Recursion}

In Section \ref{s:free_energy} we associated a unique free energy $F  \in \left( \widetilde V \cdot {\mathbb K}[[\widetilde V, \hbar]] \right)_0^{\geq 3}$ to every super quantum Airy structure. In this section, we show how its coefficients can be calculated recursively. In the spirit of \cite{ABCD}, and for the sake of completeness, we also show how existence and uniqueness of the free energy can be proven computationally from the recursive structure.

\begin{theorem}\label{thm:SAS}
Let $(V,L)$ be a quadratic super quantum Airy structure, and let $F \in \left( \widetilde V \cdot {\mathbb K}[[\widetilde V, \hbar]] \right)_0^{\geq 3}$ be its associated free energy. We can expand $F$ in the basis $\{ x^a \}_{a \geq 0}$ for $\widetilde V$ as:
\begin{equation}
F=\sum_{g\geq0}\sum_{n\geq1}\sum_{a_1,\cdots,a_n \geq 0}\frac{\hbar^g}{n!}F_{g,n}[a_1,\cdots,a_n]x^{a_1}\cdots x^{a_n},\label{def:F}
\end{equation}
where the coefficients $F_{g,n}[a_1,\cdots,a_n] \in \mathbb{K}$ are $\mathbb{Z}_2$-symmetric (with signs) under permutations of the indices $\{a_1, \ldots, a_n\}$.
Then the coefficients $F_{g,n}[a_1,\cdots,a_n]$ satisfy the recursive system:
\begin{align}
F_{g,n+1}[i,\Phi]=&\;A_{ia_1a_2}\delta_{g,0}\delta_{n,2}+ D_i\delta_{n,0}\delta_{g,1}\nonumber\\
&+\sum_{k=1}^n\sigma_{a_k\subset\Phi}\sum_{b\geq 0}B_{ia_k}^bF_{g,n}[b,\Phi\backslash a_k]+\frac{1}{2}\sum_{b,c\geq 0}C_i^{bc}F_{g-1,n+2}[c,b,\Phi]\nonumber\\
&+\frac{1}{2}\sum_{b,c\geq }\sum_{g_1+g_2=g}\sum_{\Phi_1\cup \Phi_2=\Phi}\sigma_{\Phi_1\subset\Phi}C_i^{bc}F_{g_1,n_1+1}(b,\Phi_1)F_{g_2,n_2+1}[c,\Phi_2],\label{F(i)}
\end{align}
with the auxiliary equation
\begin{equation}
F_{g,n+1}[0,a_1,a_2, \ldots, a_n]=(-1)^{|a_1|}F_{g,n+1}[a_1,0,a_2, \ldots, a_n]\label{F(0)}.
\end{equation}
Here, $\Phi=\{a_1,\cdots,a_n\}$ is an ordered set, and $\sigma_{\Phi_1\subset \Phi}$ denotes the sign of the permutation from $\Phi$ to $\Phi_1 \cup (\Phi\backslash \Phi_1)$.
\end{theorem}
We remark that the recursive formula makes sense for infinite-dimensional $V$, since by induction one can show that for any $g$ and $n$ only finitely many $F_{g,n}[a_1, \ldots, a_n]$ are non-vanishing, and hence the sums on the right-hand-side are all finite.

\begin{proof}
By Theorem \ref{thm:free_energy}, there exists a unique $F \in \left( \widetilde V \cdot {\mathbb K}[[\widetilde V, \hbar]] \right)_0^{\geq 3}$ such that $\psi_F(L_i)\cdot1=0$. We~now show that this implies the recursive system \eqref{F(i)} and the auxiliary equation \eqref{F(0)}.

To derive \eqref{F(i)} from $\psi_F(L_i)\cdot1=0$, we consider for $n\geq0$
\begin{equation}
\partial_{a_n}\cdots\partial_{a_1}\cdot\Bigl(\psi_F(L_i)\cdot1\Bigr)\Bigr|_{x=0}=0.
\end{equation}
Note that the order of the derivatives is important to have the correct sign. As $\psi_F(L_i)$ is linear, we list the computational results for each term in $L_i=L_i^{\leq 1} + L_i^{A} + L_i^{B} + L_i^{C} + L_i^{D}$ for completeness:
\begin{align}
\partial_{a_n}\cdots\partial_{a_1}\cdot\Bigl(\psi_F\left(L_i^{\leq1}\right)\cdot1\Bigr)\Bigr|_{x=0}&=\sum_{g\geq0}\hbar^{g}F_{g,n+1}[i,a_1,\cdots,a_n]\nonumber\\
\partial_{a_n}\cdots\partial_{a_1}\cdot\left(\psi_F\left(L_i^{(A)}\right)\cdot1\right)\Bigr|_{x=0}&=-\sum_{g\geq0}\hbar^g\left(A_{ia_1a_2}\delta_{n,2}\delta_{g,0}\right)\nonumber\\
\partial_{a_n}\cdots\partial_{a_1}\cdot\Bigl(\psi_F\left(L_i^{(B)}\right)\cdot1\Bigr)\Bigr|_{x=0}&=-\sum_{g\geq0}\hbar^g\sum_{k=1}^n\sum_{b\geq 0}\sigma_{a_k\subset \Phi}B_{ia_k}^bF_{g,n}[b,\Phi\backslash a_k]\nonumber\\
\partial_{a_n}\cdots\partial_{a_1}\cdot\left(\psi_F\left(L_i^{(C)}\right)\cdot1\right)\Bigr|_{x=0}&=-\sum_{g\geq0}\frac{\hbar^g}{2}\sum_{p,q \geq 0}C_i^{bc}\Biggl(F_{g-1,n+2}[c,b,\Phi]\nonumber\\
&\;\;\;\;+\sum_{g_1+g_2=g}\sum_{\Phi_1\cup \Phi_2=\Phi}\sigma_{\Phi_1\subset \Phi}F_{g_1,n_1+1}[b,\Phi_1]F_{g_2,n_2+1}[c,\Phi_2]\Biggr)\nonumber\\
\partial_{a_n}\cdots\partial_{a_1}\cdot\Bigl(\psi_F\left(L_i^{(D)}\right)\cdot1\Bigr)\Bigr|_{x=0}&=-\sum_{g\geq0}\hbar^g\left(D_i\delta_{n,0}\delta_{g,1}\right).
\end{align}
Collecting terms order by order in $\hbar$, we obtain \eqref{F(i)}.

As for the auxiliary equation, it is necessary because $F_{g,n+1}[0, a_1, a_2, \ldots, a_n]$ is not fixed by the recursive system, since there is no $L_0$. However, $F_{g,n+1}[a_1, 0, a_2, \ldots, a_n]$ is fixed, and hence $F_{g,n+1}[0, a_1, a_2, \ldots, a_n]$ is uniquely fixed by symmetry as in \eqref{F(0)}.

\end{proof}

\begin{remark}
We can in fact prove existence and uniqueness of the free energy from this computational point of view, in the spirit of \cite{ABCD}. This provides an alternative proof of Theorem \ref{thm:free_energy}. The~proof proceeds in three steps:
\begin{enumerate}
\item
We first show that $\psi_F(L_i)\cdot1=0$ implies the recursive system \eqref{F(i)} and \eqref{F(0)}, as in Theorem \ref{thm:SAS}.
\item Assuming existence of a free energy $F\in \left( \widetilde V \cdot {\mathbb K}[[\widetilde V, \hbar]] \right)_0^{\geq 3}$, with expansion given by \eqref{def:F}, we show that \eqref{F(i)} and \eqref{F(0)} uniquely reconstructs it. This is clear, since \eqref{F(i)}  is a~recursive system on $2g+n$ that reconstructs (in conjunction with \eqref{F(0)}) all coefficients $F_{g,n}[a_1,\ldots,a_n]$ from the initial conditions
\begin{equation}
F_{0,3}(i,a,b)=A_{iab},\;\;\;\;F_{0,3}(0,i,a)=(-1)^{|i|}A_{i0a},\;\;\;\;F_{1,1}(i)=D_i,
\label{eq:ic}
\end{equation}
where we used \eqref{F(0)} for the second condition. 
\item What remains to be proved is that the free energy $F$ actually exists, which is the difficult part. As in \cite{ABCD}, the idea is to start from the recursive system \eqref{F(i)} and \eqref{F(0)}, and show that, while it is not manifestly symmetric, the coefficients $F_{g,n}[a_1, \ldots, a_n]$ that it constructs indeed are $\mathbb{Z}_2$-symmetric. Therefore, the recursive system does reconstruct a free energy $F\in \left( \widetilde V \cdot {\mathbb K}[[\widetilde V, \hbar]] \right)_0^{\geq 3}$ through its expansion \eqref{def:F}. By (1), it is a solution to $\psi_F(L_i)\cdot1=0$, and by (2), it is unique, and hence we have proven existence and uniqueness. The essence of the proof of existence thus consists in showing that the recursive system \eqref{F(i)} and \eqref{F(0)} reconstructs $\mathbb{Z}_2$-symmetric coefficients. An interesting aspect of this computational proof is that it highlights the importance of the Lie superalgebra requirement. For completeness, we present this computational proof of existence in Appendix~\ref{proof of existence}.
\end{enumerate}

\end{remark}

\subsubsection{Graphical Interpretation}\label{sec:graph}

The recursive system \eqref{F(i)} has a nice interpretation in terms of sums of trivalent graphs, in parallel to the graphical interpretation presented for bosonic quantum Airy structures in \cite{ABCD,Notes by Gaetan}. For completeness, we present this graphical interpretation in this section, focusing on quadratic super quantum Airy structures without the extra fermionic variable. We follow very closely the presentation in \cite{Notes by Gaetan}. It remains to be seen whether the graphical interpretation can be extended to super quantum Airy structures with an extra fermionic variable.

The graphical interpretation is very similar to the one presented in Section 1.3 of \cite{Notes by Gaetan}. In fact, the graphs are the same, but we need to change slightly the assignment of weights to the graphs to take into account the signs arising from $\mathbb{Z}_2$-symmetry.

Let us first define the set $\mathbb{G}_{g,n+1}$, which is the same as Definition 1.4 of \cite{Notes by Gaetan}.\begin{definition}\label{def:sgraph} [Definition 1.4 of \cite{Notes by Gaetan}]
For $g\geq0$ and $n \geq 0$ such that $\chi_{g,n}:=2g+n-2\geq0$, we form the set $\mathbb{G}_{g,n+1}$ consisting of pairs $\Gamma=(G,T)$ where:
\begin{itemize}
\item $G$ is a connected trivalent graph with $2g-1+n$ trivalent vertices, $n+1$ ordered one-valent leaves, and first Betti number $b_1(G) = g$. We single out a leave and call it the root of $G$. 
\item $T \subseteq G$ is a spanning tree that includes the root of $G$, but none of the leaves.  
\item The edges $e = \{v,v'\}$ of $G$ which are not in the spanning tree $T$ connect parent vertices, i.e. the common ancestor of $v$ and $v'$ in the rooted spanning tree $T$ is either $v$ or $v'$.
\end{itemize}
We denote the ordered leaves by $\ell, \ell_1, \ell_2, \ldots, \ell_n$ in counterclockwise order\footnote{In \cite{Notes by Gaetan} the counterclockwise requirement is not specified, because the assigned weights are symmetric, whereas in our case we need to be a little more careful because the weights are only $\mathbb{Z}_2$-symmetric.}, with $\ell$ being the root. We denote by $E'(\Gamma)$ the set consisting of leaves (including the root) of $G$ and edges of $G$ that are not loops.
\end{definition}

\begin{definition}
An \emph{automorphism} of $\Gamma = (G,T) \in \mathbb{G}_{g,n+1}$ is a permutation of the edges in $G$ that preserves the graph structure. We denote by $\text{Aut}(\Gamma)$ the set of automorphisms of $\Gamma=(G,T)$.
\end{definition}

By convention, we set $\mathbb{G}_{0,1} =\mathbb{G}_{0,2} = \emptyset$. $\mathbb{G}_{0,3}$ and $\mathbb{G}_{1,1}$ both contain only one element, which are shown in Section 1.3 of \cite{Notes by Gaetan}.

Furthermore, as explained in \cite{Notes by Gaetan}, $\mathbb{G}_{g,n+1}$ has a recursive structure on $\chi_{g,n}$. If we remove from a given graph $\Gamma\in\mathbb{G}_{g,n+1}$ the vertex incident to the root $\ell$ with two more edges/leaves $\{e_1,e_2\}$, the resulting graph falls into one of the following three cases:
\begin{itemize}
\item[\textbf{I}] a graph $\Gamma'\in\mathbb{G}_{g,n}$, if one of $e_1$ or $e_2$ (we call it $e_2$ without loss of generality) is a leaf of $\Gamma$. We let $e_1$ be the root of $\Gamma'$.
\item[\textbf{I$'$}] a graph $\Gamma'\in\mathbb{G}_{g-1,n+1}$. We let $e_1$ be the root of $\Gamma'$, and $e_2$ its first leaf. Note that we need to specify which of $e_1$ and $e_2$ is the root and the first leaf here, unlike \cite{Notes by Gaetan}, because of $\mathbb{Z}_2$-symmetry.
\item[\textbf{II}] a non-ordered disjoint union of $\Gamma'_1\cup\Gamma'_2$ where $\Gamma'_i\in\mathbb{G}_{g_i,|\Phi_i|+1}$ is a graph with root $e_i$ and leaves $L_i$ such that $g_1+g_2=g$ and $L_1\cup L_2$ are leaves of $\Gamma$ distinct from the root $\ell$.
\end{itemize}

We would like to assign a weight to each graph $\Gamma = (G,T) \in  \mathbb{G}_{g,n+1}$. To do this, we need to equip $\Gamma$ with a colouring.

\begin{definition}
A \emph{colouring} of $\Gamma = (G,T) \in \mathbb{G}_{g,n+1}$ is a map $\gamma: E'(\Gamma) \to I$, where $I$ is the index set $I = \{1,2,3,\ldots \}$. In other words, it assigns a positive integer to all leaves of $G$ and edges of $G$ that are not loops.
\end{definition}

To each graph $\Gamma\in\mathbb{G}_{g,n+1}$ with colouring $\gamma$, with $\chi_{g,n} \geq0$, we assign a weight $\omega(\Gamma,\gamma)$ as follows, using the recursive decomposition (\textbf{I}, \textbf{I}$'$, \textbf{II}). First of all, we define the cases with $\chi_{g,n}=0$ as
\begin{equation}
\omega(\Gamma_{0,3},\gamma)=A_{\gamma(l_1)\gamma(l_2)\gamma(l_3)},\;\;\;\;\omega(\Gamma_{1,1},\gamma)=D_{\gamma(l_1)}.
\end{equation}
For graphs with $\chi_{g,n}\geq1$, we recursively determine the weight in terms of the decomposition (\textbf{I}, \textbf{I}$'$, \textbf{II}):
\begin{align}
 \textbf{I}\;\;&\;\;\omega(\Gamma,\gamma)=\sigma_{\gamma(e_1)\subset\gamma(L)}\,B_{\gamma(\ell_1)\gamma(e_1)}^{\gamma(e_2)}\,\omega(\Gamma',\gamma'),\\
 \textbf{I}'\;\;&\;\;\omega(\Gamma,\gamma)=C_{\gamma(\ell_1)}^{\gamma(e_2)\gamma(e_1)}\;\omega(\Gamma',\gamma'),\\
 \textbf{II}\;\;&\;\;\omega(\Gamma,\gamma)=\sigma_{\gamma(L_1)\subset\gamma(L)}C_{\gamma(\ell_1)}^{\gamma(e_1)\gamma(e_2)}\;\omega(\Gamma_1,\gamma_1)\,\omega(\Gamma_2,\gamma_2),
\end{align}
where $L=(\ell_1,...,\ell_n)$ and $\gamma(L)=(\gamma(\ell_1),...,\gamma(\ell_n))$ is the colouring of leaves. The sign factor $\sigma$ is defined in Theorem~\ref{thm:SAS}. $\gamma'$ denotes the colouring for the corresponding decomposed graph, that is, the restriction of $\gamma$ to $E'(\Gamma')$. $L_i,\gamma_i$ are similarly defined. Let us emphasize that the order of edges/leaves is important, because of the sign factors. 

Since this weight assignment precisely captures the recursive equation \eqref{F(i)}, we obtain the following lemma:

\begin{lemma}
For any $g,n\geq0$, and $i,i_1,...,i_n\in I$, we have
\begin{equation}
F_{g,n+1}[i,i_1,...,i_n]=\sum_{\Gamma\in\mathbb{G}_{g,n+1}}\sum_{\substack{\gamma\in I^{E'(\gamma)}\\\gamma(\ell_j)=i_j}}\frac{\omega(\Gamma,\gamma)}{|{\rm Aut}(\Gamma)|}.
\end{equation}
\end{lemma}

We refer the reader to \cite{Notes by Gaetan} and \cite{ABCD} for pretty pictures of the graphical representation of topological recursion.


\subsection{Gauge Transformations}

In this section we study gauge transformations of super quantum Airy structures.

\label{ss:cl:q}

Recall that $\widetilde V = V \oplus \mathbb{K}^{0|1}$. We consider the symplectic super vector space $W=\widetilde V \oplus \widetilde V^*$ equipped with the product topology and the symplectic form given by $\omega(\phi,v) = \phi(v)$ for $\phi \in \widetilde V^*$ and $v \in \widetilde V$.

Topology on $W$ was chosen in such a way that the following properties hold:
\begin{itemize}
\item if $\{ a_n \}_{n=0}^{\infty}$ is a sequence of distinct indices, $\displaystyle \lim_{n \to \infty} y_{a_n} =0$,
\item if $\{ v_n \}_{n=0}^{\infty} \subseteq V$ is a convergent sequence, then $v_n$ are eventually constant.
\end{itemize}

\subsubsection{Linear Gauge Transformations}

With this under our belt, we can define gauge transformations for super quantum Airy structures. We proceed as follows. Let $T$ be a linear operator $W \to W$ represented by matrices
\begin{subequations}\label{eq:W_matrix}
\begin{align}
T x^a =\; &  t^{a}_b x^b + b^{ab} y_b, \\
T y_a =\; &  c_{ab} x^b + d^b_a y_b.
\end{align}
\end{subequations}
Then for each fixed $a$ there are finitely many $b$ such that $t^a_b \neq 0$ or $c_{ab} \neq 0$. Moreover $T$ is continuous if and only if the following conditions are satisfied:
\begin{itemize}
\item there are finitely many pairs $a,b$ such that\footnote{In other words $c_{ab} x^a \otimes x^b \in \widetilde V \otimes  \widetilde V$.} $c_{ab} \neq 0$,
\item for each fixed $b$ there are finitely many $a$ such that\footnote{This means that $d^b_a$ represents the transpose of some operator $V \to V$.} $d^b_a \neq 0$.
\end{itemize}

If $T: \ W \to W$ is an even, linear and continuous symplectomorphism ($T \in \mathrm{Aut(W)}$ for short), there exists a unique continuous automorphism $\widetilde T$ of $\mathcal{W}_{\hbar}(\widetilde V)$ which acts on generators as in formulas (\ref{eq:W_matrix}) with $y_a$ replaced by $\hbar \partial_a$ and $\widetilde T(\hbar)=\hbar$. It may be extended uniquely to a continuous automorphism of $\widehat{\mathcal W_{\hbar}}(\widetilde V)$. Indeed, it is obvious that proposed transformation preserves algebraic relations in $\mathcal W_{\hbar}(\widetilde V)$. By the preceding discussion, also finiteness conditions in the definition of $\mathcal W_{\hbar}(\widetilde V)$ are preserved. Continuity of $\widetilde T$ is easy to see.

Now consider an even, linear and continuous symplectomorphism $T_s \in \mathrm{Aut}(W)$ that acts trivially on $\widetilde V^*$. In terms of generators, it takes the form
\begin{subequations}
\begin{gather}
T_s (x^a) = x^a + s^{ab} y_b, \\
T_s(y_a) = y_a,\label{eq:trans2}
\end{gather}
\end{subequations}
where $s$ is even and $\mathbb{Z}_2$-symmetric (that is $s^{ab} = (-1)^{|a||b|} s^{ba}$). We think of $s$ as $s \in \mathrm{Hom}(\widetilde V,\widetilde V^*)^{S_2}_0$.

Let $(V,L)$ be a super quantum Airy structure, and let $\widetilde T_s$ be the continuous automorphism of $\widehat{\mathcal W_{\hbar}}(\widetilde V)$ uniquely induced by $T_s$. Then $(V, \widetilde T_s \circ L)$ is also a super quantum Airy structure.

\begin{definition}
We say that the super quantum Airy structure $(V,\widetilde T_s \circ L)$ is \emph{gauge equivalent} to $(V,L)$, and we call $\widetilde T_s$ a \emph{linear gauge transformation}.
\end{definition}

\begin{remark}
We note here that linear gauge transformations do not change the order of a super quantum Airy structure, since it preserves the $\mathbb{Z}$-grading on $\widehat{\mathcal W_{\hbar}}(\widetilde V)$.
\end{remark}

Let us now explain the geometric meaning of gauge transformations. For any subspace $L \subseteq W$ we define its symplectic complement $L^{\perp} = \{ w \in W | \ \omega(l,w)=0 \text{ for all } l \in L\}$. We say that $L$ is Lagrangian if $L = L^{\perp}$. If $L$ is a Lagrangian subspace, every Lagrangian subspace $C$ such that $W = L \oplus C$ is called a Lagrangian complement of $L$. There is a bijection between $\mathrm{Hom}(\widetilde V,\widetilde V^*)^{S_2}$ and the set of Lagrangian complements of $\widetilde V^*$, given by $s \mapsto T_s(\widetilde V)$.

\subsubsection{Non-linear Gauge Transformations}

There are more general gauge transformations that one can consider. Since the $L_i$ can have arbitrary order, we do not need to insist that gauge transformations preserve the order of a super quantum Airy structure. Thus we can also consider non-linear gauge transformations.

Consider a formal series in $y$ of the form
\begin{equation}
s = \sum_{k=2}^{\infty} \frac{1}{k} s^{a_1...a_k} y_{a_1}... y_{a_k}.
\label{eq:general_s}
\end{equation}
We require that each term in (\ref{eq:general_s}) is even and that each $s^{a_1...a_k}$ is $\mathbb{Z}_2$-symmetric. The largest $k$ such that $s^{a_1...a_k} \neq 0$ is called the order of $s$. We also allow $s$ of infinite order, since the $L_i$ can have infinite order. Then the polynomial transformation
\begin{subequations}
\begin{gather}
T_s(x^a) = x^a + \frac{\partial s}{\partial y_a}, \\
T_s(y_a) = y_a,
\end{gather}
\end{subequations}
induces (by replacing $y_a$ with $\hbar \partial_a$) an automorphism $\widetilde T_s$ of $\widehat{\mathcal W_{\hbar}}(\widetilde V)$, which leaves the defining properties of super quantum Airy structures invariant.

\begin{definition}
We say that the super quantum Airy structure $(V, \widetilde T_s \circ L)$ is \emph{gauge equivalent} to $(V,L)$, and if the order of $s$ is greater than two, we call $\widetilde T_s$ a \emph{non-linear gauge transformation}.
\end{definition}

\subsubsection{The Partition Functions of Gauge Equivalent Super Quantum Airy Structures}

Given two gauge equivalent super quantum Airy structures, what is the relation between their partition functions? To answer this question we need to understand gauge transformations as conjugations of the differential operators $L_i$ of a super quantum Airy structure. Let $D_s \in \widehat{\mathcal W_{\hbar}}(\widetilde V)^{\geq 2}$ be of the form
\begin{equation}
D_s = \sum_{k=2}^{\infty} \hbar^k s^{a_1...a_k} \partial_{a_1}... \partial_{a_k},
\end{equation}
where the $s$ are even $\mathbb{Z}_2$-symmetric tensors. We claim that the $L_i'$ of the super quantum Airy structure $(V, \widetilde T_s \circ L)$ gauge equivalent to $(V,L)$ are given by the conjugated differential operators
\begin{equation}\label{eq:Li'}
L_i' =  \exp \left( \frac{D_s}{\hbar} \right) L_i \exp \left( - \frac{D_s}{\hbar} \right).
\end{equation}
But for this we need to make sense of this expression.

Clearly $\frac{D_s}{\hbar}$, regarded as an operator on the ring $R$ introduced in the Section \ref{s:ringR}, satisfies:
\begin{equation}
\forall n \in \mathbb N, \ \exists n_0 \in \mathbb N, \ n' \geq n_0 \implies  \left( \frac{D_s}{\hbar} \right)^{n'} (R) \subseteq \mathfrak m^n.
\end{equation}
Therefore $\exp \left( \frac{D_s}{\hbar} \right)$ makes sense as an operator on $R$. More precisely, for every $f \in R$ let
\begin{equation}
\exp \left( \frac{D_s}{\hbar} \right)f = \sum_{n=0}^{\infty} \frac{1}{n!} \left( \frac{D_s}{\hbar} \right)^n f.
\end{equation}
After projecting the right-hand-side onto terms of fixed degree in $\hbar$ and $x$, only finitely many terms are nonzero, so this is well defined. We have $\exp \left( \frac{D_s}{\hbar} \right) \exp \left( \frac{D_t}{\hbar} \right)=\exp \left( \frac{D_{s+t}}{\hbar} \right)$. Moreover, for any $f \in R$ we have the identity
\begin{equation}
\exp \left( \frac{[D_s,\,\cdot\,]}{\hbar} \right)(L_i) \cdot f = \exp \left( \frac{D_s}{\hbar} \right) L_i \exp \left( - \frac{D_s}{\hbar} \right) f.
\end{equation}
Therefore, \eqref{eq:Li'} makes sense, and indeed reconstructs the differential operators $L_i'$ of the gauge equivalent Airy structure $(V, \widetilde T_s \circ L)$.

Now let $Z$ be the partition function associated to the super quantum Airy structure $L_i$, that is, $Z$ is annihilated by the $L_i$, and $Z$ evaluated at $x=0$ is equal to $1$. Then $\exp \left( \frac{D_s}{\hbar} \right) Z $ is annihilated by the gauge transformed operators $L_i'$. Let $\mathcal N$ be equal to $\exp \left( \frac{D_s}{\hbar} \right) Z $ evaluated at $x=0$. It is easy to check that $\mathcal N = 1 + O (\hbar)$. Therefore $\mathcal N^{-1} \in R$ exists and
\begin{equation}\label{eq:gaugeZ}
Z'= \mathcal N^{-1} \exp \left( \frac{D_s}{\hbar} \right) Z \in R
\end{equation}
 is a solution of equation $L_i' Z' =0$ such that $Z'$ evaluated at $x=0$ is equal to $1$. By the uniqueness of partition function, $Z'$ coincides with the partition function associated to the Airy structure $L_i'$. This gives us the relation between the partition function $Z'$ of the gauge transformed super quantum Airy structure $(V, \widetilde T_s \circ L)$ to the partition function $Z$ of the original super quantum Airy structure $(V,L)$.

\begin{remark}
In \cite{ABCD} an alternative formula for $Z'$ in terms of formal gaussian integrals was given. A similar formula also works also in the supersymmetric case. However, it is only valid for linear gauge transformations. We are not aware of its generalization for non-linear gauge transformations.
\end{remark}


\subsection{Classical vs Quantum}

Quantum Airy structures were originally defined in \cite{KS} in terms of quantizations of classical Airy structures. In this section we explore the parallel story for super Airy structures.

\subsubsection{Classical Limit}

We now define the classical limit of super quantum Airy structures. We introduce the algebras ${\mathbb K}[W] = \frac{\mathcal W_{\hbar}(\widetilde V)}{\hbar \cdot \mathcal W_{\hbar}(\widetilde V)}$ and ${\mathbb K}[[W]] = \frac{\widehat{\mathcal W_{\hbar}}(\widetilde V)}{\hbar \cdot \widehat{\mathcal W_{\hbar}}(\widetilde V)}$.

\begin{definition}
We denote the quotient map by $\mathrm{Cl}$, and call it the \emph{classical limit}.
\end{definition}

Taking the classical limit thus amounts to replacing all $\hbar \partial_a$ by $y_a$, and setting all terms of higher order in $\hbar$ to zero. We interpret ${\mathbb K}[W]$ and ${\mathbb K}[[W]]$ as the superalgebras of polynomials and formal series respectively on some classical phase space with linear coordinate system $x^a, y_a$. The natural symplectic structure on $W$ induces a Poisson bracket on ${\mathbb K}[W]$ and ${\mathbb K}[[W]]$, which may also be computed as
\begin{equation}
\{ \mathrm{Cl}(f), \mathrm{Cl}(g) \} = \mathrm{Cl} \left( \frac{1}{\hbar} [f,g] \right).
\end{equation}

Now let $(V,L)$ be a super quantum Airy structure. Define $L_i^{\mathrm{cl}}(x,y)=\mathrm{Cl}(L_i)$ as being the classical limit of the operators $L_i$. The equation $L_i^{\mathrm{cl}}=0$ may then be thought of as defining a ``characteristic variety'' $\Sigma$ in phase space, even if this is slightly artificial in the presence of odd coordinates. Nevertheless, some geometric notions may still be defined, e.g.\ the Zariski cotangent space. It~may~be identified naturally with $V \oplus {\mathbb K}^{0|2}$, with ${\mathbb K}^{0|2}$ spanned by the additional Grassman variable together with its conjugate momentum. In particular, it is coisotropic (and in fact Lagrangian if there is no extra Grassman variable). The classical limit of the free energy may be interpreted as a parametrisation of $\Sigma$ in terms of the variables $x$ in a formal neighbourhood of zero. In fact, by repeating the steps performed in the proof of existence and uniqueness of the partition function associated to a super quantum Airy structure, one may show the following fact:

\begin{theorem}
There exists a unique $F_{\mathrm{cl}} \in {\mathbb K}[[\widetilde V]]^{\geq 3}$ such that $L_i^{\mathrm{cl}}(x, d F_{\mathrm{cl}})=0$. It coincides with the classical limit of the free energy associated to $L_i$.
\end{theorem}

\subsubsection{Super Classical Airy Structures and Quantization}

In the previous subsection we defined the classical limit of super quantum Airy structures. In fact, we could have started by defining super classical Airy structures, and think of super quantum Airy structures as quantizations of classical structures, in the spirit of \cite{KS}.

\begin{definition}
Let $V$ be a super vector space, $\widetilde V = V \oplus {\mathbb K}^{0|1}$, and $W = \widetilde V \oplus \widetilde V^*$.  A \emph{super classical Airy structure} is a pair $(V,L^{\mathrm{cl}})$, with $L^{\mathrm{cl}}: \ V^* \to {\mathbb K}[[W]]$ an even continuous linear operator such that:
\begin{enumerate}
\item The left ideal $\mathfrak L \subseteq {\mathbb K}[[W]]$ generated by $L^{\mathrm{cl}}(V^*)$ is involutive under the Poisson bracket, i.e. $\{ \mathfrak L, \mathfrak L\} \subseteq  \mathfrak L$;
\item $L^{\mathrm{cl}}(y_i)^{\leq 1} = y_i$.
\end{enumerate}
\end{definition}

We have seen that every super quantum Airy structure determines uniquely its classical limit. But given a super classical Airy structure, does there exist a quantization that is a super quantum Airy structure? This turns out to be a difficult problem in general.

Nevertheless, it  admits a simple solution in the special case of quadratic $L^{\mathrm{cl}}.$
Define $L$ initially by replacing $y_a \mapsto \hbar \partial_a$, with all derivatives to the right of $x^a$. Then all $L(\phi)$ are well-defined elements of the Weyl algebra, but they need not satisfy correct commutation relations. Simple calculation shows that
\begin{equation}
\zeta(\phi, \psi) =  \frac{1}{\hbar} L([\phi,\psi])- \frac{1}{\hbar^2} [L(\phi),L(\psi)]
\end{equation}
is a number for each $\phi, \psi \in V^*$. It follows from the Jacobi identity that $\zeta$ is a Lie superalgebra cocycle (see \cite{Fuks:cohomology} for the relevant definitions). Clearly $\zeta$ is even and continuous. Upon replacement
\begin{equation}
L(\phi) \mapsto L(\phi) + \phi(v), \quad v \in V_0,
\label{eq:L_normal_ordering_constant}
\end{equation}
$\zeta$ changes by a coboundary. It follows that the cohomology class $[\zeta] \in H^2(V^*,{\mathbb K})_0$ (continuous, even Lie superalgebra cohomology) doesn't depend on the operator ordering prescription. Moreover quantization of $L^{\mathrm{cl}}$ exists if and only if $[\zeta]=0$. In the finite-dimensional case it is guaranteed that $\zeta$ is a coboundary, because Weyl quantization is always possible. In this ordering scheme mixed terms $x^a y_b$ are replaced by $\frac{\hbar}{2} \left( x^a \partial_b + (-1)^{|a| |b|} \partial_b x^a \right)$.

Now, given a quantization, one can still ask if other quantizations may be obtained by performing transformations of the form (\ref{eq:L_normal_ordering_constant}). This is possible only if $v$ is such that $[\phi, \psi](v)=0$ for any $\phi, \psi$. In other words, $v$ has to be a cocycle. Therefore if a quantization exists, its ambiguity is measured by the cohomology group $H^1(V^*,\mathbb K)_0$.

\subsubsection{Bosonic Classical Airy Structures}

\begin{definition}\label{def:bosonic}
We call a classical Airy structure $L^{\mathrm{cl}}: \ V^* \to {\mathbb K}[[W]]$ \emph{bosonic}, if $V=V_0$ and $L^{\mathrm{cl}}$ does not depend on the extra fermionic variable comming from ${\mathbb K}^{0|1}$.
\end{definition}

Given any super classical Airy structure $L^{\mathrm{cl}}:  \ V^* \to {\mathbb K}[[W]]$, one can always produce a bosonic one based on $V_0$. Indeed, consider a restriction of this map to $L^{\mathrm{cl}}_0: \ (V_0)^*  
\to {\mathbb K}[[W]]$. Since $L^{\mathrm{cl}}$ is even, so is its restriction. Therefore image of $L^{\mathrm{cl}}_0$ consists of even elements. Those elements are linear  
combination of monomials, which can be of two types:  
either have no fermionic variables or have an even number of them.  
Observe that both of these subsets are closed under the  
Poisson bracket. Let $W_0=V_0\oplus V_0^*$ and let $\pi_0: \ {\mathbb K}[[W]] \to {\mathbb K}[[W_0]]$ be the  
projection onto the subspace spanned by those elements which have no  
fermionic variables. We define a bosonic Airy structure as a  
composition $L^{\mathrm{cl},b}=\pi_0\circ L^{\mathrm{cl}}_0:  \ (V_0)^* \to {\mathbb K}[[W_0]]$. The subspace spanned by those monomials in  
${\mathbb K}[[W]]$, which have even and nonzero number  
of fermionic variables, is a left ideal with respect to the Poisson bracket. Therefore the condition  
$\{ L^{\mathrm{cl}}_i, L^{\mathrm{cl}}_j\} = \hbar f_{ij}^k L^{\mathrm{cl}}_k$, for some $f_{ij}^k\in {\mathbb K}[[W]]$, implies that $\{ L^{\mathrm{cl},b}_i, L^{\mathrm{cl},b}_j \} = \hbar  
\pi_0(f_{ij}^k) L^{\mathrm{cl},b}_k$. This proves that $L^{\mathrm{cl},b}$ is a classical Airy structure.

\begin{remark}
The analog of the construction outlined above doesn't work for quantum super Airy structures in general. The reason is that the commutator of two terms with even, nonzero number of fermionic variables may contain terms with no fermionic variables. Instead one may consider the classical limit, remove fermionic variables and generators and then try to quantize again. This is always possible to carry out (but is possibly ambiguous) in the finite-dimensional case.
\end{remark}


\section{Finite-Dimensional Examples}    \label{sec-finite}

In this section we study examples of finite-dimensional quadratic super quantum Airy structures. Moreover we propose a classification scheme for these objects.

\subsection{Low-Dimensional Examples}

Our first step is to consider super vector spaces $V$ of low dimension. The ``purely bosonic" case, with ${\rm dim}\,V_1 = 0$ and no extra fermionic variable, is not the subject of our current studies. The case with ${\rm dim}\, V_1= 0$ but with an extra fermionic variable, i.e. $\widetilde V = V_0 \oplus \mathbb{K}^{0|1}$ is straightforward. Thus we will assume that ${\rm dim}\,V_1 \geq 1$. The ``purely fermionic" case, with ${\rm dim}\,V_0 = 0$, is essentially trivial, since the only allowed super quantum Airy structure has the form
\begin{equation}
L_i = \hbar \partial_{\theta_i},
\end{equation}
with $\theta_i$ odd variables. Thus we may also assume that ${\rm dim}\,V_0 \geq 1$.

\subsubsection{The Superalgebras of Dimension $(1|1)$}

The simplest non-trivial case then consists in the superalgebras of dimension $(1|1)$. In this section we classify all quadratic super quantum Airy structures that can be constructed from these superalgebras, assuming that we have no extra fermionic variables. We leave the case with an extra fermionic variables for future work.

\begin{remark}
We remark that from the point of view of the partition function, all examples based on superalgebras of dimension $(1|1)$ with no extra fermionic variables are rather trivial, since the partition function is purely bosonic (i.e. does not depend on the single Grassmann variable entering in the construction). This follows directly from the requirement that the free energy is $\mathbb{Z}_2$-even. Nevertheless, from an algebraic viewpoint it is interesting to classify which quadratic super quantum Airy structures can be constructed based on superalgebras of dimension $(1|1)$.
\end{remark}

Up to isomorphism, there exist three distinct complex Lie superalgebras of dimension $(1|1)$. We denote the bosonic generator by $L$ and the fermionic generator by $G$.
\begin{enumerate}
\item The abelian superalgebra, with commutation relations:
\begin{equation}
[L,L]=0, \qquad [L,G] = 0, \qquad [G,G]=0.
\end{equation}
\item The algebra of affine automorphisms of $\mathbb C^{0|1}$, with commutation relations:
\begin{equation}
[L,L]=0, \qquad [L,G] = \hbar G, \qquad [G,G] = 0.
\end{equation}
\item The $\mathcal N=1, d=1$ supersymmetry (SUSY) algebra, with commutation relations:
\begin{equation}
[L,L]=0, \qquad [L,G] =  0, \qquad [G,G] = \hbar L.
\end{equation}
\end{enumerate}
Below we provide all quadratic super quantum Airy structures that can be constructed as representations of these algebras, up to changes of bases and gauge transformations. This classification can be proved directly, by brute force calculations. For the sake of brevity, we omit the details.

In the following, $A,B,C,D \in \mathbb{C}$ always stand for arbitrary constants. We denote the bosonic variable by $x$ and the fermionic variable by $\theta$.
\begin{enumerate}
    \item For the abelian superalgebra, there are three families of quadratic super quantum Airy structures. The first one takes the form:
    \begin{subequations}
    \begin{gather}
    G = (1-x) \hbar \partial_{\theta},\\
    L = \hbar \partial_x - \frac{1}{2} A x^2 - \hbar (x \partial_x + \theta \partial_{\theta}) - \hbar D,     \end{gather}
    \end{subequations}
    while the second one is:
    \begin{subequations}
    \begin{gather}
            G = (1 - \hbar \partial_x) \hbar \partial_{\theta},\\
        L = \hbar \partial_x - \frac{\hbar^2}{2} C \partial_x^2 - \hbar D.
    \end{gather}
    \end{subequations}
    The third possibility is
    \begin{subequations}
    \begin{gather}
    G = \hbar \partial_{\theta},\\
    L = L_0,
    \end{gather}
    \end{subequations}
    where $L_0$ is an arbitrary $\theta$-independent bosonic generator.

    \item For the algebra of affine automorphisms of $\mathbb C^{0|1}$, there are also three families of quadratic super quantum Airy structures.The first one is:
    \begin{subequations}
    \begin{gather}
        G = (1-x) \hbar \partial_{\theta}, \\
        L = \hbar \partial_x - \frac{1}{2} A x^2 - \hbar x \partial_x - 2 \hbar \theta \partial_{\theta} - \hbar D.
    \end{gather}
    \end{subequations}
    The second family reads:
    \begin{subequations}
    \begin{gather}
    G = (1 - \hbar \partial_x) \hbar \partial_{\theta}, \\
    L = \hbar \partial_x - \hbar \theta \partial_{\theta} - \frac{\hbar^2}{2} C_x^{xx} \partial_x^2 - \hbar D,
    \end{gather}
    \end{subequations}
    while the third one is:
    \begin{subequations}
    \begin{gather}
        G = \hbar \partial_{\theta}, \\
        L = - \hbar \theta \partial_{\theta} + L_0,
    \end{gather}
    \end{subequations}
    where $L_0$ is an arbitrary $\theta$-independent bosonic generator.

    \item For the $\mathcal N=1, d=1$ SUSY algebra we have $L = \frac{2}{\hbar} G^2$, so it is sufficient to provide the form of $G$. There are again three possibilities:
    \begin{subequations}
    \begin{gather}
        G = \hbar \partial_{\theta} + \frac{\hbar}{2} \left( \theta \partial_x + x \partial_{\theta} \right), \\
        G = \hbar \partial_{\theta} + \frac{\hbar}{2} \theta \partial_x - \hbar^2 \partial_x \partial_{\theta}, \\
        G = \hbar \partial_{\theta} + \frac{\hbar}{2} \theta \partial_x.
    \end{gather}
    \end{subequations}
\end{enumerate}

\subsubsection{Superalgebras of Dimensions $(2|1)$ and $(1|2)$}

The brute force classification of quadratic super quantum Airy structures that can be obtained as representations of superalgebras of dimensions $(2|1)$ and $(1|2)$ is already a little tedious. We will thus only present a few interesting examples here.

\begin{example}
Our first example starts with the $(2|1)$-dimensional $\mathcal N=1, d=1$ SUSY algebra extended by a dilatation operator. We denote the two bosonic generators by $L_1,L_2$ and the fermionic generator by $G$.
It has commutation relations:
\begin{equation}\label{eq:extended}
[L_1,L_2]= \hbar L_2, \qquad [L_1,G] =  \frac{\hbar}{2}  G, \qquad [L_2,G] = 0, \qquad [G,G] = \hbar L_2.
\end{equation}
We construct a quadratic super quantum Airy structure with no extra fermion as a~representation of this algebra. We denote the bosonic variables by $x,y$ and the fermionic variable by $\theta$.

The quadratic super quantum Airy structure reads:
\begin{subequations}
\begin{gather}
    L_1 =   \hbar \partial_x   -  \frac{1}{2} A x^2 - \hbar \left( \frac{1}{2} + B \right) \theta \partial_{\theta} -  \hbar (1+B) y \partial_y - \hbar B x \partial_x - \hbar D, \\
    L_2 = \hbar \partial_y - \hbar  x \partial_y, \\
    G = \hbar \partial_{\theta} + \frac{\hbar}{2} \theta \partial_y - \hbar B x \partial_{\theta},
\end{gather}
\end{subequations}
where $A,B,D \in \mathbb{C}$ are arbitrary constants.

Note that, as for the $(1|1)$ examples, the partition function here does not depend on the fermionic variable $\theta$, since it must be $\mathbb{Z}_2$-even.
\end{example}

Let us now study examples where the partition function depends on fermionic variables.

\begin{example}
For our next example, we start with the same $(2|1)$-dimensional $\mathcal N=1, d=1$ SUSY algebra extended by a dilatation operator, with commutation relations \eqref{eq:extended}, but we construct a quadratic super quantum Airy structures with an extra fermionic variable. We denote the bosonic variables by $x^1,x^2$ and the fermionic variables by $\theta^0, \theta^1$.

The representation reads:
\begin{subequations}
\begin{eqnarray}
L_1 &  = & \hbar \partial_{x^1}  - 2 \hbar x^1\partial_{x^1} - \hbar x^2\partial_{x^2} - \frac{\hbar}{2}\theta^1\partial_{\theta^1} + \frac{3 \hbar}{2} \theta^0\partial_{\theta^0},
\\
L_2 & = & \hbar \partial_{x^2 } - \beta \hbar x^2\partial_{x^1} - \theta^0 \theta^1,
\\
G & = & \hbar\partial_{\theta^1} + \theta^0x^2 + \frac{\hbar}{2} \theta^1\partial_{x^2}- \frac{\hbar^2}{2}\beta\partial_{\theta^0}\partial_{x^1},
\end{eqnarray}
\end{subequations}
with $\beta \in {\mathbb C}$ an arbitrary constant.

The constraints $L_1 Z = L_2 Z = G Z = 0$ uniquely fix the partition function $Z$. The result is
\begin{equation}
Z =  \exp \left( \frac{1}{\hbar} x^2 \theta^0 \theta^1 \right).
\end{equation}
\end{example}

\begin{example}
Our last example in this section is a quadratic super quantum Airy structure with no extra fermionic variable, but such that its partition function depends on fermionic variables. To this end, we start with a superalgebra of dimension $(1|2)$. We denote its bosonic generator by $L$ and its fermionic generators by $G_1, G_2$, and the corresponding bosonic and fermionic variables by $x$ and $\theta_1, \theta_2$ respectively.

We choose the superalgebra with commutation relations:
\begin{equation}
\left[L,G_i\right] = \hbar\, G_i, \hskip 1cm \left[G_i,G_j\right] = 0, \hskip 1cm i,j = 1,2.
\end{equation}
Our representation is:
\begin{subequations}
\begin{eqnarray}
L & = & \hbar\partial_x  - x^2 - \theta^1\theta^2 + \hbar^2 \partial_x^2 + \hbar^2\partial_{\theta^1}\partial_{\theta^2},
\\[4pt]
\label{Airy:structure:simple:1:2}
G_1 & = & \hbar\partial_{\theta^1} - x\theta^2 + \hbar x\partial_{\theta^1} - \hbar\theta^2 \partial_x + \hbar^2\partial_x\partial_{\theta^1},
\\[4pt]
G_2 & = & \hbar\partial_{\theta^2} + x\theta^1 + \hbar x\partial_{\theta^2} + \hbar\theta^1 \partial_x + \hbar^2\partial_x\partial_{\theta^2}.
\end{eqnarray}
\end{subequations}

From the constraints $L Z = G_1 Z = G_2 Z = 0$, it is straightforward to compute the free energy perturbatively. Up to terms of order five in the variables $x,\theta^1,\theta^2$, we get:
\begin{equation}
\label{free:energy:simple:1:2}
F  =  \frac13x^3 - \frac15 x^5 + \left(x-x^3\right)\theta^1\theta^2 - \hbar \left(x^2 + \theta^1\theta^2\right) + \hbar^2 x + \ldots .
\end{equation}
Thanks to the simplicity of the Airy structure (\ref{Airy:structure:simple:1:2}) the classical free energy can be in this case expressed in terms of elementary functions:
\begin{equation}
\label{Airy:structure:simple:Fcl}
F_{\rm cl} =
\frac18\log\left(2x +\sqrt{1+4x^2}\right)+ \frac{x}{4}\sqrt{1+4x^2} - \frac{x}{2}
+
\frac{\sqrt{1+4x^2} + 2x -1}{\sqrt{1+4x^2} + 2x +1}\theta^1\theta^2.
\end{equation}

We will use this example to illustrate the action of the gauge symmetry on the partition function. Notice that under the gauge transformation
\begin{equation}
x \;  \to \; x + \hbar\partial_x,
\hskip 1cm
\theta^1 \; \to \; \theta^1 - \hbar\partial_{\theta^2},
\hskip 1cm
\theta^2 \; \to \; \theta^2 + \hbar\partial_{\theta^1},
\end{equation}
the operators go to
\begin{subequations}
\begin{eqnarray}
L \; \to \; L' & = & \hbar\partial_x - x^2 - 2\hbar x\partial_x - \hbar - \theta^{1} \theta^2 - \hbar\theta^1\partial_{\theta^1} - \hbar \theta^2\partial_{\theta^2},
\\[4pt]
G_1 \; \to \; G_1' & = & \hbar\partial_{\theta^1} - x \theta^2 - 2\hbar\theta^2\partial_x,
\\[4pt]
G_2 \; \to \; G_2' & = & \hbar\partial_{\theta^2} + x \theta^1 + 2\hbar\theta^1\partial_x.
\end{eqnarray}
\end{subequations}
In particular, $L', G_1'$ and $G_2'$ are now all first order differential operators, and we can solve explicitly for the free energy at the quantum level. We get:
\begin{equation}
\label{Airy:structure:simple:Fprim}
F' = -\frac14x(x+1) -\frac18(1+4\hbar)\log(1-2x) + \frac{x+2\hbar}{1-2x}\theta^1\theta^2.
\end{equation}

According to our discussion in Section \ref{ss:cl:q} -- see \eqref{eq:gaugeZ} -- we have an identity
\begin{equation}
\label{gauge:tarnform:explicite}
{\rm e}^{\frac{1}{\hbar}F} = \sum\limits_{k=0}^\infty \frac{(-\hbar)^k}{k!} \left(\left[{\textstyle \frac12}\partial_x^2 + \partial_1\partial_2,\,\cdot\,\right]\right)^k {\rm e}^{\frac{1}{\hbar}F'}.
\end{equation}
Using a simple Mathematica code, validity of (\ref{gauge:tarnform:explicite})
may be checked to any required order (at least for the explicitly known classical part of the l.h.s. of (\ref{gauge:tarnform:explicite})).

\end{example}


\subsection{Classification scheme}
\label{s:classification}

So far we have studied a few examples of low-dimensional quadratic super quantum Airy structures. In this section we present a classification scheme for finite-dimensional super quantum Airy structures whose vector space $V^*$ is a finite-dimensional Lie superalgebra $\mathfrak g$. We focus on finite-dimensional super quantum Airy structures, but we remark here that part of the discussion is also valid for the infinite-dimensional case.

Let $\mathfrak g$ be a finite-dimensional Lie superalgebra. We denote its structure constants by $f$.
The process of finding  all quadratic super quantum Airy structures for $\mathfrak g$ can be divided into three steps:
\begin{enumerate}
\item
Classify at most quadratic classical hamiltonians $L_i^{\mathrm{cl}}$ up to affine automorphism.
\item
Look for points of the zero locus $L_i^{\mathrm{cl}}=0$ at which gradients of $L_i^{\mathrm{cl}}$ are linearly independent in order to rewrite $L_i^{\mathrm{cl}}$ in the form required by the definition of super classical Airy structures.
\item
Quantize the system.
\end{enumerate}
Parts of this procedure are expressed as classical problems in representation theory (classify all representations of a given dimension, find all invariant symplectic forms, compute certain cohomology groups) or algebraic geometry (describe the zero locus of a given set of polynomials), whose solutions are known at least in certain special cases.

\begin{remark}
For clarity in this section we will write $L_i$ for the classical hamiltonians, dropping the superscript $\mathrm{cl}$.
\end{remark}

\subsubsection{Purely Quadratic Hamiltonians}

As a first step one has to classify all representations of $\mathfrak g$ by purely quadratic hamiltonians:
\begin{equation}
\{ L_i^2 , L_j^2 \} = f_{ij}^k L_k^2,
\label{eq:purely_quadratic_Poisson}
\end{equation}
where $\{ \cdot, \cdot \}$ is the Poisson bracket. $L_i^2$ have to be polynomials in $2 \dim \mathfrak g$ variables (or with one extra Grassman variable together with its conjugate momentum) $z_{a}$ subject to the elementary Poisson bracket relations
\begin{equation}
\{ z_a, z_b \} = \omega_{ab},
\label{eq:Poisson_symplectic}
\end{equation}
where $\omega$ is a symplectic form on the super vector space $W$ with basis $\{ z_a \}$. Hence
\begin{equation}
L_i^2 = \frac{1}{2} z_a M_i^{ab} z_{b},
\label{eq:quadratic_L}
\end{equation}
with $M_i^{ab} = (-1)^{|a||b|} M_i^{ba}$. This symmetry condition is equivalent to demanding that the linear operator on $W$ given by $z_a \mapsto z_b M^{bc}_i \omega_{ca}$ is an infinitesimal symplectomorphism. The Poisson bracket relations (\ref{eq:purely_quadratic_Poisson}) are satisfied if and only if the operators $M_i$ furnish a linear representation of $\mathfrak g$ on $W$. Conversely, given any representation of $\mathfrak g$ with an invariant symplectic form $\omega$ we may construct the space ${\mathbb K}[z]$, with Poisson bracket given by (\ref{eq:Poisson_symplectic}), and then we define the $L_i^2$ by (\ref{eq:quadratic_L}).

To summarize, the classification of purely quadratic hamiltonians (depending on a given number of variables) representing the algebra $\mathfrak g$ is equivalent to classifying all symplectic representations\footnote{Two symplectic representations $W_1, W_2$ are regarded as equivalent if there exists an even $\mathfrak g$-intertwiner $W_1 \to W_2$ which is also a symplectomorphism.} of $\mathfrak g$ of dimension $2 \dim \mathfrak g$ or $2 \dim \mathfrak g + 0|2$. 

\subsubsection{Affine Extensions}
\label{s:affine}

Now fix a representation of $\mathfrak g$ by quadratic hamiltonians $L_i^2$. We ask if we can add linear and constant terms $L_i^1, L_i^0$ such that the commutation relations are preserved:
\begin{subequations}
\begin{gather}
L_i = L_i^2 + L_i^1 + L_i^0, \\
\{ L_i, L_j \} = f_{ij}^k L_k. \label{eq:L_PB}
\end{gather}
\end{subequations}
Assuming that equation (\ref{eq:purely_quadratic_Poisson}) holds, (\ref{eq:L_PB}) reduces to the following system of equations:
\begin{subequations}\label{eq:affine_ext_constraints}
\begin{gather}
\{ L_i^2, L_j^1 \} + \{ L_i^1, L_j^2 \} - f_{ij}^k L_k^1 =0, \label{eq:L1_constraint} \\
\{ L_i^1, L_j^1 \} - f_{ij}^k L_k^0 =0. \label{eq:L0_constraint}
\end{gather}
\end{subequations}
We observe that given any solutions of these equations, a new solution may be obtained by shifting even variables:
\begin{align}
L_i'
=& \exp \left( \{\epsilon(z), \cdot \} \right) (L_i) \nonumber\\
=& L_i +  \{\epsilon(z), L_i^1 + L_i^2 \} + \frac{1}{2} \{ \epsilon(z), \{\epsilon(z), L_i^2 \} \},\label{eq:shift}
\end{align}
where $\epsilon(z)$ is an arbitrary linear combination of even $z_a$.
This is equivalent to the replacement:
\begin{subequations}\label{eq:L_shifts}
\begin{gather}
L_i^1 \mapsto L_i^1 +  \{ \epsilon(z) , L_i^2 \}, \\
L_i^0 \mapsto L_i^0 + \{ \epsilon(z), L_i^1 \} + \frac{1}{2} \{\epsilon(z) , \{ \epsilon(z) , L_i^2 \} \}. \label{eq:L0_shifts}
\end{gather}
\end{subequations}
Shifting even variables is an isomorphism of the Poisson algebra. Solutions related by such transformations should be regarded as equivalent at this stage.

\begin{definition}\label{d:equiv}
We say that two solutions of \eqref{eq:L_PB} are \emph{equivalent} if they are related by a shift of even variables as in \eqref{eq:L_shifts}.
\end{definition}

Our plan to classify solutions of equations (\ref{eq:affine_ext_constraints}) up to equivalence is as follows: first solve the constraint (\ref{eq:L1_constraint}), then find solutions of (\ref{eq:L0_constraint}) for a given $L^1_i$.

\begin{lemma}
Given $L_i^2$ that satisfy \eqref{eq:purely_quadratic_Poisson}, the space of all solutions to \eqref{eq:L1_constraint} up to equivalence coincides with the even subspace of the cohomology group $H^1(\mathfrak g, W)$.
\end{lemma}

\begin{proof}
Observe that equation (\ref{eq:L1_constraint}) is linear in $L^1$. Moreover it has the form of a $1$-cocycle condition for $\mathfrak g$ valued in the module $W$. Solutions of the form $L_i^1 =  \{ \epsilon(z) , L_i^2 \}$ for some linear combination of even variables $\epsilon(z)$ may be identified with coboundaries. Therefore the space of all solutions up to equivalence coincides with the even subspace of the cohomology group $H^1(\mathfrak g, W)$.
\end{proof}

Let us now turn to \eqref{eq:L0_constraint}. Equation (\ref{eq:L0_constraint}) is quadratic in $L^1$.  If $L^1$ is already chosen such that \eqref{eq:L1_constraint} is satisfied, (\ref{eq:L0_constraint}) is a linear equation for $L^0$. For some $L^1$ it may turn out that there are no solutions at all.

\begin{definition}
We call linear terms $L^1_i$ such that there exists some solution to \eqref{eq:L0_constraint} \emph{admissible}.
\end{definition}

Admissibility depends only on the equivalence class of $L^1$. Indeed, if $L^1$ is such that a consistent $L^0$ can be chosen, then for any $L'^1$ of the form $L'^1_i=L^1_i +  \{ \epsilon(z), L^2_i \}$, a consistent $L'^0$ can be found in the form of (\ref{eq:L0_shifts}).

\begin{lemma}
Given an admissible $L^1_i$ that satisfies \eqref{eq:L1_constraint}, the space of all solutions to \eqref{eq:L0_constraint} is an affine space over $\left( \frac{\mathfrak g}{[\mathfrak g, \mathfrak g]} \right)^*_0$.
\end{lemma}

\begin{proof}
Suppose that $L^1$ is admissible and pick one $L^0$ satisfying (\ref{eq:L0_constraint}). Then $L'^0$ is another solution if and only if
\begin{equation}
f_{ij}^k \left( L'^0_k - L^0_k \right)=0.
\end{equation}
This equation means that the even linear functional $L'^0-L^0 \in \mathfrak g_0^*$ vanishes on the commutator ideal $[\mathfrak g, \mathfrak g]$. Therefore it may be regarded as an element of $\left( \frac{\mathfrak g}{[\mathfrak g, \mathfrak g]} \right)_0^*$.
\end{proof}

\subsubsection{Choice of the Origin and a Lagrangian Complement} \label{sec:smooth_point}

Suppose that we have chosen a set of at most quadratic hamiltonians $L_i$ satisfying (\ref{eq:L_PB}). Now we look for hamiltonians in the same equivalence class that are in the form required by the definition of super classical Airy structures. In other words, we want to perform a shift \eqref{eq:shift} to bring the~hamiltonians in a form such that $L^0_i=0$ and the $L^1_i$ are linearly independent. After such shift, $L_i^0$ changes according to \eqref{eq:L0_shifts}. Requiring that the new $L_i^0$ vanishes, we get an equation for $\epsilon(z)$:
\begin{equation}
L_i^0 +  \{ \epsilon(z), L_i^1 \} + \frac{1}{2} \{ \epsilon(z) , \{ \epsilon(z) , L_i^2 \} \}=0.
\end{equation}
Its set of all solutions may be identified with the zero locus $\Sigma = \{ z \in W_0 \ | \ (L_i)_0(z)=0  \}$ (here the~subscript $0$ means that we ignore odd variables and odd generators of $\mathfrak g$). Not all solutions are admissible, because we must ensure that after the shift, the linear terms $L_i^1$ are~linearly independent. This means that we must keep only those points for which the matrix of partial derivatives
\begin{equation}
D^a_i = \left. \frac{\partial}{\partial z_a} L_i \right|_{z=0}
\end{equation}
has rank $\dim(\mathfrak g)$. The set of all elements of $\Sigma$ satisfying this condition will be denoted by $\Sigma_{s}$. It~is~a~Zariski open subset of $\Sigma$. In \cite{KS} elements of $\Sigma_{s}$ were called the smooth points of $\Sigma$. We~note that this is not completely consistent with the standard terminology. For example the zero locus $Z(f)$ of the polynomial $f(x,y) = x^2 \in \mathbb K[x,y]$ is nonsingular, even though we have $\left. df \right|_{Z(f)}=0$.

Once a point of $\Sigma_s$ is chosen and the generators $L_i$ are put in a form with $L^0_i=0$ with linearly independent $L^1_i$, we define
\begin{equation}
y_i = L_i^1.
\end{equation}
Equation (\ref{eq:L0_constraint}) combined with $L^0_i=0$ gives
\begin{equation}
\{ y_i, y_j \}=0.
\end{equation}
Therefore the variables $y_i$ span an isotropic subspace in $W$. If there is no ``extra'' odd variable, that is $\dim W = 2 \dim \mathfrak g$, this subspace is Lagrangian. For simplicity, for the remainder of this section we restrict attention to this special case. The analysis in the situation with an additional fermion has to be slightly adjusted. The linear span of $y_i$ will be denoted by $V^*$.

We may now find elements $x^i$ such that
\begin{subequations}
\begin{gather}
\{ y_i, x^j \} = \delta_i^j, \\
\{ x^i, x^j \} =0.
\end{gather}
\end{subequations}
The set of all solutions to these conditions is in one-to-one correspondence with the set of Lagrangian complements $V$ of the subspace $V^* \subset W$ spanned by $y_i$. Once the $L_i$ are expressed in terms of the $y_i$ and $x^i$, we have a super classical Airy structure.

Let us end this section by specifying how the super classical Airy structures constructed by the procedure outlined above depend on a number of arbitrary choices made along the way.

First of all, there is an ambiguity in the choice of $x^i$, or equivalently, in the choice of the Lagrangian complement $V$ of the subspace $V^* \subseteq W$. As described earlier, this ambiguity is precisely the so-called gauge freedom. We regard Airy structures related by a gauge transformation as equivalent. Indeed, not only gauge-transformed generators are related by an explicit automorphism of the Poisson algebra, but also (after quantization) partition functions are related by a formal gaussian smearing transformation.

Secondly, we have chosen a point of $\Sigma_s$. We can get more such points by exponentiating the~action of $\mathfrak g_0$ on ${\mathbb K}[z]_0$. This is possible, because action of a Lie algebra on a finite-dimensional vector space exponentiates to an action of the corresponding Lie group, and each $\{ L_i, \cdot \}$ preserves a filtration of ${\mathbb K}[z]$ by finitely-dimensional subspaces $\{ {\mathbb K}[z]^{\leq n} \}_{n \in \mathbb N}$ of all polynomials of degree at most $n$. Clearly hamiltonians related by $G$-transformations give rise to isomorphic Airy structures. It may happen that the set $\Sigma_s$ is disconnected. Then it is possible that several non-isomorphic Airy structures may be obtained from the same set of hamiltonians.

Since the dimension of $\Sigma_s$ coincides with the dimension of $\mathfrak g_0$ and the stabilizer of each point of $\Sigma_s$ in $G$ is discrete, $G$ acts locally transitively on $\Sigma_s$. Therefore $G$-orbits are open in $\Sigma_s$ (with respect to the analytic topology). Since set $\Sigma_s$ is semialgebraic, it has finitely many connected components. Therefore its connected components are clopen. Combining these two facts we conclude that the orbits of the $G$-action on $\Sigma_s$ are precisely the connected components of $\Sigma_s$. Notice that this is not necessarily true for the $G$-action on the whole $\Sigma$, since the $G$-orbits in $\Sigma$ are in general not open in $\Sigma$. For example if $L_i$ are homogeneous of degree $2$, then $z=0$ is always a solution which is a fixed point for $G$, hence an orbit. Except for some trivial cases it is not a discrete point of $\Sigma$, hence we have orbits of $G$ in $\Sigma$ which are proper subsets of their connected components. For any orbit $O \subseteq \Sigma_s$ the stabiliser $\Gamma$ of any $p \in O$ is discrete. This means that $O$ may be identified with the homogeneous space $\frac{G}{\Gamma}$ and that $G$ is the universal covering space of $O$. In particular, we conclude that the homotopy groups of $O$ are given by $\pi_1 (O) \cong \Gamma$ and $\pi_k (O) \cong \pi_k(G)$ for $k \geq 2$. Finally, we observe that the hamiltonian vector fields generated by $L_i$ provide a global framing for the tangent bundle of $O$, so $TO \twoheadrightarrow O$ is trivial.

\subsubsection{Quantization}

Suppose that we have constructed a quadratic super classical Airy structure through the procedure outlined above. We then ask whether this classical Airy structure may be lifted to the quantum level. We will see that, in the finite-dimensional case, the answer is always affirmative. (In the infinite-dimensional case, there is a possible cohomological obstruction). Moreover it turns out that the set of all consistent quantizations is an affine space over $\left( \frac{\mathfrak g}{[\mathfrak g, \mathfrak g]} \right)_0^*$.
\begin{lemma}
Let the $L_i$ form a quadratic super classical Airy structure, and define
\begin{equation}
\zeta_{ij} = \frac{1}{\hbar^2} \left( [L_i,L_j] - \hbar f_{ij}^k L_k \right),
\end{equation}
which is a constant, independent of $x_i$, $y_i$ and $\hbar$.
Then the $L_i$ can be lifted to a quadratic super quantum Airy structure if and only if the cohomology class $[\zeta] \in H^2(\mathfrak g, \mathbb K)_0$ vanishes. If this is the case, the space of all consistent quantizations is an affine space over $\left( \frac{\mathfrak g}{[\mathfrak g, \mathfrak g]} \right)_0^*$.
\end{lemma}

\begin{proof}
To quantize $L_i$, we need to replace all $y_i$ by $\hbar \frac{\partial}{\partial x^i}$. The meaning of this operation is ambiguous for mixed terms, i.e. for products $x^i y_j$, which supercommute on the classical level but not after quantization. Due to this ambiguity, in general we need to introduce constants of order $\hbar$ in our quantum hamiltonians. These terms are invisible at the classical level.

Note that it is always possible to replace $x^i y_j$ by the ``normally-ordered" expressions $\hbar x^i \frac{\partial}{\partial x^i}$, but unfortunately this ordering prescription does not lead to differential operators satisfying the commutation relations of the algebra $\mathfrak g$ in general. However, a simple calculation shows that this is almost true, in the sense that the quantity
\begin{equation}
\zeta_{ij} = \frac{1}{\hbar^2} \left( [L_i,L_j] - \hbar f_{ij}^k L_k \right)
\end{equation}
does not depend on $x$, $y$ or $\hbar$. By the Jacobi identity, it satisfies the $2$-cocycle condition
\begin{equation}
\zeta_{k [i } f^k_{jl]} =0,
\end{equation}
where square bracket denotes $\mathbb Z_2$-graded skew-symmetrization. One may try to get rid of the problematic $\zeta$ by shifting the generators by constants proportional to $\hbar$:
\begin{equation}
L_i \mapsto L_i + \hbar D_i,
\end{equation}
with $D_i =0$ for $|i|=1$. This transformation has the following effect on $\zeta$:
\begin{equation}
\zeta_{ij} \mapsto \zeta'_{ij} = \zeta_{ij} - f^k_{ij} D_k.
\label{eq:zeta_cohomologous_change}
\end{equation}
This means that $\zeta$ changes by a coboundary. In other words, the cohomology class of $\zeta$ depends only on the classical generators and not on the choice of the ordering prescription. In order to get the commutation relations of $\mathfrak g$ at the quantum level, we have to impose the condition $\zeta_{ij}'=0$. By the preceding discussion, this is possible if and only if the cohomology class $[\zeta] \in H^2(\mathfrak g, \mathbb K)_0$ vanishes. Now suppose that we have found a particular solution $D_i$ such that $\zeta_{ij}'=0$. We ask if other choices $D_i'$ are possible. It follows from the formula (\ref{eq:zeta_cohomologous_change}) that $D_i'$ is a consistent constant term for $L_i$ if and only if
\begin{equation}
f^k_{ij} \left( D_i' -D_i \right) =0.
\end{equation}
This condition means that the functional $D \in \mathfrak g_0^*$ vanishes on the commutator ideal $[\mathfrak g, \mathfrak g]$. Therefore the space of all consistent constant terms is affine over $\left( \frac{\mathfrak g}{[\mathfrak g, \mathfrak g]} \right)_0^*$.
\end{proof}

The situation is simpler when $\mathfrak g$ is finite-dimensional, in which case there is no cohomological obstruction to quantization.

\begin{lemma}
If $\mathfrak g$ is finite-dimensional, then $[\zeta] = 0$, and all quadratic super classical Airy structures can be quantized.
\end{lemma}

\begin{proof}
If $\mathfrak g$ is finite-dimensional, there exists a simpler quantization procedure, namely Weyl quantization. In this scheme every term of the form $x^i y_j$ in the classical hamiltonians is replaced by $\frac{\hbar}{2} \left( x^i \frac{\partial}{\partial x^j} + \frac{\partial}{\partial x^j} x^i \right)$ at the quantum level. This has the advantage that the commutation relations between quantum $L_i$ are automatically satisfied. Thus a quantization always exists: it not necessary to shift further with $D_i$-terms (although it is still possible to construct other quantizations as in the previous Lemma). Since the space of consistent quantizations now has a distinguished origin, it is a vector (rather than affine) space. This has the corollary that the cohomology class $[\zeta]$ described in the previous paragraph vanishes identically, since Weyl quantization guarantees existence of a quantization for finite-dimensional quadratic super classical Airy structures.
\end{proof}

\begin{remark}
In the infinite-dimensional case, Weyl quantization fails in general. Indeed, in order to have well-defined operators $L_i$, we need to be able to use commutation relations to put all $\frac{\partial}{\partial x^j}$ to the right of all $x^i$ in such a way that the coefficient in front of each derivative is finite. If we insist on Weyl ordering, this may turn out to be impossible. Indeed, we would like to replace expressions of the form $b_i^j x^i y_j$ with $\frac{\hbar}{2} b_i^j \left( x^i \frac{\partial}{\partial x^j} + \frac{\partial}{\partial x^j} x^i \right) = \frac{\hbar}{2} b_i^j x^i \frac{\partial}{\partial x^j} + \frac{\hbar}{2} b_i^i$.  Unfortunately, the contraction $b_i^i$ is meaningless in general, as it contains an infinite sum. Therefore $[\zeta] \neq 0$ is possible. 
\end{remark}

We have claimed earlier that Weyl quantization, whenever possible, is distinguished among all quantization schemes. The main reason for this is that it is canonical, with no room for arbitrary choices. Another pleasant property is that it is covariant with respect to symplectic transformations, in the sense that classical hamiltonians expressed in different coordinates are quantized to the same (up to isomorphism) quantum operators.
This is not true for all quantization schemes. We illustrate this feature with the simplest possible example.

\begin{example}
Consider the purely quadratic hamiltonian
\begin{equation}
L= \frac{1}{2} x^2 - \frac{1}{2} y^2.
\end{equation}
In this situation Weyl quantization and normal ordering quantizations agree. Both give
\begin{equation}
L_q = \frac{1}{2} x^2 - \frac{\hbar^2}{2} \partial_x^2.
\end{equation}
Define now $a= \frac{x+y}{\sqrt{2}}$ and $b= \frac{x-y}{\sqrt{2}}$. The classical hamiltonian takes the form
\begin{equation}
L' = ab.
\end{equation}
The Weyl quantization of this hamiltonian differs from the normal ordering quantization (and infinitely many different quantization prescriptions). It gives
\begin{equation}
L_q' = \hbar b \frac{\partial}{\partial b} + \frac{\hbar}{2}.
\end{equation}
We introduce new generators $x = \frac{b + \partial_b}{\sqrt{2}}$, $\partial_x = \frac{-b + \partial_b}{\sqrt{2}}$. They satisfy the same algebraic relations as $b, \partial_b$. Reexpressing\footnote{This transformation may be implemented by an automorphism of the Weyl algebra.} $L_q'$ in terms of $x, \partial_x$ we recover the hamiltonian $L_q$. This doesn't happen if $L'$ is quantized with any other ordering prescription.\end{example}

\subsubsection{An Example}

We now illustrate the steps outlined in the classification scheme in a specific example. Our starting point is the Lie superalgebra with a single bosonic generator $H$ and two fermionic generators $Q_1,Q_2$, with the only nonzero commutator $[Q_1,Q_2]=H$. We note that the simply connected Lie group generated by the bosonic part of this algebra is isomorphic to $\mathbb C$.

\begin{enumerate}
\item
The first step in the classification is the construction of all representations of the superalgebra by purely quadratic hamiltonians. For the sake of brevity, we will not perform this step here. Rather, we will focus on one choice of a purely quadratic representation of the superalgebra.

Let us denote the bosonic variables by $x,y$ satisfying $\{ y, x \}=1$, and the fermionic variables by $\theta_1, \theta_2, \xi_1, \xi_2$ satisfying $\{ \theta_i , \theta_j \} = \{ \xi_i, \xi_j \} =0$ and $\{ \xi_i, \theta_j \} = \delta_{ij}$. We will choose the following purely quadratic representation of the superalgebra:
\begin{subequations}
\begin{align}
Q_1^2 :=& x \theta_1, \\
Q_2^2 :=& y \xi_1, \\
H^2 :=& \{ Q_1^2, Q_2^2 \} = x y -  \theta_1 \xi_1.
\end{align}
\end{subequations}

\item
The second step consists in adding linear and constant terms in a way that preserves the Poisson brackets. We first construct the most general linear terms by solving \eqref{eq:L1_constraint}, up to equivalences generated by shifts of even variables (see Definition \ref{d:equiv}). We then solve for constant terms using \eqref{eq:L0_constraint}.

In our context, \eqref{eq:L1_constraint} becomes the equations:
\begin{subequations}
\begin{align}
\{ Q_1^2, Q_2^1 \} + \{ Q_1^1, Q_2^2 \} =& H^1,\\
\{Q_1^2, Q_1^1 \} =&0,\\
\{Q_2^2, Q_2^1 \} =& 0,
\end{align}
\end{subequations}
with $Q_1^1, Q_2^1$ odd, and $H^1$ even. The most general solution to these equations is
\begin{subequations}
\begin{align}
Q_1^1 =& \alpha_1 \theta_1 + \beta_1 \theta_2 + \gamma_1 \xi_2, \\
Q_2^1 =& \alpha_2 \xi_1 + \beta_2 \xi_2 + \gamma_2 \theta_2,\\
H^1 =& \alpha_2 x + \alpha_1 y,
\end{align}
\end{subequations}
for constants $\alpha_1, \alpha_2, \beta_1, \beta_2, \gamma_1, \gamma_2 \in \mathbb{C}$. However, the terms $\alpha_1 \theta_1$ and $\alpha_2 \xi_1$ in $Q_1^1$ and $Q_2^1$ respectively (and the corresponding terms in $H^1$) can be obtained by shifting the even variables $x \mapsto x + \alpha_1$ and $y \mapsto y + \alpha_2$. According to the classification scheme, we consider these solutions as equivalent (see Definition \ref{d:equiv}), and so we can set $\alpha_1 = \alpha_2 = 0$ without loss of generality, and we get the linear terms
\begin{subequations}
\begin{align}
Q_1^1 =&  \beta_1 \theta_2 + \gamma_1 \xi_2, \\
Q_2^1 =&  \beta_2 \xi_2 + \gamma_2 \theta_2,\\
H^1 =& 0.
\end{align}
\end{subequations}

We then need to add constant terms using \eqref{eq:L0_constraint}. First, since $Q_1$ and $Q_2$ are fermionic, they cannot have constant terms. From \eqref{eq:L0_constraint}, the equations that we have to solve are
\begin{equation}\label{eq:L0_constraint_ex}
\{Q_1^1, Q_1^1\} = \{Q_2^1, Q_2^1 \} = 0, \qquad \{Q_1^1, Q_2^1\} = H^0.
\end{equation}
The first two equalities impose that $\beta_1 \gamma_1 = \beta_2 \gamma_2 = 0$. In particular, we must have that either $\beta_1 =0 $ or $\gamma_1 = 0$. In the case with $\beta_1$ non-zero, we can always use the transformation $\theta_2 \mapsto \xi_2$, $\xi_2 \mapsto \theta_2$ (which does not modify the original quadratic hamiltonians and preserves Poisson brackets) to make $\beta_1$ vanish. Thus we can assume without loss of generality that $\beta_1 = 0$. The linear terms become
\begin{subequations}
\begin{align}
Q_1^1 =&  \gamma_1 \xi_2, \\
Q_2^1 =&  \beta_2 \xi_2 + \gamma_2 \theta_2,\\
H^1 =& 0.
\end{align}
\end{subequations}
Then, from the last equality in \eqref{eq:L0_constraint_ex}, we get:
\begin{equation}
H^0 = \gamma_1 \gamma_2.
\end{equation}
As a result, up to equivalences we have found the general representation:
\begin{subequations}
\begin{align}
Q_1 =& \gamma_1 \xi_2 +x \theta_1,\\
Q_2 =& \beta_2 \xi_2 + \gamma_2 \theta_2 + y \xi_1, \\
H =&  \gamma_1 \gamma_2 + x y -  \theta_1 \xi_1,
\end{align}
\end{subequations}
with $ \beta_2 \gamma_2 = 0$. By changing values of the coefficients $\gamma_1, \beta_2, \gamma_2$ one may obtain several genuinely different Airy structures.

For clarity we now pick one particular example. We remark that the ``trivial'' case with $\gamma_1=\beta_2=\gamma_2=0$ is rather boring; indeed, in this case it turns out that $\Sigma_s$ is empty. Thus we consider the slightly more complicated case with $\gamma_1 = \gamma_2 = 1$ and $\beta_2 = 0$. The generators take the form:
\begin{subequations}
\begin{align}
Q_1 =& \xi_2 +x \theta_1,\\
Q_2 =&  \theta_2 + y \xi_1, \\
H =&  1 + x y -  \theta_1 \xi_1.
\end{align}
\end{subequations}

\item
The third step is to shift even variables to find an equivalent representation with the generators in the form of a super quantum Airy structure (if possible). As explained in Section \ref{sec:smooth_point}, this corresponds to finding these points of the zero locus
\begin{equation}
(H)_0 = 1 + x y = 0
\end{equation}
where the gradient of $(H)_0$ is nonzero. They take the form
\begin{equation}
x = a, \qquad y= - \frac{1}{a},
\end{equation}
with $a \in \mathbb{C} \setminus \{ 0 \}$. We see that $\Sigma_s$ is connected and isomorphic to $\mathbb C^{\times} \cong \frac{\mathbb C}{\mathbb Z}$, so its fundamental group is $\mathbb Z$. Hence the stabiliser of any point with respect to the $\mathbb C$-action is isomorphic to $\mathbb Z$. We can double check that this is correct by solving for the Hamiltonian flow generated by $(H)_0$. To this end we evaluate the Poisson brackets
\begin{subequations}
\begin{gather}
\{ H, x \} = x, \\
\{ H, y \} = -y.
\end{gather}
\end{subequations}
Therefore the solution of the Hamilton equations $\frac{df}{dt} = \{ H,f \}$ takes the form
\begin{subequations}
\begin{gather}
x(t) = e^t x(0), \\
y(t) = e^{-t} y(0).
\end{gather}
\end{subequations}
Clearly we can obtain any element of $\Sigma_s$ by flowing from any given initial point, say $x(0)=1$, $y(0)=-1$. The relation $xy+1=0$ is explicitly preserved by the flow. Moreover we have $(x(t),y(t))= (x(0),y(0))$ if and only if $t \in 2 \pi i \mathbb Z$, confirming that the stabiliser of $(x(0),y(0))$ is infinite cyclic.

Since $\Sigma_s$ is connected, we are free to choose $a=1$. Performing the corresponding shifts of the even variables $x \mapsto x + 1$ and $y \mapsto y - 1$ brings the generators in the form:
\begin{subequations}
\begin{align}
Q_1 =&\theta_1+ \xi_2  +x \theta_1,\\
Q_2 =&  \theta_2 - \xi_1 + y \xi_1, \\
H =& y - x + x y -  \theta_1 \xi_1.
\end{align}
\end{subequations}

\item Finally, the last step in the classification is to choose a Lagrangian complement. First, we define new canonical momenta $\pi_1 = Q_1^1$, $\pi_2 = Q_2^1$, and $p = H^1$ to bring the linear terms in the form of a super quantum Airy structure. More explicitly,
\begin{subequations}
\begin{align}
\pi_1 =&  \theta_1+\xi_2, \\
\pi_2 =& \theta_2 - \xi_1,\\
p =& y-x.
\end{align}
\end{subequations}
Then, we have to choose a Lagrangian complement, i.e. odd linear generators $\kappa^1, \kappa^2$ and an even $q$ such that $\{ \pi_i, \kappa^j \}= \delta^j_i$, $\{ p, q \}=1$, $\{ \kappa^i, \kappa^j \}=0$. The following choice is convenient:
\begin{subequations}
\begin{align}
\kappa^1 =& \frac{1}{2} \left( \theta_2 + \xi_1 \right), \\
\kappa^2 =& \frac{1}{2} \left( \xi_2 - \theta_1 \right), \\
q =& \frac{1}{2} \left( x+y \right).
\end{align}
\end{subequations}
In terms of these variables, the generators become:
\begin{subequations}
\begin{align}
Q_1 =&\pi_1 + \frac{1}{4}(2 q - p) (\pi_1 - 2 \kappa^2),\\
Q_2 =& \pi_2 + \frac{1}{4}(2 q + p)(2 \kappa^1 - \pi_2), \\
H =& p + q^2 - \frac{1}{4} p^2   - \frac{1}{4}(\pi_1 - 2 \kappa^2)(2 \kappa^1 - \pi_2).
\end{align}
\end{subequations}
This completes the construction of the super classical Airy structure. It can be quantized as usual to get a super quantum Airy structure. For instance, using Weyl quantization, the generators become the following differential operators in the variables $\kappa^1, \kappa^2$ and $q$:
\begin{subequations}
\begin{align}
Q_1 =& \hbar \partial_{\kappa^1}  - q \kappa^2 +\frac{\hbar}{2} q \partial_{\kappa^1} + \frac{\hbar}{2} \kappa^2 \partial_q - \frac{\hbar^2}{4} \partial_q \partial_{\kappa^1},\\
Q_2=& \hbar \partial_{\kappa^2} + q \kappa^1 - \frac{\hbar}{2} q \partial_{\kappa^2} + \frac{\hbar}{2} \kappa^1 \partial_q - \frac{\hbar^2}{4} \partial_q \partial_{\kappa^2},\\
H =& \hbar \partial_q + q^2 - \kappa^1 \kappa^2 + \frac{\hbar}{2} \kappa^1 \partial_{\kappa^1} - \frac{\hbar}{2} \kappa^2 \partial_{\kappa^2} - \frac{\hbar^2}{4} \partial_q^2 + \frac{\hbar^2}{4} \partial_{\kappa^1} \partial_{\kappa^2}.
\end{align}
\end{subequations}
By construction, those differential operators have the form of a super quantum Airy structure, and they form a representation of the original superalgebra, with the only non-zero commutator $[Q_1, Q_2] = H$.
\end{enumerate}


\subsection{The $\mathfrak{osp}(1|2)$ example}

In this section we explain how we can construct quadratic super quantum Airy structures using representation theory of Lie superalgebras, following \cite{ABCD}. We then apply the procedure to the particular case of the $\mathfrak{osp}(1|2)$ Lie superalgebra.

Let $(V^*,[\,\cdot\,,\,\cdot\,])$ be a finite dimensional Lie superalgebra. A quadratic super classical Airy structure $L^{\rm cl}$ can be understood as a Lie superalgebra homomorphism: for any $\phi,\psi \in V^*$, and with $\{\,\cdot\,,\,\cdot\,\}$ denoting the canonical Poisson bracket on $W = V^*\oplus V$, we have
\begin{equation}
\label{super:Lie:homomorhism}
\left\{L^{\rm cl}(\phi),L^{\rm cl}(\psi)\right\} = L^{\rm cl}\left([\phi,\psi]\right).
\end{equation}
For a quadratic Airy structure $L^{\rm cl}(\phi) = \phi + {\mathcal L}(\phi),\;{\mathcal L}(\phi) \in {\mathbb K}[W]^2,$ (\ref{super:Lie:homomorhism})
is equivalent to
\begin{subequations}
\begin{gather}
\left\{{\mathcal L}(\phi),{\mathcal L}(\psi)\right\} = {\mathcal L}([\phi,\psi]), \label{quadratic:conditions}
\\[2pt]
\left\{\phi, {\mathcal L}(\psi)\right\}+ \left\{{\mathcal L}(\phi), \psi\right\} =[\phi,\psi]. \label{eq:J_cocycle}
\end{gather}
\end{subequations}
In particular ${\mathcal L}$ is itself a Lie superalgebra homomorphism, and the formula
\begin{equation}
\label{the:W:representation}
\rho_W(\phi)(w) = \left\{{\mathcal L}(\phi),w\right\}, \quad \mathrm{for} \quad \phi \in V^*, \quad w \in W,
\end{equation}
defines a representation of $(V^*,[\,\cdot\,,\,\cdot\,])$ on $W$.

Following \cite{ABCD}, we now show how one can use representation theory to construct a super classical Airy structure. Suppose that $\rho_M$ is a representation of the Lie superalgebra $V^*$ on a $2\dim V$ dimensional space $M$ equipped with a $V^*$ invariant symplectic form $\omega_M$. Suppose further that there exists $\Omega \in M$ such that $\rho_M(V^*)\cdot\Omega$ is a Lagrangian subspace of $M.$
Thus, in particular, $\dim \rho_M(V^*)\cdot\Omega = \dim V$ and the map
\begin{equation}
\label{I:map:definition}
I: V^* \ni \phi \; \mapsto \; I(\phi) = \rho_M(\phi)\cdot\Omega\ \in\ \rho_M(V^*)\cdot\Omega
\end{equation}
is an even isomorphism. Choose a Lagrangian complement $\Sigma$ of $\rho_M(\phi)\cdot \Omega$, and define
a map $K:\ \Sigma \to V^{**} \simeq V$ by the formula
\begin{equation}
\label{K:map:definition}
\left\{\phi, K(s)\right\} = \omega_M\left(\rho_M(\phi)\cdot\Omega,s\right), \quad \mathrm{for} \quad \phi \in V^*, \quad s \in \Sigma.
\end{equation}
Clearly $K$ is also an even isomorphism and consequently
\begin{equation}
\label{J:map:definition}
    J = I^{-1}\oplus K:\  M =  \rho_M(V^*)\cdot \Omega\oplus\Sigma\ \mapsto\  V^* \oplus V = W,
\end{equation}
is an even isomorphism as well.  $J$ is a symplectomorphism, which follows from the fact that the four (sub)spaces $V, V^*, \rho_M(V^*) \cdot \Omega$ and $\Sigma$ are Lagrangian, and for $s \in \Sigma, \phi \in V^*$ we have
\begin{equation}
\left\{J\left(\rho_M(\phi)\cdot\Omega\right),J(s)\right\} =  \left\{I^{-1}\left(I(\phi)\right),K(s)\right\} = \omega_M\left(\rho_M(\phi)\cdot\Omega,s\right)
\end{equation}
by (\ref{K:map:definition}). We can thus define a symplectic representation of $V^*$ on $W$ by the formula
\begin{equation}
\rho_W = J \circ \rho_M \circ J^{-1},
\end{equation}
and consequently, via (\ref{the:W:representation}), quadratic hamiltonians ${\mathcal L}(\phi)$ satisfying (\ref{quadratic:conditions}).

Equation (\ref{eq:J_cocycle}) is also satisfied, since
with our definitions (viewing $\phi,\psi \in V^*$ via the natural embedding $ V^* \simeq V^*\oplus 0 \hookrightarrow V^*\oplus V = W$ as elements of $W$), we have for homogeneous $\phi$ and $\psi$,
\begin{eqnarray}
\left\{{\mathcal L}(\phi), \psi\right\}+ \left\{\phi, {\mathcal L}(\psi)\right\}
& = &
\rho_W(\phi)(\psi) - (-1)^{|{\mathcal L}(\psi)||\phi|}\rho_W(\psi)(\phi)\nonumber
\\[2pt]
& = &
J\left(\rho_M(\phi)I(\psi)\right) - (-1)^{|{\mathcal L}(\psi)||\phi|}J\left(\rho_M(\psi)I(\phi)\right)\nonumber
\\[2pt]
& = &
J\left(\rho_M(\phi)\rho_M(\psi)\cdot\Omega\right)- (-1)^{|{\mathcal L}(\psi)||\phi|}J\left(\rho_M(\psi)\rho_M(\phi)\cdot\Omega\right)\nonumber
\\[2pt]
& = &
I^{-1}\left(\rho_M([\phi,\psi])\cdot\Omega \right) \; = \; [\phi,\psi].
\end{eqnarray}

So to summarize, the construction goes as follows:
\begin{enumerate}
\item We start with a Lie superalgebra $V^*$ and a representation $\rho_M(V^*)$ on some $2 \dim V$-dimensional space $M$;
\item We construct a symplectomorphism $J: M \to W = V^* \oplus V$, and a representation $\rho_W = J \circ \rho_M \circ J^{-1}$ on $W$;
\item By \eqref{the:W:representation}, this gives us a quadratic super classical Airy structure. It may be then quantized using Weyl prescription.
\end{enumerate}

Let us now apply this construction to the $\mathfrak{osp}(1|2)$ Lie superalgebra.
Recall that $\mathfrak{osp}(1|2)$  is generated by three even and two odd vectors, $l_0,l_{\pm}$ and $q_{\pm}$ respectively, satisfying the relations
\begin{eqnarray}
\nonumber
\label{osp:1:2}
[l_0,l_\pm] = \pm l_\pm && [l_+,l_-] = 2l_0.
\\[4pt]
[l_0,q_{\pm}] = \pm {\textstyle\frac12}q_{\pm}, && [l_\pm, q_\mp] = q_\pm,
\\[4pt]
\nonumber
[q_{\pm},q_{\pm}] = \pm l_\pm, && [q_+,q_-] = - l_0.
\end{eqnarray}

\begin{proposition}
Let $V$ be a $(3|2)$-dimensional super vector space, and let $\widetilde V = V \oplus \mathbb{K}^{0|1}$. We choose a basis $\{x^1, x^2, x^3 \}$ and $\{\theta^1, \theta^2\}$ for the even and odd subspaces of $V$ respectively, with dual basis $\{y_1, y_2, y_3\}$ and $\{ \xi_1, \xi_2 \}$. We let $\theta^0$ be a basis for $\mathbb{K}^{0|1}$. Then the linear operator $L: V^* \to \widehat{\mathcal{W}}(\tilde V)$ defined by
\begin{subequations}
\begin{align}
L(y_1)
 = &
\hbar\partial_{x^1}
- \sqrt{3}\,\theta^0\theta^1  - 12\left(x^1\right)^2
- \frac32\hbar x^2\partial_{x^1} - \frac{10}{3}\hbar x^3\partial_{x^2}  + 2\sqrt{3}\hbar\,\theta^1\partial_{\theta^0} + \hbar \theta^2\partial_{\theta^1},
\\
L(y_2)  = &
\hbar\partial_{x^2}
- \frac12\hbar x^1\partial_{x^1} - \frac32\hbar x^2\partial_{x^2} - \frac52\hbar x^3\partial_{x^3} - \hbar\theta^1\partial_{\theta^1} - 2\hbar\theta^2\partial_{\theta^2} - \frac34 \hbar,
\\
L(y_3)
 = &
\hbar\partial_{x^3}
- \frac{16}{3}\hbar x^1\partial_{x^2}- \frac{3}{2}\hbar x^2\partial_{x^3} + \frac{\sqrt{3}}{2}\hbar\,\theta^0\partial_{\theta^1} + 4\hbar\theta^1\partial_{\theta^2}
+ \frac{3}{16}\hbar^2\partial_{x^1}^2 - \sqrt{3}\,\hbar^2\partial_{\theta^1}\partial_{\theta^0},
\end{align}
\end{subequations}
and
\begin{subequations}
\begin{align}
L(\xi_1)
 = &
\hbar\partial_{\theta^1}
+
\sqrt{3}\, x_1\theta^0
-
\frac12\hbar\theta^1\partial_{x^1} + \frac13\hbar\theta^2\partial_{x^2}
+
2\sqrt{3}\hbar\, x^1\partial_{\theta^0} - \frac32\hbar x_2\partial_{\theta_1} + 5 \hbar x_3\partial_{\theta_2},
\\
L(\xi_2)
 = &
\hbar\partial_{\theta^2}
+
\frac{\sqrt{3}}{8} \hbar\theta^0\partial_{x^1} - \frac43\hbar\theta^1\partial_{x^2} + \frac12 \hbar\theta^2\partial_{x^3}
+
2\hbar x^1\partial_{\theta^1} - \frac32 \hbar x^2\partial_{\theta^2} + \frac{\sqrt{3}}{4} \hbar^2\partial_{\theta^0}\partial_{x^1},
\end{align}
\end{subequations}
is a quadratic super quantum Airy structure (with an extra fermionic variable), realized as a differential representation of the $\mathfrak{osp}(1|2)$ Lie superalgebra under the identification
\begin{equation}
\{ l_-, l_0, l_+, q_-, q_+ \} \mapsto  \{ L(y_1), L(y_2), L(y_3), L(\xi_1), L(\xi_2) \}.
\end{equation}
\end{proposition}

\begin{remark}
Clearly, the operator $L$ satisfies the second property in the definition of super quantum Airy structures (see Definition \ref{d:SQAS}), and it is quadratic. All that remains to be checked is that it forms a representation of the $\mathfrak{osp}(1|2)$ Lie superalgebra. This could be checked by brute force calculation. Let us instead construct this super quantum Airy structure using the general approach presented above.\end{remark}

\begin{proof}
We apply the construction described above for the particular case of the $\mathfrak{osp}(1|2)$ Lie superalgebra, with generators satisfying the relations \eqref{osp:1:2}. We let $V$ be a $(3|2)$-dimensional super vector space. We choose a basis $\{x^1, x^2, x^3 \}$ and $\{\theta^1, \theta^2\}$ for the even and odd subspaces of $V$ respectively, with dual basis $\{y_1, y_2, y_3\}$ and $\{ \xi_1, \xi_2 \}$. We identify $V^*$ with the algebra $\mathfrak{osp}(1|2)$ by:
\begin{equation}
\{y_1, y_2, y_3, \xi_1, \xi_2 \} \leftrightarrow \{ l_-, l_0, l_+, q_-, q_+ \}.
\end{equation}

The irreducible representations of $\mathfrak{osp}(1|2)$ are odd-dimensional. We will modify the construction above slightly, to allow our representation space $M$ to be odd-dimensional. We enlarge $W$ by adding an extra fermionic variable. We define $\widetilde V = V \oplus \mathbb{K}^{0|1}$, and  let $\widetilde{W} = \widetilde{V} \oplus \widetilde{V}^*$. We will construct a symplectomorphism $J$ as a map between the representation space $M$ and a suitably chosen, odd-dimensional subspace of the space $\widetilde W$. The result of the construction will be a quadratic super classical Airy structure with an extra fermionic variable.

Since $V$ is five-dimensional, we take $M$ to be 11-dimensional and denote its basis vectors by
$f_m$ and $e_k,$ where $m = -2,-1,\ldots,2$ and $k = -\frac52, -\frac32,\ldots, \frac52.$ The representation of $\mathfrak{osp}(1|2)$  in question is given by
\begin{equation}
\rho_M(\xi_2) f_m = \sqrt{\frac{3+m}{2}}\, e_{m+\frac12},
\hskip 1cm
\rho_M(\xi_1) f_m = \sqrt{\frac{3-m}{2}}\, e_{m-\frac12},
\end{equation}
and
\begin{equation}
\rho_M(\xi_2) e_k = \sqrt{\frac{\frac52 - k}{2}}\, f_{k+\frac12},
\hskip 1cm
\rho_M(\xi_1) e_k = -\sqrt{\frac{\frac52 + k}{2}}\, f_{k+\frac12},
\end{equation}
so that $\rho_M(\xi_2) e_{\frac52} = \rho_M(\xi_1) e_{-\frac52} = 0$. The action of other generators can be computed from (\ref{osp:1:2}).

The symplectic form on $M,$ invariant with respect to the $\mathfrak{osp}(1|2)$ algebra, i.e.\ satisfying
\begin{subequations}
\begin{gather}
\omega_M\left(\rho_M(y_i)e_k,e_l\right) + \omega_M\left(e_k,\rho_M(y_i)e_l\right) \; = \; \omega_M\left(\rho_M(y_i)f_m,f_n\right) +\omega_M\left(f_m,\rho_M(y_i)f_n\right)  =  0,
\\
\omega_M\left(\rho_M(\xi_j)e_k,f_n\right) + \omega_M\left(e_k,\rho_M(\xi_j)f_n\right)  =  0,
\end{gather}
\end{subequations}
is determined uniquely up to an overall normalization. Its non-zero elements read
\begin{subequations}
\begin{eqnarray}
\omega_M(e_{\frac12},e_{-\frac12}) \; = \; - \omega_M(e_{\frac32},e_{-\frac32}) \; = \; \omega_M(e_{\frac52},e_{-\frac52}) & = & 1,
\\[4pt]
\omega_M(f_0,f_0) \; = \; -\omega_M(f_1,f_{-1}) \; = \; \omega_M(f_2,f_{-2}) & = & 1,
\\[4pt]
\omega_M(f_i,f_j) - \omega_M(f_j,f_i) & = & 0,
\\[4pt]
\omega_M(e_i,e_j) +  \omega_M(e_j,e_i) & = & 0.
\end{eqnarray}
\end{subequations}

The maximal isotropic subspaces of $M$ are of dimension 5 and we can choose one of such subspaces to be generated by the action
of $\rho_m(V^*)$ on $e_{\frac32}$ (which thus plays the role of the vector $\Omega$ from the initial paragraphs of this subsection):
\begin{equation}
M_+ = \rho_M(V^*)e_{\frac32} = {\rm span}\Big\{e_{\frac12},f_1,e_{\frac32},f_2,e_{\frac52}\Big\}.
\end{equation}
Defining the map $I\ : \ V^* \mapsto M_+$ as in (\ref{I:map:definition}) we thus get
\begin{equation}
I(y_1) = 2\sqrt{2}\, e_{\frac12},
\hskip 5mm
I(\xi_1) = -\sqrt{2} f_1,
\hskip 5mm
I(y_2) = \frac32 \, e_{\frac32},
\hskip 5mm
I(\xi_2) = \frac{1}{\sqrt 2} f_2,
\hskip 5mm
I(y_3) = \sqrt{5}\, e_{\frac52}.
\end{equation}
As a complement of $M_+$, we take the space $M_0 \oplus M_-,$ where $M_0 = {\mathbb K}f_0$ and
\begin{equation}
M_- = {\rm span}\Big\{e_{-\frac12},f_{-1},e_{-\frac32},f_{-2},e_{-\frac12}\Big\}.
\end{equation}
Note that $M_-$ is also a maximal isotropic subspace of $M.$ We can now define the map $K\ : \ M_- \mapsto V$ as in (\ref{K:map:definition})
with the result
\begin{equation}
K(e_{-\frac12})= 2\sqrt{2}\,x^1,
\hskip .15cm
K(f_{-1}) = \sqrt{2}\, \theta^1,
\hskip .15cm
K(e_{-\frac32})=  -\frac32x^2,
\hskip .15cm
K(f_{-2}) = \frac{1}{\sqrt{2}}\, \theta^2,
\hskip .15cm
K(e_{-\frac52})=  \sqrt{5}\, x^3.
\end{equation}

Let us now denote the basis of the ${\mathbb K}^{0|1}$ subspace appearing in the decomposition $\widetilde V = V \oplus {\mathbb K}^{0|1}$ by $\theta^0$, and the corresponding
element of the dual basis in $\widetilde V^{*}$ by $\xi_0.$ Taking $\eta = \xi_0 + \frac12\theta^0$, so that
\begin{equation}
\{\eta,\eta\} = \{\xi_0,\theta^0\} = \xi_0(\eta^0) = 1 = \omega_M(f_0,f_0),
\end{equation}
and defining $J_0(f_0) = \eta$, we construct the symplectomorphism
\begin{equation}
J = I^{-1}\oplus J_0\oplus K :\ M_+ \oplus {\mathbb K}f_0 \oplus M_- \;\mapsto \; V^* \oplus {\mathbb K}\eta \oplus V
\end{equation}
satisfying the properties required for $\rho_{\widetilde W} = J^{-1} \circ \rho_M\circ J$ to define via (\ref{the:W:representation}) a quadratic super classical Airy structure
on ${\mathbb K}[\widetilde W]$ (depending on an extra fermionic variable). Calculating the matrix elements of $\rho_{\widetilde W}$, and using (\ref{the:W:representation}), we get
\begin{subequations}
\begin{eqnarray}
L^{\rm cl}(y_1) & = & y_1 -\frac32 x^2y_1 -\frac{10}{3}x^3 y_2 - 12\left(x^1\right)^2 +2\sqrt{3}\theta^1\eta + \theta^2\xi_1,
\\[2pt]
L^{\rm cl}(y_2) & = & y_2 - \frac12 x^1 y_1 - \frac32 x^2 y_2 - \frac52 x^3 y_3 - \theta^1\xi_1 -2 \theta^2\xi_2 - \frac32 \eta^2,
\\[2pt]
L^{\rm cl}(y_3) & = & y_3 + \frac{3}{16}\left(y_1\right)^2 - \frac{16}{3}x^1y_2 - \frac32 x^2 y_3 + \sqrt{3}\eta\xi_1 + 4\theta^1\xi_2,
\end{eqnarray}
\end{subequations}
as well as
\begin{subequations}
\begin{eqnarray}
L^{\rm cl}(\xi_1) & = & \xi_1 - \frac12 \theta^1 y_1 + \frac13 \theta^2 y_2 + 2\sqrt{3}x^1\eta - \frac32 x^2\xi_1 + 5 x^3\xi_2,
\\[4pt]
L^{\rm cl}(\xi_2) & = & \xi_2 - \frac43 \theta^1 y_2+ \frac12\theta^2 y_3 + \frac{\sqrt{3}}{4}\eta y_1 + 2 x^1\xi_1 - \frac32 x^2\xi_2.
\end{eqnarray}
\end{subequations}
Finally, applying Weyl quantization, we get
\begin{subequations}
\begin{align}
L(y_1)
 = &
\hbar\partial_{x^1}
- \sqrt{3}\,\theta^0\theta^1  - 12\left(x^1\right)^2
- \frac32\hbar x^2\partial_{x^1} - \frac{10}{3}\hbar x^3\partial_{x^2}  + 2\sqrt{3}\hbar\,\theta^1\partial_{\theta^0} + \hbar \theta^2\partial_{\theta^1},
\\
L(y_2)  = &
\hbar\partial_{x^2}
- \frac12\hbar x^1\partial_{x^1} - \frac32\hbar x^2\partial_{x^2} - \frac52\hbar x^3\partial_{x^3} - \hbar\theta^1\partial_{\theta^1} - 2\hbar\theta^2\partial_{\theta^2} - \frac34 \hbar,
\\
L(y_3)
 = &
\hbar\partial_{x^3}
- \frac{16}{3}\hbar x^1\partial_{x^2}- \frac{3}{2}\hbar x^2\partial_{x^3} + \frac{\sqrt{3}}{2}\hbar\,\theta^0\partial_{\theta^1} + 4\hbar\theta^1\partial_{\theta^2}
+ \frac{3}{16}\hbar^2\partial_{x^1}^2 - \sqrt{3}\,\hbar^2\partial_{\theta^1}\partial_{\theta^0},
\end{align}
\end{subequations}
and
\begin{subequations}
\begin{align}
L(\xi_1)
 = &
\hbar\partial_{\theta^1}
+
\sqrt{3}\, x_1\theta^0
-
\frac12\hbar\theta^1\partial_{x^1} + \frac13\hbar\theta^2\partial_{x^2}
+
2\sqrt{3}\hbar\, x^1\partial_{\theta^0} - \frac32\hbar x_2\partial_{\theta_1} + 5 \hbar x_3\partial_{\theta_2},
\\
L(\xi_2)
 = &
\hbar\partial_{\theta^2}
+
\frac{\sqrt{3}}{8} \hbar\theta^0\partial_{x^1} - \frac43\hbar\theta^1\partial_{x^2} + \frac12 \hbar\theta^2\partial_{x^3}
+
2\hbar x^1\partial_{\theta^1} - \frac32 \hbar x^2\partial_{\theta^2} + \frac{\sqrt{3}}{4} \hbar^2\partial_{\theta^0}\partial_{x^1}.
\end{align}
\end{subequations}

\end{proof}

We would like to make several remarks about the constructed super quantum Airy structure:
\begin{itemize}
\item It is not possible to construct an additional odd generator, say $Q$, which together with $ L(y_i) $ and $L(\xi_i)$ forms a super quantum Airy structure for some superalgebra of dimension $3|3$. Indeed, suppose that such an extension exists. Writing down the Jacobi identity for the extended algebra and using the fact that the first cohomology group of $\mathfrak{osp}(1|2)$ valued in the adjoint module vanishes, one can show that with no loss of generality $Q$ may be assumed to be central. Similar arguments show that $Q$ has to annihilate the whole $6|6$-dimensional module $M \oplus \mathbb K^{0|1}$. This means that we have to take the corresponding Hamiltonian to be purely linear. The only linear variable which supercommutes with all $L$ is $\xi_0 - \frac{1}{2} \theta^0$, which does not work because it is not nilpotent.
\item It is not possible to construct a super quantum Airy structure for $\mathfrak{osp}(1|2)$ which does not involve and additional fermionic variable. The reason is that there is no $6|4$-dimensional representation whose bosonic part is the irreducible $6$-dimensional representation of $\mathfrak{sl}(2, \mathbb C)$, which is known \cite{ABCD} to be necessary to construct an Airy structure for the even subalgebra. For similar reasons, the irreducible representation that we chose is the only one which allows the construction of a super quantum Airy structure.
\item In the language of the classification scheme for super quantum Airy structures that we have outlined in Section \ref{s:classification},
the fact that the Lagrangian embedding of $\mathfrak{osp}(1|2)$ into $M \oplus \mathbb K^{0|1}$ may be obtained by acting with the generators on some reference vector $\Omega$ is equivalent to the statement that the linear terms of our generators are trivial, i.e. can be obtained by an affine shift of coordinates. There is no need to consider more general linear terms, because the relevant cohomology groups vanish \cite{Kac:superalgebras}.
\item We have $[\mathfrak{osp}(1|2),\mathfrak{osp}(1|2)]= \mathfrak{osp}(1|2)$, so the quantization procedure is unique.
\end{itemize}


\subsection{Super Frobenius Algebras}

For our last finite-dimensional examples, we construct super quantum Airy structures from super Frobenius algebras, in the spirit of \cite{ABCD} (see also \cite{Notes by Gaetan}).

We define super Frobenius algebras following \cite{SFA}, where a definition of $G$-twisted Frobenius algebras is given. We concentrate on the $G=\mathbb{Z}_2$ case, which corresponds to super Frobenius algebras.

\begin{definition}
A super Frobenius algebra $\mathbb{A}_s=\mathbb{A}_0\oplus \mathbb{A}_1$ over $\mathbb{K}$ is a finite-dimensional $\mathbb{Z}_2$-graded vector space equipped with a super-commutative, associative product $\mathbb{A}_g\otimes \mathbb{A}_h\rightarrow \mathbb{A}_{gh}$ respecting grading, and a non-degenerate bilinear form $\phi:\mathbb{A}_g\otimes \mathbb{A}_h\rightarrow\mathbb{K}$ where $\phi=0$ unless $|g|+|h|=0$ and $g,h\in\{0,1\}$. 
\end{definition}

Given a super Frobenius algebra, for any choice of even elements $\theta_A, \theta_B, \theta_C \in \mathbb{A}_0$, we can construct a super quantum Airy structure as follows: 

\begin{lemma}\label{lem:SFA}
Let $\{e_i\}$ be a basis for a super Frobenius algebra $\mathbb{A}_s$, and $\{e^j\}$ be the dual basis, i.e.
\begin{equation}
\phi(e_i,e^j)=(-1)^{|i||j|}\phi(e^j,e_i)=\delta_i^j.\label{basis}
\end{equation}
Then, for any even $\theta_A,\theta_B,\theta_C\in\mathbb{A}_0$, the coefficients
\begin{equation}
A_{ijk}=\phi(\theta_Ae_ie_je_k),\;\;\;\;B_{ij}^k=\phi(\theta_Be_ie_je^k),\;\;\;\;C_i^{jk}=\phi(\theta_Ce_ie^je^k),\label{super Frobenius ABC}
\end{equation}
together with an arbitrary coefficient $D_i$, define a super quantum Airy structure on $V=\mathbb{A}_s$ with vanishing structure constants $f_{ij}^k=0$. \end{lemma}

\begin{remark}
We note that this super quantum Airy structure does not depend on an extra fermionic variable, i.e. the number of differential operators match with the dimension of the super vector space $\mathbb{A}_s$.
\end{remark}

\begin{proof}
To prove that this is a super quantum Airy structure, we will show that the coefficients satisfy the conditions of Lemma \ref{lem:SAS}.

Note that \eqref{basis} implies that every $a\in\mathbb{A}_s$ can be written as $a=\phi(a,e^i)e_i$. This gives
\begin{align}
B_{ik}^pA_{jpl}=&\phi(\theta_Be_ie_ke^p)\phi(\theta_Ae_je_pe_l)=\phi(\theta_Ae_j\phi(\theta_Be_ie_ke^p)e_pe_l)\nonumber\\
=&\phi(\theta_Ae_j\theta_Be_ie_ke_l)=\phi(\theta_A\theta_Be_je_ie_ke_l).
\end{align}
Thus, we have
\begin{align}
&B_{ik}^pA_{jpl}+(-1)^{|k||l|}B_{il}^pA_{jpk}+(-1)^{|i||j|}B_{ij}^pA_{pkl}\nonumber\\
&=\phi(\theta_A\theta_Be_je_ie_ke_l)+(-1)^{|k||l|}\phi(\theta_A\theta_Be_je_ie_le_k)+(-1)^{|i||j|}\phi(\theta_A\theta_Be_ie_je_ke_l)\nonumber\\
&=3\phi(\theta_A\theta_Be_je_ie_ke_l)=(-1)^{|i||j|}(i\leftrightarrow j),
\end{align}
and hence \eqref{BA} is satisfied.

Similarly, we find
\begin{align}
B_{ik}^pB_{jp}^l&=\phi(\theta_B^2e_je_ie_ke^l),\\
C_i^{lp}A_{jpk}&=\phi(\theta_A\theta_Ce_je_ie^le_k),\\
C_i^{kp}B_{jp}^l&=\phi(\theta_B\theta_Ce_je_ie^ke^l),\\
B_{ij}^pC_p^{kl}&=\phi(\theta_B\theta_Ce_ie_je^ke^l).
\end{align}
These ensure that all terms in the remaining conditions \eqref{BB-CA}, \eqref{CB} and \eqref{CA-BD} are $\mathbb{Z}_2$-symmetrical under $i\leftrightarrow j$, and hence the conditions of Lemma \ref{lem:SAS} are satisfied.
\end{proof}

A natural question then is to determine what this super quantum Airy structure associated to a super Frobenius algebra calculates. What is the meaning of the $F_{g,n}$?

In the standard, bosonic case, it is well known that a two-dimensional topological quantum field theory (2D TQFT) naturally defines the structure of a Frobenius algebra, and that, conversely, to any Frobenius algebra can be associated a unique 2D TQFT. In this case, it is shown in \cite{ABCD,Notes by Gaetan} that the quantum Airy structure naturally associated to a Frobenius algebra, with the choice $\theta_A=\theta_B=\theta_C=1$ and $D_i = \phi(e_i \cdot H)$ for $H = \sum_{j} e_j e^j$, solves the corresponding TQFT, in the sense that its $F_{g,n}$ compute the amplitudes of the 2D TQFT (up to a simple combinatorial factor).

In the general setting, it is known that a $G$-equivariant 2D TQFT defines the structure of a $G$-twisted Frobenius algebra, from which it can be recovered \cite{TQFT2}. Super Frobenius algebras correspond to the case with $G=\mathbb{Z}_2$. We expect that a story analogous to the bosonic case holds here as well, namely, that the $F_{g,n}$ associated to the super quantum Airy structure naturally constructed from a super Frobenius algebra with $\theta_A=\theta_B=\theta_C=1$ and choice of $D_i$ as above, compute the amplitudes of the corresponding 2D super TQFT. This correspondence should be made precise and investigated further.


\section{Infinite-Dimensional Examples}
\label{sec:VOSA}

In this section we construct examples of infinite-dimensional, quadratic, subalgebraic, super quantum Airy structures, as representations of subalgebras of super Virasoro algebras. The construction follows along the lines of \cite{BBCCN, Milanov}. In the bosonic case, many infinite-dimensional quantum Airy structures compute interesting enumerative invariants, such as intersection numbers over the moduli space of curves, Hurwitz numbers, Gromov-Witten invariants, etc. We expect the super quantum Airy structures that we construct in this section to also have interesting enumerative interpretations, which we leave for future work. In particular, they may be related to the recent supersymmetric generalization of JT gravity presented in \cite{Stanford:2019vob}.

In this section we focus on constructing quadratic super quantum Airy structures. But the construction can naturally be generalized to higher order super quantum Airy structures along the lines of \cite{BBCCN}, which we also leave for future work.

We do not review the well known definitions for vertex operator algebras (VOAs), vertex operator super algebras (VOSAs) and their representations. See \cite{BBCCN} for relevant definitions in the context of the construction presented here, and for instance \cite{BZF,VOSA0,VOSA1,VOSA2,VOSA3,VOSA4} for more details on VOAs and VOSAs.

\subsection{Quantum Airy Structures from The Free Boson VOA}

Before we study super quantum Airy structures constructed as representations of subalgebras of the super Virasoro algebra, let us review the standard bosonic construction for the Virasoro algebra. In the process we will generalize the construction of \cite{BBCCN}, discovering a VOA realization of the ``topological recursion without branched covers'' presented in \cite{ABCD}.

Our goal in this section is to construct quantum Airy structures as representations of subalgebras of the Virasoro algebra with central charge $c=1$:
\begin{equation}
[L_m, L_n] = (m-n)L_{m+n} + \delta_{m,-n} \frac{1}{12} m(m^2-1).
\end{equation}
This algebra arises as the algebra of modes for the energy-momentum tensor of the free boson VOA, which is central in our construction. So let us start by reviewing the main features of the free boson VOA.

\subsubsection{The Free Boson VOA}\label{s:freeboson}

We study the free boson VOA (also called the Heisenberg VOA). It is generated by a single vector $b_{-1} \ket{0} \in V$, where $V$ is the space of states. Here $\ket{0} \in V$ is the vacuum vector. The state-operator correspondence reads
\begin{equation}
b(z) := Y(b_{-1} \ket{0}, z) = \sum_{m \in \mathbb{Z}} b_m z^{-m-1},
\end{equation}
where the modes of the bosonic field generate the Heisenberg algebra
\begin{equation}
[b_m, b_n] = m \delta_{m,-n} .
\end{equation}
The vacuum vector $\ket{0}$ is annihilated by all $b_k$ with $k \geq 0$, and the space of states $V$ is the Fock space consisting of all excited modes
\begin{equation}
b_{-k_1} \cdots b_{-k_n} \ket{0}, \qquad k_1, \ldots, k_n \in \mathbb{Z}_{>0}.
\end{equation}
Their corresponding operators are
\begin{equation}\label{eq:genop}
Y(b_{-k_1}\cdots b_{-k_n}\ket{0},z)=\altcolon\prod_{i=1}^n\frac{1}{(k_i-1)!}\left(\frac{d}{dz}\right)^{k_i-1}Y(b_{-1}\ket{0},z)\altcolon,
\end{equation}
where $\altcolon\cdots\altcolon$ denotes normal ordering, i.e., all modes $b_k$ with negative $k$ are on the left and those with positive $k$ are on the right.

The conformal vector $\ket{\omega}$ for the free boson VOA is:
\begin{equation}
\ket{\omega} = \frac{1}{2} b_{-1} b_{-1} \ket{0}.
\end{equation}
Its operator takes the form
\begin{equation}
T(z) := Y(\ket{\omega}, z) = \sum_{m \in \mathbb{Z}} L_m z^{-m-2},
\end{equation}
with its modes generating the Virasoro algebra with central charge $c=1$:
\begin{equation}\label{eq:VIR}
[L_m, L_n] = (m-n)L_{m+n} + \delta_{m,-n} \frac{1}{12} m(m^2-1).
\end{equation}
The modes of the conformal field can be related to the Heisenberg modes as follows. From \eqref{eq:genop}, we have:
\begin{equation}
T(z) =  \frac{1}{2}  Y(b_{-1} b_{-1} \ket{0}) = \frac{1}{2} \altcolon Y(b_{-1} \ket{0}, z) Y(b_{-1} \ket{0}, z) \altcolon = \frac{1}{2} \altcolon b(z) b(z) \altcolon,
\end{equation}
and hence
\begin{equation}\label{eq:Lboson}
L_m =  \frac{1}{2}  \sum_{k \in \mathbb{Z}} \altcolon b_k b_{m-k} \altcolon.
\end{equation}

\subsubsection{Untwisted and Twisted Representations of the Free Boson VOA}

To construct quantum Airy structures as representations of subalgebras of the Virasoro algebra, we will start with two different representations of the free boson VOA: an untwisted representation, and a $\mathbb{Z}_2$-twisted representation (see for instance \cite{BBCCN} for more details on twisted representations for VOAs in the context of Airy structures).

In both cases, we will represent the bosonic modes $b_m$ as endomorphisms of the space $\mathbb{K}[[V,\hbar]]$, where $V$ is an infinite-dimensional vector space. In other words, we represent the modes of the Heisenberg algebra as differential operators in $\widehat{\mathcal{W}}_\hbar(V)$, which turns the Virasoro modes into differential operators in $\widehat{\mathcal{W}}_\hbar(V)$ as well.

The untwisted representation is basically the VOA itself. The state-field correspondence for the representation is 
\begin{equation}
Y^M(b_{-1} \ket{0}, z) = \sum_{m \in \mathbb{Z}} b_m^M z^{-m-1},
\end{equation}
with the $b^M_m$ endomorphisms of $\mathbb{K}[[V,\hbar]]$. The Virasoro modes of the untwisted representation then take the form:
\begin{equation}\label{eq:Luntwisted}
L_m^M =  \frac{1}{2}  \sum_{k \in \mathbb{Z}} \altcolon b_k^M b^M_{m-k} \altcolon.
\end{equation}

For the $\mathbb{Z}_2$-twisted representation, we consider the order two automorphism that acts on the Fock space as:
\begin{align}
\sigma:& V \to V \nonumber\\
& b_{-k_1} \cdots b_{-k_n} \ket{0} \mapsto (-1)^{\sum_{i=1}^n k_i} b_{-k_1} \cdots b_{-k_n} \ket{0}.
\end{align}
This automorphism preserves the vacuum vector $\ket{0}$ and the conformal vector $\ket{\omega} = \frac{1}{2} b_{-1} b_{-1} \ket{0}$. However, it does not preserve the fundamental vector $b_{-1} \ket{0}$, which picks a sign.

Thus the state-field correspondence for this $\mathbb{Z}_2$-twisted representation takes the form
\begin{equation}\label{eq:twistboson}
Y^\sigma(b_{-1} \ket{0},z) = \sum_{r \in \mathbb{Z} + \frac{1}{2} } b^\sigma_r z^{-r-1},
\end{equation}
with fractional exponents, and twisted modes that generate the Heisenberg algebra:
\begin{equation}
[ b_r^\sigma, b_s^\sigma] = r \delta_{r,-s}.
\end{equation}

The Virasoro modes $L_m^\sigma$ of the $\mathbb{Z}_2$-twisted representation are related to the twisted bosonic modes $b_r^\sigma$ as follows:

\begin{lemma}\label{lem:sigma}
Consider the $\sigma$-twisted representation of the free boson VOA, with state-field correpondence \eqref{eq:twistboson}. The conformal field
\begin{equation}
Y^\sigma(\ket{\omega}, z) = \frac{1}{2} Y^\sigma( b_{-1} b_{-1} \ket{0}, z) = \sum_{m \in \mathbb{Z}} L_m^\sigma z^{-m-2}
\end{equation}
has Virasoro modes given by
\begin{equation}\label{eq:Ltwisted}
L_m^\sigma = \frac{1}{2} \sum_{r \in \mathbb{Z} + \frac{1}{2}} \altcolon b_r^\sigma b_{m-r}^\sigma \altcolon + \frac{1}{16} \delta_{m,0}.
\end{equation}
\end{lemma}

\begin{proof}
The proof goes as Lemma 4.2 of \cite{BBCCN}.
The conformal field is
\begin{equation}
Y^\sigma(\ket{\omega}, z) = \frac{1}{2} Y^\sigma( b_{-1} b_{-1} \ket{0}, z) .
\end{equation}
We want to rewrite it as a normal ordered product of twisted bosonic fields. For clarity, let us denote the twisted bosonic field by $b^\sigma(z) := Y^\sigma(b_{-1} \ket{0},z)$. 

For this we can use the product formula for twisted representations (see for instance Definition 3.18 in \cite{BBCCN}):
\begin{equation}\label{eq:prodform}
Y^\sigma(b_{-1} b_{-1} \ket{0}, z) = \frac{1}{2} \left[ \frac{\partial^2}{\partial z_1^2} \left( (z_1 - z_2)^2 b^\sigma(z_1) b^\sigma(z_2) \right) \right]_{z_1=z_2=z}.
\end{equation}
We have:
\begin{align}
b^\sigma(z_1) b^\sigma(z_2) =& \altcolon b^\sigma(z_1) b^\sigma(z_2) \altcolon + \sum_{\substack{k_1, k_2 \in \mathbb{Z} + \frac{1}{2} \\ k_1>0, k_2<0}} [b_{k_1}^\sigma, b_{k_2}^\sigma] z_1^{-k_1-1} z_2^{-k_2-1} \nonumber \\
=&  \altcolon b^\sigma(z_1) b^\sigma(z_2) \altcolon + \sum_{\substack{k_1 \in \mathbb{Z} + \frac{1}{2} \\ k_1>0}} k_1 z_1^{-k_1-1} z_2^{k_1-1} \nonumber \\
=&  \altcolon b^\sigma(z_1) b^\sigma(z_2) \altcolon + \sum_{k \in \mathbb{Z}_{\geq 0}} \frac{\partial}{\partial z_2} \left( \frac{z_2^{k }}{z_1^{k+1}}  \sqrt{\frac{z_2}{z_1}}\right) \nonumber\\
=&  \altcolon b^\sigma(z_1) b^\sigma(z_2) \altcolon + \frac{\partial}{\partial z_2} \left( \frac{1}{z_1-z_2}  \sqrt{\frac{z_2}{z_1}}\right).
\end{align}
Substituting back in \eqref{eq:prodform}, we get:
\begin{align}
Y^\sigma(b_{-1} b_{-1} \ket{0}, z) =&  \frac{1}{2} \left[ \frac{\partial^2}{\partial z_1^2} \left( (z_1 - z_2)^2 \altcolon b^\sigma(z_1) b^\sigma(z_2) \altcolon +  (z_1 - z_2)^2 \frac{\partial}{\partial z_2} \left( \frac{1}{z_1-z_2}  \sqrt{\frac{z_2}{z_1}}\right) \right) \right]_{z_1=z_2=z} \nonumber\\
=&\altcolon b^\sigma(z) b^\sigma(z) \altcolon + \frac{1}{8 z^2},
\end{align}
where the second line follows from a straightforward calculation. Therefore,
\begin{equation}
Y^\sigma(\ket{\omega},z) = \frac{1}{2} \altcolon b^\sigma(z) b^\sigma(z) \altcolon + \frac{1}{16 z^2} = \sum_{m \in \mathbb{Z}} L_m^\sigma z^{-m-2}.
\end{equation}
Combining with \eqref{eq:twistboson}, we extract the relation between the modes for the $\mathbb{Z}_2$-twisted representation:
\begin{equation}\label{eq:Ltwisted}
L_m^\sigma = \frac{1}{2} \sum_{r \in \mathbb{Z} + \frac{1}{2}} \altcolon b_r^\sigma b_{m-r}^\sigma \altcolon + \frac{1}{16} \delta_{m,0}.
\end{equation}
\end{proof}

\subsubsection{Quantum Airy Structures from Untwisted Representations of the Free Boson VOA}

Let us now construct quantum Airy structures using these representations of the free boson VOA.
We proceed in three steps:
\begin{enumerate}
\item We choose a representation (untwisted or twisted) of the free boson VOA to obtain a differential representation for the Virasoro modes:  $L_n \in \widehat{\mathcal{W}}_\hbar(V)$. As we will see, the operators $\hbar L_n$ have degree $2$ (with the notion of degree defined in \eqref{degree}). We thus have constructed a representation of the Virasoro algebra by quadratic differential operators $\hbar L_i$. 
\item We pick a subalgebra of the Virasoro algebra. Our quantum Airy structure will be obtained as a representation of this particular Virasoro subalgebra.
\item We shift some of the bosonic modes to create linear terms to the operators $\hbar L_i$, without changing the algebra, so that the second condition of Definition \ref{d:SQAS} is satisfied. This constructs a representation of the chosen Virasoro subalgebra in the form of a quantum Airy structure. 
\end{enumerate}

For convenience, let us define the following notation. For any integer $R$, we define
\begin{equation}
\delta_{i \leq R} = \begin{cases} 1 & \text{if $i \leq R$} \\ 0 & \text{ if $i > R$.} \end{cases}
\end{equation}

From untwisted representations of the free boson VOA, we construct the following three classes of quantum Airy structures:

\begin{proposition}\label{prop:untwisted}
We represent the Heisenberg algebra as 
\begin{equation}
\forall k \in \mathbb{Z}_{\geq 1}, \qquad b_k^M = \sqrt{\hbar} \partial_k, \qquad b_{-k}^M =\frac{1}{\sqrt{\hbar}} k x^k, \qquad b_0^M = \sqrt{\hbar} \partial_0.
\end{equation}
Let $\{x^0, x^1, x^2, \ldots \}$ be a basis for $V$, with dual set $\{y_0, y_1, y_2, \ldots \}$. Define the differential operators $H_i \in \widehat{\mathcal{W}}_{\hbar} (V)$, $i \in \mathbb{Z}_{\geq 0}$:
\begin{equation}
H_i =  \hbar \partial_i + \frac{\hbar}{2}  \sum_{k \in \mathbb{Z}} \altcolon b_k^M b^M_{i+N-1-k} \altcolon ,
\end{equation}
which form a representation of the Virasoro subalgebra
\begin{equation}\label{eq:untwistedC1}
[ H_i, H_j] = \hbar (i-j) H_{i+j+N-1}.
\end{equation}

\begin{enumerate}
\item Let $N$ be any integer $N \geq 0$. The linear operator $H: V^* \to \widehat{\mathcal{W}}_{\hbar} (V)$ defined by:
\begin{equation}\label{eq:untwistedH1}
\forall i \in \mathbb{Z}_{\geq 0}, \qquad H(y_i) = H_i +  \hbar D_i \delta_{i \leq N-1},
\end{equation}
for arbitrary constant $D_i$, $i=0,\ldots,N-1$, forms a quantum Airy structure as a representation of the Virasoro subalgebra \eqref{eq:untwistedC1}.
\item Let $N$ be any integer $N \geq -1$. The linear operator $H: V^* \to \widehat{\mathcal{W}}_{\hbar} (V)$ defined by
\begin{subequations}\label{eq:untwistedH2}
\begin{align}
\forall i \in \mathbb{Z}_{\geq 1}, \qquad H(y_i) =& H_i  + \hbar D_i \delta_{i \leq N+1}, \\
H(y_0) =& \hbar \partial_0 + \frac{\hbar^2}{2} C_0 \partial_0^2 + \hbar D_0,
\end{align}
\end{subequations}
for arbitrary constants $D_i$, $i=0,\ldots,N+1$ and $C_0$, forms a quantum Airy structure as a representation of the Virasoro subalgebra \eqref{eq:untwistedC1} extended by:
\begin{equation}\label{eq:untwistedC2}
[H(y_0), H(y_i)] = 0.
\end{equation}
\item Let $N$ be any integer $N \geq -1$. Let us now formally set the bosonic zero mode $b_0^M = 0$. Consider the subspace $V_{red} \subset V$ spanned by $\{x^1, x^2, \ldots \}$. The linear operator $H: V^*_{red} \to \widehat{\mathcal{W}}_{\hbar} (V_{red})$ defined by
\begin{equation}\label{eq:untwistedH3}
\forall i \in \mathbb{Z}_{\geq 1}, \qquad H(y_i) =  H_i \Big|_{b_0^M = 0} + \hbar D_i \delta_{i \leq N+1},
\end{equation}
for arbitrary constants $D_i$, $i=1,\ldots,N+1$, forms a quantum Airy structure as a representation of the Virasoro subalgebra \eqref{eq:untwistedC1}.

\end{enumerate}

\end{proposition}

\begin{proof}
We start with the untwisted representation of the free boson VOA. For our first two classes of quantum Airy structures, we represent the bosonic modes as endomorphisms on the space $M=\mathbb{K}[[V,\hbar]]$, where $V$ is an infinite-dimensional vector space with basis $\{x^0, x^1, x^2, \ldots \}$, as:
\begin{equation}
\forall k \in \mathbb{Z}_{\geq 1}, \qquad b_k^M = \sqrt{\hbar} \partial_k, \qquad b_{-k}^M =\frac{1}{\sqrt{\hbar}} k x^k, \qquad b_0^M = \sqrt{\hbar} \partial_0,
\end{equation}
which form a representation of the Heisenberg algebra. From \eqref{eq:Luntwisted}, the Virasoro modes read:
\begin{equation}\label{eq:Lun1}
\hbar L_m^M =  \frac{\hbar}{2}  \sum_{k \in \mathbb{Z}} \altcolon b_k^M b^M_{m-k} \altcolon.
\end{equation}
We multiplied the Virasoro modes by $\hbar$ here so that they have degree $2$ according to the grading \eqref{degree} on $\widehat{\mathcal{W}}_\hbar(V)$. This rescales the Virasoro algebra by $\hbar$.

We now choose the following infinite sequence of subalgebras of the Virasoro algebra:
\begin{equation}
[L_m^M, L_n^M] =  (m-n) L_{m+n}^M, \qquad m,n \geq N,
\end{equation}
for arbitrary fixed integer $N \geq -1$. We will construct quantum Airy structures for each choice of such subalgebra, using the representation of the Virasoro algebra \eqref{eq:Lun1}. To do so, we need to bring the operators $\hbar L_m^M$ in the form of quantum Airy structures by creating appropriate linear terms.

For our first class of quantum Airy structures, we shift indices so that for any choice of subalgebra (choice of $N$) our operators are indexed by integers $i=0,1,2,\ldots$.
Thus, given an integer $N \geq -1$, we define the quadratic hamiltonians:
\begin{equation}\label{eq:quadha}
\forall i \in \mathbb{Z}_{\geq 0}, \qquad H_i^2 := \hbar L_{i+N}^M =  \frac{\hbar}{2}  \sum_{k \in \mathbb{Z}} \altcolon b_k^M b^M_{i+N-k} \altcolon,
\end{equation}
which have commutation relations:
\begin{equation}\label{eq:virLcheck}
[ H_i^2, H_j^2] = \hbar (i-j) H_{i+j+N}^2.
\end{equation}
To create appropriate linear terms, we do the shift $b_N \mapsto b_N + \frac{1}{\sqrt{\hbar}}$, which creates a linear term $H_i^1 = \sqrt{\hbar} b_i = \hbar \partial_i$ in the hamiltonians without changing the commutation relations. It however also creates an unwanted constant term $H_N^0 = \frac{1}{2}$, but since $H_N$ never appears in the right-hand-side of the commutation relations \eqref{eq:virLcheck}, we can get rid of this constant term without changing the algebra. We thus obtain hamiltonians:
\begin{equation}
\forall i \in \mathbb{Z}_{\geq 0}, \qquad H_i =  \hbar \partial_i + \frac{\hbar}{2}  \sum_{k \in \mathbb{Z}} \altcolon b_k^M b^M_{i+N-k} \altcolon,
\end{equation}
which have the form of a quantum Airy structure. However, this quantum Airy structure has no ``$A$'' or ``$D$'' terms (see \eqref{quadratic:Ls}), and hence its associated partition function $Z$ is trivial: that is, $Z=1$ (see \eqref{F(i)}). For $N=-1$, there is not much that we can do. However, for $N \geq 0$ we can make the partition function non-trivial. Looking at the commutation relations \eqref{eq:virLcheck}, we see that the operators $H_i$, $i=0,\ldots,N$ never appear on the right-hand-side of the commutation relations. Thus we can add quadratic ``$D$'' terms of the form $\hbar D_i$ to those operators without changing the commutation relations (since these terms commute with all $H_i$). We thus obtain our first class of quantum Airy structures. For convenience, in this case we redefine $N \mapsto N-1$, so that the operators take the same form as for the other cases. We then have $N \geq 0$. We obtain the quantum Airy structure \eqref{eq:untwistedH1}, as a representation of the algebra \eqref{eq:untwistedC1}.

For our second class of quantum Airy structures, we shift indices so that for any $N$ our operators are indexed by integers $i=1,2,\ldots$:
\begin{equation}\label{eq:quadha}
\forall i \in \mathbb{Z}_{\geq 1}, \qquad H_i^2 := \hbar L_{i+N-1}^M =  \frac{\hbar}{2}  \sum_{k \in \mathbb{Z}} \altcolon b_k^M b^M_{i+N-1-k} \altcolon,
\end{equation}
with commutation relations:
\begin{equation}
[ H_i^2, H_j^2] = \hbar (i-j) H_{i+j+N-1}^2.
\end{equation}
To create linear terms, we do the shift $b_{N-1} \mapsto b_{N-1} + \frac{1}{\sqrt{\hbar}}$, which creates a linear term $H_i^1 = \sqrt{\hbar} b_i = \hbar \partial_i$ in the hamiltonians without changing the commutation relations. We get rid of the unwanted constant term $H_{N-1}^0 = \frac{1}{2}$ as above, and introduce non-trivial $D$-terms $\hbar D_i$ for $i=1,\ldots,N+1$, since the corresponding $H_i$, $i=1,\ldots,N+1$ do not appear on the right-hand-side of the commutation relations. This is not quite an Airy structure though, since the hamiltonians depend on the variable $x^0$, which does not appear in the linear terms. We notice however that the $H_i$ only depend on the derivative $\partial_0$ through the bosonic zero mode $b_0^M$, they do not depend on the variable $x^0$ itself. Thus, we can introduce an auxiliary operator
\begin{equation}
\mathcal{H}_0 = \hbar \partial_0 + \frac{\hbar^2}{2} C_0 \partial_0^2 + \hbar D_0,
\end{equation}
which commutes with all $H_i$, $i \geq 0$, to get a quantum Airy structure. The result is the quantum Airy structure \eqref{eq:untwistedH2}, as a representation of the algebra \eqref{eq:untwistedC2}.

 For our third class of quantum Airy structures, we modify the representation of the Heisenberg algebra slightly, and set the bosonic zero mode $b_0^M = 0$. That is, we consider a ``momentum zero'' representation. To get a quantum Airy structure, we do the same manipulations as for the second class above. We end up with the quantum Airy structure \eqref{eq:untwistedH3}, as a representation of the algebra \eqref{eq:untwistedC1}.

\end{proof}

Before we proceed, a few remarks are in order.

\begin{remark}
Note that from the point of view of partition functions, the third class can be understood as a special case of the second class. Indeed, consider the second class of quantum Airy structures with $C_0 = D_0 = 0$. The constraint $\mathcal{H}_0 Z = 0$ implies that $Z$ does not depend on the variable $x^0$. Then any term in the constraints $H_i Z = 0$ that involves $b_0^M$ identically vanishes, and hence the constraints $H_i Z = 0$ become identical to the constraints of the third class of quantum Airy structures. Therefore, the partition function of the third class is equal to the partition function of the second class with $C_0 = D_0 = 0$.
\end{remark}

\begin{remark}
In the construction above we introduced the shifts $\hbar D_i$ as defining a new representation for the Virasoro subalgebra. We could however take a different viewpoint. We could stick with the representation without the added $D$-terms, and then, instead of solving the constraints $H_i Z = 0$ to define the partition function, we could solve the constraints
\begin{equation}
H_i Z = - \hbar D_i Z,
\end{equation}
with possibly non-vanishing $D_i$ in the allowed range (for instance $i=0,\ldots,N-1$ in the first class). This is entirely equivalent to what we did above, but from the point of view of vertex operator algebras, by solving these constraints we would be constructing so-called ``Whittaker modules'' for the Virasoro algebra. We are currently investigating this connection further.
\end{remark}

It turns out that the quantum Airy structures in the third class of Proposition \ref{prop:untwisted} have appeared in the literature before, in a different form. Indeed, in Section 10 of \cite{ABCD}, a variant of the Eynard-Orantin topological recursion was constructed, the so-called ``topological recursion without branched covers''. We now show that the quantum Airy structures of Proposition \ref{prop:untwisted} are examples of this topological recursion without branched covers. We thus obtain a realization of the topological recursion without branched covers in terms of representations of the free boson VOA, which sheds light on its origin. It remains to be seen however whether the associated partition functions compute interesting enumerative invariants, and whether the other quantum Airy structures of Proposition \ref{prop:untwisted} have similar realizations.

\begin{proposition}
The data of a spectral curve for the topological recursion without branched covers of Section 10 of \cite{ABCD} is given by:
\begin{itemize}
\item A Riemann surface $\Sigma$;
\item A meromorphic one-form $\omega_{0,1}$ on $\Sigma$;
\item A bilinear differential $\omega_{0,2}$ on $\Sigma^2$;
\item A finite subset $\mathfrak{r} \subset \Sigma$ such that $\omega_{0,1}$ has at most simple zeros on $\mathfrak{r}$;
\item A meromorphic one-form $\omega_{1,1}$ on $\Sigma$ such that for any point $p \in \mathfrak{r}$, $z^2 \frac{\omega_{1,1}(z)}{\omega_{0,1}(z)}$ is holomorphic at $p$, where $z$ is a local coordinate near $p$.
\end{itemize}
Consider the spectral curve given by the data:
\begin{equation}
\Sigma=\mathbf{P}^1,\;\;\;\;\omega_{0,1}(z)=-\frac{dz}{z^{N}},\;\;\;\;\omega_{0,2}(z_1,z_2)=\frac{dz_1dz_2}{(z_1-z_2)^2},\;\;\;\;\mathfrak{r} = \{0\}, \;\;\;\;  \omega_{1,1}(z)=-\sum_{k=1}^{N+1}D_k\frac{dz}{z^{k+1}},\label{TR for no branch cover}
\end{equation}
where $z$ is a coordinate on $\mathbb{P}^1$ and $N \in \mathbb{Z}_{\geq 0}$. Then, for any $N \geq 0$, the quantum Airy structure associated to the topological recursion without branched cover on this spectral curve, which was constructed in Section 10 of \cite{ABCD}, precisely corresponds to the quantum Airy structure in the third class of Proposition \ref{prop:untwisted}. 
\end{proposition}

\begin{proof}
Following \cite{ABCD}, we define, for $k \in \mathbb{Z}_{\geq 1}$,\footnote{Note that our definition for $\xi_k$ and $\xi^*_{k+1}$ is rescaled by $(k+1)^{-1}$ and $k$ respectively with respect to \cite{ABCD}. }
\begin{equation}\label{eq:basisTR}
\xi_k(z) = \frac{dz}{z^{k+1}}, \qquad \xi^*_k(z) = z^k, \qquad \theta(z) := \frac{1}{\omega_{0,1}(z)}=-\frac{z^{N}}{dz}.
\end{equation}

According to the recipe of Section 10 in \cite{ABCD} (see also Proposition 8.13), the coefficients of the quantum Airy structure associated to topological recursion without branched covers can be calculated as follows. The $D$-coefficients can be obtained by expanding the one-form $\omega_{1,1}(z)$ in the basis of differentials $\xi_k(z)$:
\begin{equation}
\omega_{1,1}(z) = - \sum_{k=1}^{N+1} D_k \xi_k(z).
\end{equation}
As for $A,B,C$, they can be calculated as:
\begin{align}
A_{ijk}&=\underset{z\rightarrow 0}{\text{Res}}(\xi^*_i(z)d\xi_j^*(z)d\xi^*_k(z)\theta(z)),\\
B_{ij}^{k}&=\underset{z\rightarrow0}{\text{Res}}(\xi^*_i(z)d\xi_j^*(z)\xi_k(z)\theta(z)),\\
C_{i}^{jk}&=\underset{z\rightarrow0}{\text{Res}}(\xi^*_i(z)\xi_j(z)\xi_k(z)\theta(z)).
\end{align}
Substituting \eqref{eq:basisTR} in these equations, we obtain:
\begin{equation}
A_{ijk}=0,\;\;\;\;B_{ij}^k=-j\delta_{i+j+N-1,k},\;\;\;\;C_i^{jk}=-\delta_{i+N-j-1,k}.
\end{equation}
Thus, the resulting differential operators are, for $i \in \mathbb{Z}_{\geq 1}$,
\begin{equation}
H_{i}=\hbar\frac{\partial}{\partial x^i}+\hbar\sum_{k-l=-i-N+1} k x^k \partial_l +\frac{\hbar^2}{2}\sum_{k+l=i+N-1} \partial_k \partial_l +\hbar D_i\delta_{i\leq N+1},
\end{equation}
with $k,l \in \mathbb{Z}_{\geq 1}$, which can be rewritten as
\begin{equation}
H_i =  \hbar \partial_i + \frac{\hbar}{2}  \sum_{k \in \mathbb{Z}} \altcolon b_k^M b^M_{i+N-1-k} \altcolon + \hbar D_i \delta_{i \leq N+1},
\end{equation}
with
\begin{equation}
b_k^M = \sqrt{\hbar} \partial_k, \qquad b_{-k}^M =\frac{1}{\sqrt{\hbar}} k x^k, \qquad b_0^M = 0,
\end{equation}
for all $k \in \mathbb{Z}_{\geq 1}$. Those are precisely the differential operators $H_i$ of Proposition \ref{prop:untwisted}.

\end{proof}

\subsubsection{Quantum Airy Structures from $\mathbb{Z}_2$-Twisted Representations of the Free Boson VOA}

We now construct another class of quantum Airy structures obtained from the $\sigma$-twisted representation of the free boson VOA (see Lemma \ref{lem:sigma}). This is an example of the general construction of \cite{BBCCN} for $W$-algebras, although only the cases with $N=-1$ and $N=0$ (and their generalizations to $W$-algebras) were considered there.

\begin{proposition}\label{prop:twisted}
Let $N$ be any fixed integer $N \geq -1$. We represent the Heisenberg algebra for the twisted bosonic modes as:
\begin{equation}
\forall r \in \mathbb{Z}_{\geq 0} + \frac{1}{2}, \qquad b_r^\sigma = \sqrt{\hbar} \partial_{r+\frac{1}{2}}, \qquad b_{-r}^\sigma = \frac{r}{\sqrt{\hbar}} x^{r + \frac{1}{2}}.
\end{equation}
Let $\{x^1, x^2, x^3, \ldots \}$ be a basis for $V$, with dual set $\{y_1, y_2, y_3, \ldots \}$. The linear operator $H: V^* \to \widehat{\mathcal{W}}_{\hbar} (V)$ defined by:
\begin{equation}\label{eq:twistedH1}
\forall i \in \mathbb{Z}_{\geq 1}, \qquad H_i:= H(y_i) = \hbar \partial_i + \frac{\hbar}{2} \sum_{r \in \mathbb{Z} + \frac{1}{2}} \altcolon b_r^\sigma b_{i+N-1-r}^\sigma \altcolon + \frac{\hbar}{16} \delta_{i,1-N} + \hbar D_i \delta_{i \leq N+1},
\end{equation}
for arbitrary constants $D_i$, $i=1,\ldots,N+1$, forms a quantum Airy structure as a representation of the Virasoro subalgebra
\begin{equation}\label{eq:twistedC1}
[H_i, H_j] = \hbar (i-j) H_{i+j+N-1}.
\end{equation}

\end{proposition}

\begin{proof}
We start with the $\sigma$-twisted representation of the free boson VOA. The twisted bosonic modes, see \eqref{eq:twistboson}, form the Heisenberg algebra $[b_r^\sigma, b_s^\sigma] = r \delta_{r,-s}$, where $r,s \in \mathbb{Z} + \frac{1}{2}$. We represent the twisted bosonic modes as endomorphisms on the space $\mathbb{K}[[V,\hbar]]$ as:
\begin{equation}
\forall r \in \mathbb{Z}_{\geq 0} + \frac{1}{2}, \qquad b_r^\sigma = \sqrt{\hbar} \partial_{r+\frac{1}{2}}, \qquad b_{-r}^\sigma = \frac{r}{\sqrt{\hbar}} x^{r + \frac{1}{2}}.
\end{equation}
Note that there is no choice of bosonic zero mode here, because the bosonic field is twisted by the $\mathbb{Z}_2$ automorphism. 
From \eqref{eq:Ltwisted}, the Virasoro modes take the form:
\begin{equation}
L_m^\sigma = \frac{1}{2} \sum_{r \in \mathbb{Z} + \frac{1}{2}} \altcolon b_r^\sigma b_{m-r}^\sigma \altcolon + \frac{1}{16} \delta_{m,0} 
\end{equation}

Now as in Proposition \ref{prop:untwisted}, we choose the subalgebra
\begin{equation}
[L_m^\sigma, L_n^\sigma] = (m-n) L_m^\sigma, \qquad m,n \geq N,
\end{equation}
for an arbitrary fixed integer $N \geq -1$. 
We now shift indices so that our operators are indexed by integers $i=1,2,3,\ldots$ for any $N$. For any $N \geq -1$, we define the quadratic hamiltonians:
\begin{equation}
H_i^2 := \hbar L^\sigma_{i+N-1} =\frac{\hbar}{2} \sum_{r \in \mathbb{Z} + \frac{1}{2}} \altcolon b_r^\sigma b_{i+N-1-r}^\sigma \altcolon + \frac{\hbar}{16} \delta_{i,1-N}. 
\end{equation}
Those have commutation relations:
\begin{equation}
[H_i^2, H_j^2] = \hbar(i-j) H^2_{i+j+N-1}.
\end{equation}
To add linear terms, we consider the shift $b_{N-\frac{1}{2}} \mapsto b_{N-\frac{1}{2}} + \frac{1}{\sqrt{\hbar}}$, which creates linear terms $H_i^1 = \sqrt{\hbar} b_{i-\frac{1}{2}}^\sigma = \hbar \partial_i$. It also creates a constant term $H_N^0 = \frac{1}{2}$, but as usual we get rid of it without changing the algebra since $H_N$ does not appear on the right-hand-side of the commutation relations. We also add $D$-terms $\hbar D_i$ for $i=1,\ldots, N+1$ to the operators $H_i$, $i=1,\ldots,N+1$, since they do not appear on the right-hand-side of the commutation relations. The resulting quantum Airy structures are \eqref{eq:twistedH1}, as representations of the algebra \eqref{eq:twistedC1}.

\end{proof}

These quantum Airy structures are known to produce interesting enumerative invariants:
\begin{itemize}
\item
For $N=-1$, the quantum Airy structure takes the explicit form
\begin{gather}
\forall i \in \mathbb{Z}_{\geq 1}, \qquad H_i = \hbar \partial_i+ \frac{\hbar^2}{2} \sum_{j=1}^{i-2} \partial_{j}  \partial_{i-1-j} + \frac{\hbar}{2} \sum_{j = i-1}^\infty (2j-2i+3) x^{j-i+2}  \partial_{j} \nonumber\\
+ \frac{1}{8} (x^1)^2 \delta_{i,1}+ \frac{\hbar}{16} \delta_{i,2},
\end{gather}
where the second sum is understood to vanish for terms with $j < 1$. This the quantum Airy structure associated to the Eynard-Orantin topological recursion on the Airy spectral curve (up to trivial rescaling of the variables $x^i$) \cite{EO,EO2}. The constraints $H_i Z = 0$ reproduce the well-known Virasoro constraints for the Kontsevich-Witten tau-function of the KdV hierarchy, and $Z$ is a generating function for intersection numbers on the moduli space of curves \cite{Witten,Kontsevich,DVV}.
\item
For $N=0$, the quantum Airy structure takes the form:
\begin{gather}\label{eq:bgw}
\forall i \in \mathbb{Z}_{\geq 1}, \qquad H_i = \hbar \partial_i+ \frac{\hbar^2}{2} \sum_{j=1}^{i-1} \partial_{j}  \partial_{i-j} + \frac{\hbar}{2} \sum_{j = i}^\infty (2j-2i+1) x^{j-i+1}  \partial_{j} \nonumber\\
+ \frac{\hbar}{16} \delta_{i,1} + \hbar D_1 \delta_{i,1}.
\end{gather}
For $D_1=0$, this is the quantum Airy structure associated to the Eynard-Orantin topological recursion on the Bessel spectral curve (up to rescaling of variables $x^i$) \cite{Bessel}. In this case, the constraints $H_i Z =0 $ reproduce the Virasoro constraints for the Br\'{e}zin-Gross-Witten tau-function of the KdV hierarchy \cite{BGW,BGROSS,Gross:1980he,Mironov:1994mv}. $Z$ is now a generating function for intersection numbers on the moduli space of curves involving Norbury's cohomology class \cite{Norbclass, DYZ}.
\item
For arbitrary $N \geq 1$, the quantum Airy structure reads ($\forall i \in \mathbb{Z}_{\geq 1}$):
\begin{equation}
H_i = \hbar \partial_i+ \frac{\hbar^2}{2} \sum_{j=1}^{i+N-1} \partial_{j}  \partial_{i+N-j} + \frac{\hbar}{2} \sum_{j = i+N}^\infty (2j-2i-2N+1) x^{j-i-N+1}  \partial_{j}+ \hbar D_i \delta_{i \leq N+1}.
\end{equation}
It is at the moment unknown whether the partition function that it computes has an interesting enumerative interpretation. Remark that for $N \geq 1$, the partition function is non-trivial only if the $D_i$ do not all vanish. In this case, is it a generating function for some intersection numbers on the moduli space curves? Is $Z$ a tau-function for the KdV hierarchy? These questions certainly deserve further investigation.
\end{itemize}

\subsection{Super Quantum Airy Structures from the Free Boson-Fermion VOSA}

Along the same lines as the bosonic construction of the previous section, we now construct classes of examples of infinite-dimensional, quadratic, super quantum Airy structures as representations of subalgebras of the super Virasoro algebra with central charge $c=\frac{3}{2}$:
\begin{align}
[L_m,L_n]&=(m-n)L_{n+m}+\delta_{m,-n}\frac{1}{8}m(m^2-1),\nonumber\\
[L_n,G_r]&=\left(\frac{n}{2}-r\right)G_{n+r},\label{eq:superVir}\\
\{G_r,G_s\}&=2L_{r+s}+\delta_{r,-s}\frac{1}{2}\left(r^2-\frac{1}{4}\right).\nonumber
\end{align}
We will construct super quantum Airy structures as representations of subalgebras of the super Virasoro algebra in both the Neveu-Schwarz (NS) sector (where $r,s$ are half-integers) and the Ramond sector (where $r,s$ are integers).

Our main tool is the free boson-fermion vertex operator super algebra (VOSA), which has $N=1$ supersymmetry. We will construct our super quantum Airy structures from untwisted and twisted representations of the free boson-fermion VOSA. So let us first review the main features of this theory. We refer the reader to \cite{VOSA1,VOSA2,VOSA3} for more details.

\begin{remark}
In this section we construct classes of super quantum Airy structures as representations of subalgebra of the super Virasoro algebra in the NS and Ramond sectors. A natural question then is whether the associated partition functions $Z$ compute interesting enumerative invariants. This is unclear at the moment, and certainly deserves further investigation. For instance, it would be very interesting to see whether these super quantum Airy structures are related to the supersymmetric generalization of JT gravity and Mirzakhani's recursion presented in \cite{Stanford:2019vob}.

We also remark that we only consider the free boson-fermion VOSA here, which has $N=1$ supersymmetry. But it would interesting to investigate whether super quantum Airy structures can be constructed as representations of VOSAs with $N=2$ supersymmetry as well.
\end{remark}

\subsubsection{The Free Boson-Fermion VOSA}

The free boson VOA was introduced in Section \ref{s:freeboson}. Let us now introduce the free fermion VOSA. It is again generated by a single vector $\psi_{-\frac{1}{2}}\ket{0} \in V_f$, where $V_f$ is the space of states. Here $\ket{0} \in V_f$ is the vacuum vector. The state-operator correspondence reads:
\begin{equation}
Y(\psi_{-\frac{1}{2}}\ket{0},z)=\sum_{m \in \mathbb{Z}} \psi_{m + \frac{1}{2}} z^{-m-1} = \sum_{r\in\mathbb{Z}+\frac{1}{2}}\psi_r z^{-r-\frac{1}{2}},\label{Y fermion}
\end{equation}
where the modes $\psi_r$ generate the Clifford algebra $\{\psi_r,\psi_s\}=\delta_{r,-s}$.

The vacuum vector $\ket{0}$ is annihilated by all  $\psi_{k}$ with $k > 0$, and the space of states $V_f$ is the Fock space of all excited states 
\begin{equation}
\psi_{-k_1}\cdots\psi_{-k_n}\ket{0}, \qquad k_1, \ldots, k_n \in \mathbb{Z}_{\geq 0} + \frac{1}{2}.
\end{equation}
Define a normal ordering on the fermionic modes as:
\begin{equation}
\altcolon\psi_r\psi_s\altcolon \; = \begin{cases} \psi_r\psi_s & \text{for $ r\leq s$,} \\ - \psi_s \psi_r & \text{for $r > s$.} \end{cases}
\end{equation}
Then the operators corresponding to the states in the Fock space are:
\begin{gather}\label{eq:fermionop}
Y(\psi_{-k_1}\cdots \psi_{-k_n}\ket{0},z) =\nonumber\\
\altcolon  \frac{1}{\left(k_1 - \frac{1}{2} \right)!} \left(\frac{d}{dz}\right)^{k_i-\frac{1}{2}}Y(\psi_{-\frac{1}{2}}\ket{0},z) \cdots \frac{1}{\left(k_n - \frac{1}{2} \right)!} \left(\frac{d}{dz}\right)^{k_n-\frac{1}{2}}Y(\psi_{-\frac{1}{2}}\ket{0},z)  \altcolon.
\end{gather}

The conformal vector $\ket{\omega}$ for the free fermion VOSA reads:
\begin{equation}
\ket{\omega}=\frac{1}{2}\psi_{-\frac{3}{2}}\psi_{-\frac{1}{2}}\ket{0}.
\end{equation}
Its operator takes the form
\begin{equation}
Y(\ket{\omega},z) = \sum_{m\in\mathbb{Z}} L_m z^{-m-2},
\end{equation}
with its modes generating the Virasoro algebra with central charge $c=\frac{1}{2}$:
\begin{equation}
[L_m,L_n] = (m-n) L_{m+n} + \frac{1}{24} m(m^2-1) \delta_{m,-n}.
\end{equation}
The Virasoro modes are related to the modes of the fermionic field as follows. From \eqref{eq:fermionop}, we have:
\begin{align}
Y(\ket{\omega},z) =& \frac{1}{2} \altcolon \left(  \frac{d}{dz} Y(\psi_{-\frac{1}{2}} \ket{0}, z) \right) Y(\psi_{-\frac{1}{2}} \ket{0}, z)  \altcolon \nonumber\\
=& - \frac{1}{2} \sum_{r_1, r_2 \in \mathbb{Z} + \frac{1}{2}} \left(r_1 + \frac{1}{2} \right) \altcolon \psi_{r_1} \psi_{r_2}  \altcolon z^{-r_1-r_2-2}.
\end{align}
Thus
\begin{align}\label{eq:Lfermion}
L_m =& - \frac{1}{2} \sum_{r \in \mathbb{Z} + \frac{1}{2}}  \left( r + \frac{1}{2} \right) \altcolon \psi_r \psi_{m-r} \altcolon \nonumber\\
=& \frac{1}{2} \sum_{r \in \mathbb{Z} + \frac{1}{2} } \left(r + \frac{m}{2} \right) \altcolon \psi_{-r} \psi_{r+m} \altcolon,
\end{align}
where the second equality follows from a straightforward calculation.

Now let $V_b$ be a free boson VOA, and $V_f$ be a free fermion VOSA. Let us consider their tensor product $V=V_b\otimes V_f$. More precisely, the vector space is a Fock space of all states excited by bosonic modes $b_{-n}$ and fermionic modes $\psi_{-r}$, where we assume $[b_n,\psi_r]=0$ for any $n\in\mathbb{Z}$ and $r\in\mathbb{Z}+\frac12$. The state-operator correspondence $Y_V$ satisfies $Y_V(\ket{u}\otimes \ket{v},z)=Y_{V_b}(\ket{u},z)\otimes Y_{V_f}(\ket{v},z)$ for $\ket{u}\in V_b$ and $\ket{v}\in V_f$. 

We define the vacuum vector of the combined theory as $\ket{0} = \ket{0}_{V_b}\otimes\ket{0}_{V_f}$. From now on we will omit the tensor product symbols for clarity. The conformal vector $\ket{\omega}$ for $V$ reads:
\begin{equation}
\ket{\omega}=\frac12\left(b_{-1}b_{-1}+ \psi_{-\frac{3}{2}}\psi_{-\frac{1}{2}}\right) \ket{0}.
\end{equation}
Its modes are given by the sum of \eqref{eq:Lboson} and \eqref{eq:Lfermion}:
\begin{equation}\label{eq:virn1no}
L_n = \frac{1}{2}\sum_{k\in\mathbb{Z}}\altcolon b_{k}b_{n-k}\altcolon  + \frac{1}{2}\sum_{r\in\mathbb{Z}+\frac{1}{2}}\left(r+\frac{n}{2}\right)\altcolon\psi_{-r}\psi_{n+r}\altcolon .
\end{equation}
Those generate the Virasoro algebra with central charge $c=3/2$:
\begin{equation}
[L_m,L_n] = (m-n) L_{m+n} + \frac{1}{8} m(m^2-1) \delta_{m,-n} .
\end{equation}

It turns out that this theory has $N=1$ supersymmetry. This means that it has a superconformal vector $\ket{\tau}$, whose modes, together with the Virasoro modes, generate a super Virasoro algebra in the NS sector. For the free boson-fermion VOSA, the superconformal vector is \cite{VOSA3}:
\begin{equation}
\ket{\tau}=b_{-1} \psi_{-\frac12} \ket{0}.
\end{equation}
Its operator reads:
\begin{equation}
Y(\ket{\tau},z) = \sum_{m \in \mathbb{Z}} G_{m-\frac{1}{2}} z^{-m-1} = \sum_{r \in \mathbb{Z} + \frac{1}{2}} G_r z^{-r-\frac{3}{2}},
\end{equation}
with the modes $G_r$ given by
\begin{equation}
G_r=\sum_{m\in\mathbb{Z}}  \psi_{r-m} b_m .\label{G free}
\end{equation}
These $G_r$, together with the $L_m$ of \eqref{eq:virn1no}, generate an $N=1$ super Virasoro algebra in the NS sector with central charge $3/2$:
\begin{align}
[L_m,L_n]&=(m-n)L_{n+m}+\delta_{m,-n}\frac{1}{8}m(m^2-1),\nonumber\\
[L_n,G_r]&=\left(\frac{n}{2}-r\right)G_{n+r},\label{eq:superVir}\\
\{G_r,G_s\}&=2L_{r+s}+\delta_{r,-s}\frac{1}{2}\left(r^2-\frac{1}{4}\right).\nonumber
\end{align}
We call this theory the \emph{free boson-fermion VOSA.}


\subsubsection{Untwisted and Twisted Representations for the Free Boson-Fermion  VOSA}

Our goal is now to construct super quantum Airy structures as representations of subalgebras of the super Virasoro algebra (in both NS and Ramond sectors). To this end, we will construct untwisted and twisted representations for the free boson-fermion VOSA. 

We will construct four different representations. The first one is the untwisted one, which is obtained directly from the natural representation of the super Heisenberg algebra. The super Virasoro modes take the form \eqref{eq:virn1no} and \eqref{G free} in terms of the representation of the bosonic and fermionic modes.

For the three twisted ones, we will use three distinct $\mathbb{Z}_2$ automorphisms of the VOSA:
\begin{enumerate}
\item The $\sigma$-twisted representation is obtained using the $\mathbb{Z}_2$ automorphism $\sigma: V \to V$ that we already studied in Lemma \ref{lem:sigma}, extended trivially to the fermionic sector. It acts on the Fock space of the free boson-fermion VOSA as follows:
\begin{equation}
\sigma:~ b_{-k_1} \cdots b_{-k_m} \psi_{-r_1} \cdots \psi_{-r_n} \ket{0} \mapsto (-1)^{\sum_{i=1}^m k_i} b_{-k_1} \cdots b_{-k_m} \psi_{-r_1} \cdots \psi_{-r_n} \ket{0}.
\end{equation}
In other words, it acts on the bosons as before, but leaves the fermions invariant. It preserves the vacuum vector $\ket{0}$ and the conformal vector $\ket{\omega}$. However, it twists the boson $b_{-1} \ket{0}$, and also the superconformal vector $ \ket{\tau} = b_{-1} \psi_{-\frac{1}{2} } \ket{0}$. Therefore, the $\sigma$-twisted representation will product a representation of the super Virasoro algebra in the Ramond sector.
\item The $\mu$-twisted representation is obtained using the parity $\mathbb{Z}_2$ automorphism $\mu: V \to V$, which sends odd vectors to minus themselves and keeps even vectors invariant. On the Fock space, it acts as:
\begin{equation}
\mu:~ b_{-k_1} \cdots b_{-k_m} \psi_{-r_1} \cdots \psi_{-r_n} \ket{0} \mapsto (-1)^{n} b_{-k_1} \cdots b_{-k_m} \psi_{-r_1} \cdots \psi_{-r_n} \ket{0}.
\end{equation}
It preserves the vacuum vector $\ket{0}$ and the conformal vector $\ket{\omega}$, but it twists the fermion $\psi_{- \frac{1}{2}} \ket{0}$ and the superconformal vector $\ket{\tau}= b_{-1} \psi_{-\frac{1}{2} } \ket{0}$. We will then again obtain a representation of the super Virasoro algebra in the Ramond sector.
\item Our last twisted representation is the $\rho$-twisted representation, with $\rho = \sigma \circ \mu : V \to V$, where we combine both automorphisms. The combined $\mathbb{Z}_2$ automorphism acts on the Fock space as:
 \begin{equation}
\sigma \circ \mu:~ b_{-k_1} \cdots b_{-k_m} \psi_{-r_1} \cdots \psi_{-r_n} \ket{0} \mapsto (-1)^{n +\sum_{i=1}^m k_i} b_{-k_1} \cdots b_{-k_m} \psi_{-r_1} \cdots \psi_{-r_n} \ket{0}.
\end{equation}
As usual, it keeps the vacuum and the conformal vectors invariant. While it twists both the boson vector $b_{-1} \ket{0}$ and the fermion vector $\psi_{-\frac{1}{2}} \ket{0}$, it keeps the superconformal vector $\ket{\tau}= b_{-1} \psi_{-\frac{1}{2} } \ket{0}$ invariant. Thus we will get a representation of the super Virasoro algebra in the NS sector.
\end{enumerate}

Let us now calculate the super Virasoro modes for all three $\mathbb{Z}_2$-twisted representations.
\begin{lemma}\label{lem:twistedreps}
\quad
\begin{itemize}
\item For the $\sigma$-twisted representation, 
the super Virasoro modes (in the Ramond sector) take the form
\begin{subequations}\label{eq:SVsigma}
\begin{align}
L_m^\sigma =& \frac{1}{2} \sum_{r \in \mathbb{Z} + \frac{1}{2}} \altcolon b_r^\sigma b_{m-r}^\sigma \altcolon + \frac{1}{2}\sum_{r\in\mathbb{Z}+\frac{1}{2}}\left(r+\frac{m}{2}\right)\altcolon\psi^\sigma_{-r}\psi^\sigma_{m+r}\altcolon+ \frac{1}{16} \delta_{m,0} , \\
G_m^\sigma =& \sum_{r \in \mathbb{Z} + \frac{1}{2}} b_{m-r}^\sigma \psi_r^\sigma,
\end{align}
\end{subequations}
with $m \in \mathbb{Z}$, in terms of the twisted bosonic modes and untwisted fermionic modes:
\begin{equation}
Y^\sigma(b_{-1} \ket{0},z) = \sum_{r \in \mathbb{Z} + \frac{1}{2} } b^\sigma_r z^{-r-1}, \qquad Y^\sigma(\psi_{-\frac{1}{2}}\ket{0},z)=\sum_{r\in\mathbb{Z}+\frac{1}{2}}\psi^\sigma_r z^{-r-\frac{1}{2}}.
\end{equation}
\item For the $\mu$-twisted representation, the super Virasoro modes (in the Ramond sector) take the form
\begin{subequations}\label{eq:SVmu}
\begin{align}
L_m^\mu =& \frac{1}{2} \sum_{i \in \mathbb{Z} } \altcolon b_i^\mu b_{m-i}^\mu \altcolon + \frac{1}{2}\sum_{i \in\mathbb{Z}}\left(i+\frac{m}{2}\right)\altcolon\psi^\mu_{-i}\psi^\mu_{m+i}\altcolon+ \frac{1}{16} \delta_{m,0} , \\
G_m^\mu =& \sum_{i \in \mathbb{Z}} b_{m-i}^\mu \psi_i^\mu,
\end{align}
\end{subequations}
with $m \in \mathbb{Z}$, in terms of the untwisted bosonic modes and twisted fermionic modes:
\begin{equation}
Y^\mu(b_{-1} \ket{0},z) = \sum_{m \in \mathbb{Z} } b^\mu_m z^{-m-1}, \qquad Y^\mu(\psi_{-\frac{1}{2}}\ket{0},z)=\sum_{n \in\mathbb{Z}}\psi_n^\mu z^{-n-\frac{1}{2}}.
\end{equation}
\item For the $\rho=\sigma \circ \mu$-twisted representation, the super Virasoro modes (in the NS sector) take the form:
\begin{subequations}\label{eq:SVrho}
\begin{align}
L_m^\rho =& \frac{1}{2} \sum_{r \in \mathbb{Z} +\frac{1}{2} } \altcolon b_r^\rho b_{m-r}^\rho \altcolon + \frac{1}{2}\sum_{i \in\mathbb{Z}}\left(i+\frac{m}{2}\right)\altcolon\psi^\rho_{-i}\psi^\rho_{m+i}\altcolon+ \frac{1}{8} \delta_{m,0} , \\
G_r^\rho =& \sum_{s \in \mathbb{Z} + \frac{1}{2}} b_{s}^\rho \psi_{r-s}^\rho,
\end{align}
\end{subequations}
with $m \in \mathbb{Z}$ and $r \in \mathbb{Z} + \frac{1}{2}$, in terms of the twisted bosonic modes and twisted fermionic modes:
\begin{equation}
Y^\rho(b_{-1} \ket{0},z) = \sum_{r \in \mathbb{Z} + \frac{1}{2} } b^\rho_r z^{-r-1}, \qquad Y^\rho(\psi_{-\frac{1}{2}}\ket{0},z)=\sum_{n \in\mathbb{Z}}\psi_n^\rho z^{-n-\frac{1}{2}}.
\end{equation}
\end{itemize}
\end{lemma}

\begin{proof}
For the $\sigma$-twisted representation, the boson is twisted, while the fermion is not:
\begin{equation}
Y^\sigma(b_{-1} \ket{0},z) = \sum_{r \in \mathbb{Z} + \frac{1}{2} } b^\sigma_r z^{-r-1}, \qquad Y^\sigma(\psi_{-\frac{1}{2}}\ket{0},z)=\sum_{r\in\mathbb{Z}+\frac{1}{2}}\psi^\sigma_r z^{-r-\frac{1}{2}}.
\end{equation}
The calculation of the conformal field $Y^\sigma(\ket{\omega},z)$ is the same as in Lemma \ref{lem:sigma}, since only the boson is twisted. The superconformal field $Y^\sigma(\ket{\tau},z)$ is twisted, and we get:
\begin{align}
Y^\sigma(\ket{\tau},z) =& \sum_{m \in \mathbb{Z}} G^\sigma_m z^{-m - \frac{3}{2}} \nonumber \\
=& Y^\sigma(b_{-1} \ket{0}, z) Y^\sigma(\psi_{-\frac{1}{2}, z}) \nonumber\\
=& \sum_{r,s \in \mathbb{Z} + \frac{1}{2}} b_r^\sigma \psi_s^\sigma z^{-r-s-\frac{3}{2}}.
\end{align}

For the $\mu$-twisted representation, the boson is untwisted, while the fermion is:
\begin{equation}
Y^\mu(b_{-1} \ket{0},z) = \sum_{m \in \mathbb{Z} } b^\mu_m z^{-m-1}, \qquad Y^\mu(\psi_{-\frac{1}{2}}\ket{0},z)=\sum_{n \in\mathbb{Z}}\psi_n^\mu z^{-n-\frac{1}{2}}.
\end{equation}
The bosonic part of the conformal field is untwisted. For the fermionic part, the calculation follows along the same lines as in Lemma \ref{lem:sigma}. We will omit it for brevity.
The superconformal field $Y^\mu(\ket{\tau},z)$ is twisted, and we get:
\begin{align}
Y^\mu(\ket{\tau},z) =& \sum_{m \in \mathbb{Z}} G_m^\mu z^{-m - \frac{3}{2}} \nonumber \\
=& Y^\mu(b_{-1} \ket{0}, z) Y^\mu(\psi_{-\frac{1}{2}, z}) \nonumber\\
=& \sum_{i,j \in \mathbb{Z}} b_i^\mu \psi_j^\mu z^{-i-j-\frac{3}{2}}.
\end{align}

For the $\rho=\sigma \circ \mu$-twisted representation, both boson and fermion are twisted:
\begin{equation}
Y^\rho(b_{-1} \ket{0},z) = \sum_{r \in \mathbb{Z} + \frac{1}{2} } b^\rho_r z^{-r-1}, \qquad Y^\rho(\psi_{-\frac{1}{2}}\ket{0},z)=\sum_{n \in\mathbb{Z}}\psi^\rho_n z^{-n-\frac{1}{2}}.
\end{equation}
The conformal field is calculated by combining the calculations for the $\sigma$-twisted and $\mu$-twisted representation. As for the superconformal field, it is untwisted, and we get:
\begin{align}
Y^\rho(\ket{\tau},z) =& \sum_{r \in \mathbb{Z} + \frac{1}{2}} G^\rho_r z^{-r - \frac{3}{2}} \nonumber \\
=& Y^\rho(b_{-1} \ket{0}, z) Y^\rho(\psi_{-\frac{1}{2}, z}) \nonumber\\
=& \sum_{s \in \mathbb{Z}+\frac{1}{2}} \sum_{m \in \mathbb{Z}} b_s^\rho \psi_m^\rho z^{-m-s-\frac{3}{2}}.
\end{align}
\end{proof}

Let us now construct classes of super quantum Airy structures using these four representations.

\subsubsection{Super Quantum Airy Structures from Untwisted Representations of the Free Boson-Fermion VOSA}

We first construct super quantum Airy structures from the untwisted representation of the free boson-fermion VOSA. 

\begin{proposition}\label{prop:SVuntwisted}
We represent the super Heisenberg algebra of untwisted bosonic and fermionic modes as:
\begin{align}
\forall m \in \mathbb{Z}_{\geq 1},& \qquad b_m^M = \sqrt{\hbar} \frac{\partial}{\partial x^{m}}, \qquad b_{-m}^M = \frac{m}{\sqrt{\hbar}} x^{m}, \qquad b_0^M = \sqrt{\hbar} \frac{\partial}{\partial x^0}, \nonumber\\
\forall r \in \mathbb{Z}_{\geq 0} + \frac{1}{2},& \qquad \psi_r^M = \sqrt{\hbar} \frac{\partial}{\partial \theta^{r + \frac{1}{2}}}, \qquad \psi_{-r}^M = \frac{1}{\sqrt{\hbar}} \theta^{r + \frac{1}{2}}.
\end{align}
Let $\{x^0,x^1,  \ldots \}$ (even) and $\{\theta^1, \theta^2, \ldots\}$ (odd) be a basis for the super vector space $V$, with dual sets $\{y_0,y_1,  \ldots \}$ and $\{\eta_1, \eta_2, \ldots \}$. Define the operators $H_i, F_r \in  \widehat{\mathcal{W}}_{\hbar}(V)$ for $i \in \mathbb{Z}_{i \geq 0}$ and $r \in \mathbb{Z}_{\geq 0} + \frac{1}{2}$:
\begin{subequations}
\begin{align}
H_i =& \hbar \frac{\partial}{\partial x^i} +\frac{\hbar}{2}\sum_{k\in\mathbb{Z}}\altcolon b_{k}^M b_{i+N-1-k}^M\altcolon  + \frac{\hbar}{2}\sum_{r\in\mathbb{Z}+\frac{1}{2}}\left(r+\frac{i+N-1}{2}\right)\altcolon\psi_{-r}^M \psi_{i+N-1+r}^M \altcolon , \\
F_r =& \hbar \frac{\partial}{\partial \theta^{r + \frac{1}{2}}}+ \hbar \sum_{m\in\mathbb{Z}}  \psi_{r+N-1-m}^M b_m^M,
\end{align}
\end{subequations}
which generate the following subalgebra of the super Virasoro algebra in the NS sector:
\begin{align}\label{eq:untwistC1}
[H_m, H_n] =& (m-n)H_{m+n+N-1}, \nonumber\\
[H_m, F_r] =& \left( \frac{m-N+1}{2} - r  \right) F_{m+r+N-1}, \nonumber\\
\{ F_r, F_s \} =& 2 H_{r+s+N-1}.
\end{align}
\begin{enumerate}
\item Let $N$ be any integer $N \geq 0$. The linear operator $H: V^* \to \widehat{\mathcal{W}}_{\hbar}(V)$ defined by, for all $i \in \mathbb{Z}_{\geq 0}$ and $r \in \mathbb{Z}_{\geq 0} + \frac{1}{2}$,
\begin{equation}
H(y_i) = H_i +  \hbar D_i \delta_{i \leq N-1}, \qquad H(\eta_{r+\frac{1}{2}}) = F_r,
\end{equation}
for arbitrary constants $D_i$, $i=0,\ldots,N-1$, forms a super quantum Airy structure as a representation of the algebra \eqref{eq:untwistC1}.
\item Let $N$ be any integer $N \geq 1$. The linear operator $H: V^* \to \widehat{\mathcal{W}}_{\hbar}(V)$ defined by, for all $i \in \mathbb{Z}_{\geq 1}$ and $r \in \mathbb{Z}_{\geq 0} + \frac{1}{2}$,
\begin{align}
H(y_i) =& H_i + \hbar D_i \delta_{i \leq N-1}, \qquad H(\eta_{r+\frac{1}{2}}) = F_r, \nonumber\\
H(y_0) =& \hbar \frac{\partial}{\partial x^0} + \frac{\hbar^2}{2} C_0 \frac{\partial^2}{\partial (x^0)^2} + \hbar D_0,
\end{align}
for arbitrary constants $D_i$, $i=0,\ldots,N-1$ and $C_0$, forms a super quantum Airy structure as a representation of the algebra \eqref{eq:untwistC1} extended by
\begin{equation}
[H(y_0), H(y_i)] = [H(y_0), H(\eta_{r + \frac{1}{2}} )]= 0 .
\end{equation}
\item Let $N$ be any integer $N \geq -1$. Let $V_{red} \subset V$ be the subspace spanned by $\{x^0, x^1, \ldots \}$ and $\{\theta^2, \theta^3, \ldots \}$. The linear operator $H: V^*_{red} \to \widehat{\mathcal{W}}_{\hbar}(V)$ defined by, for all $i \in \mathbb{Z}_{\geq 1}$ and $r \in \mathbb{Z}_{\geq 1} + \frac{1}{2}$,
\begin{align}
H(y_i) =& H_i + \hbar D_i \delta_{i \leq N+1}, \qquad H(\eta_{r+\frac{1}{2}}) = F_r, \nonumber\\
H(y_0) =& \hbar \frac{\partial}{\partial x^0} + \frac{\hbar^2}{2} C_0 \frac{\partial^2}{\partial (x^0)^2} + \hbar D_0,
\end{align}
for arbitrary constants $D_i$, $i=0,\ldots,N+1$, and $C_0$, form a super quantum Airy structure as a representation of the algebra \eqref{eq:untwistC1} extended by
\begin{equation}
[H(y_0), H(y_i)] = [H(y_0), H(\eta_{r + \frac{1}{2}} )]= 0 .
\end{equation}
This is a super quantum Airy structure with an extra fermionic variable, $\theta^1$.
\end{enumerate}
We remark that we could have considered as separate cases setting the bosonic zero mode $b_0^M=0$. But since the result will be equivalent to the cases with an auxiliary operator $\mathcal{H}_0$ with $C_0 = D_0 = 0$, we did not consider it separately.
\end{proposition}

\begin{proof}
We start with the untwisted representation, see Lemma \ref{lem:twistedreps}. We represent the Heisenberg algebra for untwisted bosonic and fermionic modes as:
\begin{align}
\forall m \in \mathbb{Z}_{\geq 1},& \qquad b_m^M = \sqrt{\hbar} \frac{\partial}{\partial x^{m}}, \qquad b_{-m}^M = \frac{m}{\sqrt{\hbar}} x^{m}, \qquad b_0^M = \sqrt{\hbar} \frac{\partial}{\partial x^0}, \nonumber\\
\forall r \in \mathbb{Z}_{\geq 0} + \frac{1}{2},& \qquad \psi_r^M = \sqrt{\hbar} \frac{\partial}{\partial \theta^{r + \frac{1}{2}}}, \qquad \psi_{-r}^M = \frac{1}{\sqrt{\hbar}} \theta^{r + \frac{1}{2}}.
\end{align}
The super Virasoro generators (in the NS sector) take the form:
\begin{subequations}
\begin{align}
L_n^M =& \frac{1}{2}\sum_{k\in\mathbb{Z}}\altcolon b_{k}^M b_{n-k}^M\altcolon  + \frac{1}{2}\sum_{r\in\mathbb{Z}+\frac{1}{2}}\left(r+\frac{n}{2}\right)\altcolon\psi_{-r}^M \psi_{n+r}^M \altcolon ,\\
G_r^M=& \sum_{m\in\mathbb{Z}}  \psi_{r-m}^M b_m^M .
\end{align}
\end{subequations}

We consider the closed subalgebras $\{L_m^M, G_r^M \}$ with $m \geq N-1$ and $r \geq N - \frac{1}{2}$ for any $N \geq 0$, which have no central term. We shift indices to index the bosonic generators with $0,1,2,\ldots$ and the fermionic ones with $\frac{1}{2}, \frac{3}{2}, \ldots$. We define the quadratic hamiltonians:
\begin{align}\label{eq:okok}
H_i^2 :=& \hbar L^M_{i+N-1} = \frac{\hbar}{2}\sum_{k\in\mathbb{Z}}\altcolon b_{k}^M b_{i+N-1-k}^M\altcolon  + \frac{\hbar}{2}\sum_{r\in\mathbb{Z}+\frac{1}{2}}\left(r+\frac{i+N-1}{2}\right)\altcolon\psi_{-r}^M \psi_{i+N-1+r}^M \altcolon ,\\
F_r^2 :=& \hbar G^M_{r+N-1} = \hbar \sum_{m\in\mathbb{Z}}  \psi_{r+N-1-m}^M b_m^M ,
\end{align}
which have the commutation relations
\begin{align}\label{eq:right}
[H_m^2, H_n^2] =& (m-n)H^2_{m+n+N-1}, \nonumber\\
[H_m^2, F_r^2] =& \left( \frac{m-N+1}{2} - r  \right) F^2_{m+r+N-1}, \nonumber\\
\{ F^2_r, F^2_s \} =& 2 H^2_{r+s+N-1}.
\end{align}
We shift the bosonic mode $b_{N-1}^M \mapsto b_{N-1}^M + \frac{1}{\sqrt{\hbar}}$ to get linear terms $H_i^1 = \sqrt{\hbar} b_i^M = \hbar \frac{\partial}{\partial x^i}$ and $F_r^1 = \sqrt{\hbar} \psi_r^M = \hbar \frac{\partial}{\partial \theta^{r+\frac{1}{2}}}$. We also get a constant term $H^0_{N-1} = \frac{1}{2}$, which we get rid of without changing the algebra. We add $D$-terms $\hbar D_i$ to the operators $H_i$ with $i=0,\ldots,N-1$. This gives our first class of super quantum Airy structures. 

We may want to consider $\theta^1$ as being an extra fermionic variable. For this we would like to consider the smaller algebra $\{L_m^M, G_r^M \}$ with $m \geq N$ and $r \geq N + \frac{3}{2}$. This is closed only for $N \geq 0$. Then we proceed as in the previous case, and the $D$-terms that we can add are the same. Thus, the operators are precisely the same as in the previous case, minus the operator $F_{\frac{1}{2}}$. By uniqueness of the partition function, it follows that the solution to the constraint is the same as in the previous case, and hence, in particular, it also satisfies the constraint $F_{\frac{1}{2} } Z = 0$. Therefore it is the same super quantum Airy structure.

The next thing that we can try is add an auxiliary operator $\mathcal{H}_0$ as in Proposition \ref{prop:untwisted}. For this, we consider the subalgebra $\{L_m^M, G_r^M \}$ with $m \geq N$ and $r \geq N - \frac{1}{2}$, which is closed for $N \geq 1$. We then shift indices to index the bosonic operators with $1,2,\ldots$ (without the $0$) and the fermionic operators with $\frac{1}{2}, \frac{3}{2}, \ldots$. To do this, we use the same quadratic hamiltonians as \eqref{eq:okok}, with commutation relations \eqref{eq:right}. We shift the bosonic mode $b_{N-1}^M \mapsto b_{N-1}^M + \frac{1}{\sqrt{\hbar}}$ to create linear terms $H_i^1 = \sqrt{\hbar} b_i^M = \hbar \frac{\partial}{\partial x^i}$ and $F_r^1 = \sqrt{\hbar} \psi^M_r = \hbar \frac{\partial}{\partial \theta^{r+\frac{1}{2}}}$. We also get a constant term $H^0_{N-1}= \frac{1}{2}$ which we get rid of as usual. We add $D$-terms $\hbar D_i$ for $i=1,\ldots,N-1$. To get a quantum Airy structure, we need to supplement with the auxiliary operator
\begin{equation}\label{eq:supp}
\mathcal{H}_0 = \hbar \frac{\partial}{\partial x^0} + \frac{\hbar^2}{2} C_0 \frac{\partial^2}{\partial (x^0)^2} + \hbar D_0,
\end{equation}
which commutes with all other operators. This is our second class of super quantum Airy structures.
 
For the third class, we want to keep using $\mathcal{H}_0$, but we would like to think of $\theta^1$ as an extra fermionic variable. For this, we consider the subalgebra $\{L_m^M, G_r^M \}$ with $m \geq N$ and $r \geq N + \frac{1}{2}$ for any $N \geq -1$. We want to shift indices so that bosonic generators are indexed with $1,2,\ldots$ and fermionic ones with $\frac{3}{2}, \frac{5}{2}, \ldots$. The same hamiltonians \eqref{eq:okok} will do the job, with commutation relations \eqref{eq:right}.
To create appropriate linear terms, we shift the bosonic mode $b^M_{N-1} \mapsto b^M_{N-1} + \frac{1}{\sqrt{\hbar}}$ as usual. We add $D$-terms $\hbar D_i$ for $i=1,\ldots,N+1$, and the auxiliary operator \eqref{eq:supp}. This gives our third class of super quantum Airy structures.

We note that we could also consider separately the cases where we set the bosonic zero mode $b_0^M = 0$, but in the end it is equivalent to the special case of the auxiliary operator $\mathcal{H}_0$ with $C_0 = D_0 = 0$, and thus we do not consider it separately.

\end{proof}

\subsubsection{Super Quantum Airy Structures from $\sigma$-Twisted Representations of the Free Boson-Fermion VOSA}

We now consider the $\sigma$-twisted representation of the free boson-fermion VOSA. We construct the following two classes of super quantum Airy structures:

\begin{proposition}\label{prop:SVsigma}
Let $N$ be any integer $N \geq 0$. We represent the super Heisenberg algebra for the twisted bosonic modes and untwisted fermionic modes as:
\begin{align}
\forall r \in \mathbb{Z}_{\geq 0} + \frac{1}{2},& \qquad b_r^\sigma = \sqrt{\hbar} \frac{\partial}{\partial x^{r+\frac{1}{2}}}, \qquad b_{-r}^\sigma = \frac{r}{\sqrt{\hbar}} x^{r + \frac{1}{2}}, \nonumber\\
\forall r \in \mathbb{Z}_{\geq 0} + \frac{1}{2},& \qquad \psi_r^\sigma = \sqrt{\hbar} \frac{\partial}{\partial \theta^{r + \frac{1}{2}}}, \qquad \psi_{-r}^\sigma = \frac{1}{\sqrt{\hbar}} \theta^{r + \frac{1}{2}},
\end{align}
Let $\{x^1, x^2,  \ldots \}$ (even) and $\{\theta^1, \theta^2, \ldots\}$ (odd) be a basis for the super vector space $V$, with dual sets $\{y_1, y_2, \ldots \}$ and $\{\eta_1, \eta_2, \ldots \}$. Define the differential operators $H_i, F_i \in \widehat{\mathcal{W}}_{\hbar}(V)$, for $i \in \mathbb{Z}_{\geq 1}$, 
\begin{subequations}
\begin{align}
H_i  =& \hbar \frac{\partial}{\partial x^i} +  \frac{\hbar}{2} \sum_{r \in \mathbb{Z} + \frac{1}{2}} \altcolon b_r^\sigma b_{i+N-1-r}^\sigma \altcolon + \frac{\hbar}{2}\sum_{r\in\mathbb{Z}+\frac{1}{2}}\left(r+\frac{i+N-1}{2}\right)\altcolon\psi^\sigma_{-r}\psi^\sigma_{i+N-1+r}\altcolon , \\
F_i =&   \hbar \frac{\partial}{\partial \theta^i} + \hbar \sum_{r \in \mathbb{Z} + \frac{1}{2}} b_{i+N-1-r}^\sigma \psi_r^\sigma,
\end{align}
\end{subequations}
which form a representation of the following subalgebra of the super Virasoro algebra in the Ramond sector:
\begin{align}\label{eq:sigmaC1}
[H_m, H_n] =& (m-n)H_{m+n+N-1}, \nonumber\\
[H_m, F_n] =& \left( \frac{m-N+1}{2} - n  \right) F_{m+n+N-1}, \nonumber\\
\{ F_m, F_n \} =& 2 H_{m+n+N-1}.
\end{align}
\begin{enumerate}
\item The linear operator $H: V^* \to \widehat{\mathcal{W}}_{\hbar}(V)$ defined by, for all $i \in \mathbb{Z}_{\geq 1}$,
\begin{equation}
H(y_i) = H_i +  \hbar D_i \delta_{i \leq N}, \qquad H(\eta_i) = F_i,
\end{equation}
for arbitrary constants $D_i$, $i=1,\ldots,N$, forms a super quantum Airy structure as a representation of the algebra \eqref{eq:sigmaC1}.

\item Let $V_{red} \subset V$ be the subspace spanned by $\{x^1, x^2,  \ldots \}$ and $\{\theta^2, \theta^3, \ldots\}$. The linear operator $H: V^*_{red} \to  \widehat{\mathcal{W}}_{\hbar}(V)$ defined by, for all $i \in \mathbb{Z}_{\geq 1}$ and $j \in \mathbb{Z}_{\geq 2}$,
\begin{equation}\label{eq:ttt}
H(y_i) = H_i + \hbar D_i \delta_{i \leq N+1},  \qquad H(\eta_j)  = F_j,
\end{equation}
for arbitrary constants $D_i$, $i=1,\ldots, N+1$, forms a super quantum Airy structure with an extra fermionic coordinate $\theta^1$, as a representation of the algebra \eqref{eq:sigmaC1}.
\end{enumerate}

\end{proposition}

\begin{proof}
We start with the $\sigma$-twisted representation, see Lemma \ref{lem:twistedreps}. We represent the twisted bosonic modes and untwisted fermionic modes as:
\begin{align}
\forall r \in \mathbb{Z}_{\geq 0} + \frac{1}{2},& \qquad b_r^\sigma = \sqrt{\hbar} \frac{\partial}{\partial x^{r+\frac{1}{2}}}, \qquad b_{-r}^\sigma = \frac{r}{\sqrt{\hbar}} x^{r + \frac{1}{2}}, \nonumber\\
\forall r \in \mathbb{Z}_{\geq 0} + \frac{1}{2},& \qquad \psi_r^\sigma = \sqrt{\hbar} \frac{\partial}{\partial \theta^{r + \frac{1}{2}}}, \qquad \psi_{-r}^\sigma = \frac{1}{\sqrt{\hbar}} \theta^{r + \frac{1}{2}}.
\end{align}
Here there is no bosonic or fermionic zero mode. From Lemma \ref{lem:twistedreps}, the super Virasoro generators take the form \eqref{eq:SVsigma}:
\begin{subequations}
\begin{align}
L_m^\sigma =& \frac{1}{2} \sum_{r \in \mathbb{Z} + \frac{1}{2}} \altcolon b_r^\sigma b_{m-r}^\sigma \altcolon + \frac{1}{2}\sum_{r\in\mathbb{Z}+\frac{1}{2}}\left(r+\frac{m}{2}\right)\altcolon\psi^\sigma_{-r}\psi^\sigma_{m+r}\altcolon+ \frac{1}{16} \delta_{m,0} , \\
G_m^\sigma =& \sum_{r \in \mathbb{Z} + \frac{1}{2}} b_{m-r}^\sigma \psi_r^\sigma.
\end{align}
\end{subequations}
The super Virasoro algebra is in the Ramond sector.

For our first class of super quantum Airy structures, we consider the closed subalgebra $\{L^\sigma_m, G^\sigma_n \}$ with $m,n \geq N$, for any $N \geq 0$. It takes the form:
\begin{align}\label{eq:alll}
[L^\sigma_m, L^\sigma_n] =& (m-n) L^\sigma_{m+n}, \nonumber\\
[L^\sigma_m, G^\sigma_n] =& \left( \frac{m}{2} - n \right) G^\sigma_{m+n}, \nonumber\\
\{ G_m^\sigma, G^\sigma_n \} =& 2 L^\sigma_{m+n} - \frac{1}{8} \delta_{m,0} \delta_{n,0}.
\end{align}
To get rid of the central term, we redefine $L^\sigma_0 \mapsto L^\sigma_0 - \frac{1}{16}$, which does not change the rest of the algebra. We shift indices as usual to define the quadratic hamiltonians, for $i \in \mathbb{Z}_{\geq 1}$:
\begin{align}
H_i^2 :=& \hbar L^\sigma_{i+N-1} = \frac{\hbar}{2} \sum_{r \in \mathbb{Z} + \frac{1}{2}} \altcolon b_r^\sigma b_{i+N-1-r}^\sigma \altcolon + \frac{\hbar}{2}\sum_{r\in\mathbb{Z}+\frac{1}{2}}\left(r+\frac{i+N-1}{2}\right)\altcolon\psi^\sigma_{-r}\psi^\sigma_{i+N-1+r}\altcolon, \nonumber\\
F_i^2 :=& \hbar G^\sigma_{i+N-1} = \hbar \sum_{r \in \mathbb{Z} + \frac{1}{2}} b_{i+N-1-r}^\sigma \psi_r^\sigma,
\end{align}
which satisfy the commutation relations:
\begin{align}
[H_m^2, H_n^2] =& (m-n)H^2_{m+n+N-1}, \nonumber\\
[H_m^2, F_n^2] =& \left( \frac{m-N+1}{2} - n  \right) F^2_{m+n+N-1}, \nonumber\\
\{ F^2_m, F^2_n \} =& 2 H^2_{m+n+N-1}.
\end{align}
To create appropriate linear terms, we shift the bosonic modes $b^\sigma_{N - \frac{1}{2}} \mapsto b^\sigma_{N-\frac{1}{2}} + \frac{1}{\sqrt{\hbar}}$, which creates terms $H_i^1 = \sqrt{\hbar} b^\sigma_{i-\frac{1}{2}} = \hbar \frac{\partial}{\partial x^i}$ and $F_i^1 = \sqrt{\hbar} \psi^\sigma_{i-\frac{1}{2}} = \hbar \frac{\partial}{\partial \theta^i}$. This also creates a constant term $H_N^0 = \frac{1}{2}$ which we get rid of without changing the algebra. We can also add $D$-terms $\hbar D_i$ to $H_i$ for $i=1,\ldots,N$ without changing the algebra, which gives our first class of super quantum Airy structures.

For the second class, we consider a smaller closed subalgebra $\{L^\sigma_m, G^\sigma_n \}$ with $m \geq N$, $n \geq N+1$, for any $N \geq 0$. The algebra is still \eqref{eq:alll} but with no central term now. We shift indices as above: now $F_i$ starts with $i=2,3,\ldots$. We then shift the same bosonic modes to create appropriate linear terms, and introduce $D$-terms $\hbar D_i$ to $H_i$ for $i=1,\ldots,N+1$.  Note that we can add one more $D$-term, since $H_{N+1}$ does not appear anymore on the right-hand-side of the commutation relations. This gives a super quantum Airy structure, where $\theta^1$ is considered as an extra fermionic variable (since it does not appear in the linear terms). Note that in \eqref{eq:ttt} we absorbed the term $\frac{\hbar}{16} \delta_{i,1-N}$, which only appears for $N=0$, into the arbitrary constant $\hbar D_1$.

\end{proof}

\begin{remark}
Note that the case with $N=0$ for the second class of super quantum Airy structures in Proposition \ref{prop:SVsigma} is interesting. Setting $D_1=\frac{1}{16}$ to its natural value, we get a non-trivial partition function. The pure bosonic part of the Virasoro generators $H_i$ is in this case equivalent to the Virasoro operators that annihilate the Br\'ezin-Gross-Witten tau-function of the KdV hierachy (see \eqref{eq:bgw} and the discussion around there). This suggests that the partition function associated to this super quantum Airy structure may be a supersymmetric analog of the BGW tau-function, which would be worth investigating further.
\end{remark}

\begin{remark}
We also remark that for the second class of super quantum Airy structures in Proposition \ref{prop:SVsigma}, if the last arbitrary constant $D_{N+1}$ is set to zero, then we recover the same differential operators as the super quantum Airy structure in the first class, minus $F_1$. Then, by uniqueness of the partition function, we conclude that they share the same partition function, and that $F_1$ must also annihilate the partition function of the second class. Thus the super quantum Airy structures in the second class differ from for the first class only when $D_{N+1} \neq 0$.
\end{remark}

\subsubsection{Super Quantum Airy Structures from $\mu$-Twisted Representations of the Free Boson-Fermion VOSA}

We now consider the $\mu$-twisted representation of the free boson-fermion VOSA.

\begin{proposition}\label{prop:SVmu}
We represent the super Heisenberg algebra for the untwisted bosonic and twisted fermionic modes as:
\begin{align}
\forall m \in \mathbb{Z}_{\geq 1},& \qquad b_m^\mu = \sqrt{\hbar} \frac{\partial}{\partial x^{m}}, \qquad b_{-m}^\mu = \frac{m}{\sqrt{\hbar}} x^{m}, \qquad b_0^\mu = \sqrt{\hbar} \frac{\partial}{\partial x^0}, \nonumber\\
\forall m \in \mathbb{Z}_{\geq 1},& \qquad \psi_m^\mu = \sqrt{\hbar} \frac{\partial}{\partial \theta^{m}}, \qquad \psi_{-m}^\mu = \frac{1}{\sqrt{\hbar}} \theta^{m}, \qquad \psi_0^\mu = \frac{1}{\sqrt{2 \hbar}} \left( \theta^0 + \hbar \frac{\partial}{\partial \theta^0} \right).
\end{align}
Let $\{x^0, x^1, \ldots \}$ (even) and $\{\theta^1, \theta^2, \ldots \}$ (odd) be a basis for the super vector space $V$, with dual sets $\{y_0, y_1, \ldots \}$ and $\{\eta_1, \eta_2, \ldots \}$. Let $\widetilde{V} = V \oplus \mathbb{K}^{0|1}$, with $\{\theta^0 \}$ a basis for $\mathbb{K}^{0|1}$. Let us define the differential operators $H_i, F_j \in \widehat{\mathcal{W}}_\hbar(\widetilde{V})$, for $i \in \mathbb{Z}_{\geq 0}$ and $j \in \mathbb{Z}_{\geq 1}$:
\begin{subequations}
\begin{align}
H_i =&  \hbar \frac{\partial}{\partial x^i} +\frac{\hbar}{2} \sum_{j \in \mathbb{Z} } \altcolon b_j^\mu b_{i+N-1-j}^\mu \altcolon + \frac{\hbar}{2}\sum_{j \in\mathbb{Z}}\left(j+\frac{i+N-1}{2}\right)\altcolon\psi^\mu_{-j}\psi^\mu_{i+N-1+j}\altcolon,\\
F_j =& \hbar \frac{\partial}{\partial \theta^j} + \sum_{k \in \mathbb{Z}} b_{j+N-1-k}^\mu \psi_k^\mu,
\end{align}
\end{subequations}
which form a representation of the subalgebra of the super Virasoro algebra in the Ramond sector:
\begin{align}\label{eq:algmu}
[H_m, H_n] =& (m-n)H_{m+n+N-1}, \nonumber\\
[H_m, F_n] =& \left( \frac{m-N+1}{2} - n  \right) F_{m+n+N-1}, \nonumber\\
\{ F_m, F_n \} =& 2 H_{m+n+N-1}.
\end{align}

\begin{enumerate}
\item Let $N$ be any integer $N \geq 1$. The linear operator $H: V^* \to \widehat{\mathcal{W}}_\hbar(\widetilde{V})$ defined by, for all $i \in \mathbb{Z}_{\geq 0}$ and $j \in \mathbb{Z}_{\geq 1}$:
\begin{equation}\label{eq:saispus}
H(y_i) = H_i  + \hbar D_i \delta_{i\leq N-1}, \qquad H(\eta_j) = F_j,
\end{equation}
for arbitrary constants $D_i$, $i=0,\ldots,N-1$, forms a super quantum Airy structure as a representation of the algebra \eqref{eq:algmu}.

\item Let $N$ be any integer $N \geq 0$. The linear operator $H: V^* \to \widehat{\mathcal{W}}_\hbar(\widetilde{V})$ defined by, for all $i \in \mathbb{Z}_{\geq 1}$:
\begin{align}
H(y_i) =& H_i + \hbar D_i \delta_{i \leq N}, \qquad H(\eta_i) = F_i, \nonumber\\
H(y_0) =&  \hbar \frac{\partial}{\partial x^0} + \frac{\hbar^2}{2} C_0 \frac{\partial^2}{\partial (x^0)^2} + \hbar D_0,
\end{align}
for arbitrary constants $D_i$, $i=0,\ldots,N$ and $C_0$, forms a super quantum Airy structure as a representation of the algebra \eqref{eq:algmu} extended by:
\begin{equation}
[H(y_0), H(y_i)] = [H(y_0), H(\eta_i)] = 0.
\end{equation}

\end{enumerate}
We note that both of these super quantum Airy structures have an extra fermionic coordinate $\theta^0$. We remark that we could also consider a third case, where we set the bosonic zero mode $b_0^\mu = 0$, as in Proposition \ref{prop:untwisted}. But as it will be equivalent to the special case of case (2) with $C_0 = D_0 = 0$ we do not consider it separately.
\end{proposition}

\begin{proof}
We start with the $\mu$-twisted representation, see Lemma \ref{lem:twistedreps}. We represent the untwisted bosonic modes and twisted fermionic modes as:
\begin{align}
\forall m \in \mathbb{Z}_{\geq 1},& \qquad b_m^\mu = \sqrt{\hbar} \frac{\partial}{\partial x^{m}}, \qquad b_{-m}^\mu = \frac{m}{\sqrt{\hbar}} x^{m}, \qquad b_0^\mu = \sqrt{\hbar} \frac{\partial}{\partial x^0}, \nonumber\\
\forall m \in \mathbb{Z}_{\geq 1},& \qquad \psi_m^\mu = \sqrt{\hbar} \frac{\partial}{\partial \theta^{m}}, \qquad \psi_{-m}^\mu = \frac{1}{\sqrt{\hbar}} \theta^{m}, \qquad \psi_0^\mu = \frac{1}{\sqrt{2 \hbar}} \left( \theta^0 + \hbar \frac{\partial}{\partial \theta^0} \right).
\end{align}
We have both bosonic and fermionic zero modes. In particular, we expect to get only super quantum Airy structures with an extra fermionic coordinate, namely $\theta^0$. 

From Lemma \ref{lem:twistedreps}, the super Virasoro generators (in the Ramond sector) take the form:
\begin{subequations}
\begin{align}
L_m^\mu =& \frac{1}{2} \sum_{i \in \mathbb{Z} } \altcolon b_i^\mu b_{m-i}^\mu \altcolon + \frac{1}{2}\sum_{i \in\mathbb{Z}}\left(i+\frac{m}{2}\right)\altcolon\psi^\mu_{-i}\psi^\mu_{m+i}\altcolon+ \frac{1}{16} \delta_{m,0} , \\
G_m^\mu =& \sum_{i \in \mathbb{Z}} b_{m-i}^\mu \psi_i^\mu.
\end{align}
\end{subequations}

For our first class of super quantum Airy structures, we consider the subalgebra $\{L_m, G_n \}$ with $m \geq N-1$, $n \geq N$, for any $N \geq 1$. We shift indices so that bosonic operators are indexed by $\{0,1,2,\ldots \}$ and fermionic operators by $\{1,2,3,\ldots\}$. Thus we define quadratic hamiltonians:
\begin{align}\label{eq:shifthere}
H_i^2 :=& \hbar L^\mu_{i+N-1} =\frac{\hbar}{2} \sum_{j \in \mathbb{Z} } \altcolon b_j^\mu b_{i+N-1-j}^\mu \altcolon + \frac{\hbar}{2}\sum_{j \in\mathbb{Z}}\left(j+\frac{i+N-1}{2}\right)\altcolon\psi^\mu_{-j}\psi^\mu_{i+N-1+j}\altcolon+ \frac{\hbar}{16} \delta_{i,-N+1} , \nonumber\\
F_i^2 :=& \hbar G^\mu_{i+N-1} = \sum_{j \in \mathbb{Z}} b_{i+N-1-j}^\mu \psi_j^\mu,
\end{align}
with commutation relations:
\begin{align}
[H_m^2, H_n^2] =& (m-n)H^2_{m+n+N-1}, \nonumber\\
[H_m^2, F_n^2] =& \left( \frac{m-N+1}{2} - n  \right) F^2_{m+n+N-1}, \nonumber\\
\{ F^2_m, F^2_n \} =& 2 H^2_{m+n+N-1}.
\end{align}
To create appropriate linear terms, we shift the bosonic modes $b^\mu_{N-1} \mapsto b^\mu_{N-1} + \frac{1}{\sqrt{\hbar}}$, which creates linear terms $H_i^1 = \sqrt{\hbar} b^\mu_i = \hbar \frac{\partial}{\partial x^i}$ and $F_i^1 = \sqrt{\hbar} \psi_i^\mu = \hbar \frac{\partial}{\partial \theta^i}$. It also creates a constant term $H_{N-1}^0 = \frac{1}{2}$ which we get rid of without changing the algebra. We also add $D$-terms $\hbar D_i$, $i=0,\ldots,N-1$ without changing the algebra. This creates a super quantum airy structure with the extra fermionic variable $\theta^0$. Note that in \eqref{eq:saispus}, we absorbed the term $\frac{\hbar}{16} \delta_{i, -N+1}$, which only appears for $N=1$, into the arbitrary constant $D_0$.

For our second class, we consider the larger subalgebra $\{L_m, G_n \}$ with $m,n \geq N$, for any $N \geq 0$. As in Proposition \ref{prop:SVsigma}, to get rid of the central term in the algebra we redefine $L_0^\mu \mapsto L_0^\mu - \frac{1}{16}$. We shift indices as in \eqref{eq:shifthere}, but now we consider $H_i$ and $F_i$ with $i=1,2,3,\ldots$.
We shift the bosonic modes $b^\mu_{N-1} \mapsto b^\mu_{N-1} + \frac{1}{\sqrt{\hbar}}$ as before, which creates the right linear terms. We can add $D$-terms $\hbar D_i$, $i=1,\ldots,N$ without changing the algebra. 

This is not however a super quantum Airy structure, since $x^0$ does not appear in the linear terms ($H_0$ is not include in the algebra). But, as in Proposition \ref{prop:untwisted}, we notice that the $H_m$ and $F_m$ only depend on $x^0$ through the bosonic zero mode, i.e. through $\frac{\partial}{\partial x^0}$. Thus we can introduce an auxiliary bosonic operator
\begin{equation}
\mathcal{H}_0 = \hbar \frac{\partial}{\partial x^0} + \frac{\hbar^2}{2} C_0 \frac{\partial^2}{\partial (x^0)^2} + \hbar D_0,
\end{equation}
which commutes with all $H_m$, $F_m$, $m \geq 1$. The result is a super quantum Airy structure with an extra fermionic variable $\theta^0$.

We could consider a third class, where we set the bosonic zero mode $b_0^\mu = 0$. But this will end up being equivalent to the second class with $C_0 = D_0 = 0$, and hence we do not consider it separately.

\end{proof}

\subsubsection{Super Quantum Airy Structures from $\rho$-Twisted Representations of the Free Boson-Fermion  VOSA}

We finally consider the $\rho$-twisted representation of the free boson-fermion VOA. We get:

\begin{proposition}\label{prop:SVrho}
Let $N$ be any integer $N \geq -1$. We represent the super Heisenberg algebra for the twisted bosonic and fermionic modes as:
\begin{align}
\forall r \in \mathbb{Z}_{\geq 0} + \frac{1}{2},& \qquad b_r^\rho = \sqrt{\hbar} \frac{\partial}{\partial x^{r+\frac{1}{2}}}, \qquad b_{-r}^\rho = \frac{r}{\sqrt{\hbar}} x^{r + \frac{1}{2}}, \nonumber\\
\forall m \in \mathbb{Z}_{\geq 1},& \qquad \psi_m^\rho = \sqrt{\hbar} \frac{\partial}{\partial \theta^{m}}, \qquad \psi_{-m}^\rho = \frac{1}{\sqrt{\hbar}} \theta^{m}, \qquad \psi_0^\rho = \frac{1}{\sqrt{2 \hbar}} \left( \theta^0 + \hbar \frac{\partial}{\partial \theta^0} \right).
\end{align}
Let $\{x^1, x^2, \ldots \}$ (even) and $\{ \theta^1, \theta^2, \ldots \}$ (odd) be a basis for the super vector space $V$, with dual sets $\{y_1, y_2, \ldots \}$ and $\{\eta_1, \eta_2, \ldots \}$. Let $\widetilde{V} = V \oplus \mathbb{K}^{0|1}$, with $\{\theta^0 \}$ a basis for $\mathbb{K}^{0|1}$. The linear operator $H: V^* \to \widehat{\mathcal{W}}_{\hbar}(\widetilde{V})$ defined by, for all $i \in \mathbb{Z}_{\geq 1}$ and $r \in \mathbb{Z}_{\geq 0} + \frac{1}{2}$:
\begin{subequations}
\begin{align}
H_i :=& H(y_i) = \hbar \frac{\partial}{\partial x^i} + \frac{\hbar}{2} \sum_{r \in \mathbb{Z} +\frac{1}{2} } \altcolon b_r^\rho b_{i+N-1-r}^\rho \altcolon + \frac{\hbar}{2}\sum_{j \in\mathbb{Z}}\left(j+\frac{i+N-1}{2}\right)\altcolon\psi^\rho_{-j}\psi^\rho_{i+N-1+j}\altcolon \nonumber\\
& \qquad \qquad + \frac{\hbar}{8} \delta_{i,1-N} + \hbar D_i \delta_{i \leq N+1}, \\
F_r :=& H(\eta_{r+ \frac{1}{2}} ) = \hbar \frac{\partial}{\partial \theta^{r + \frac{1}{2}}}+ \sum_{s \in \mathbb{Z} + \frac{1}{2}} b_{s}^\rho \psi_{r+N-s}^\rho,
\end{align}
\end{subequations}
with arbitrary constants $D_i$, $i=1,\ldots,N+1$, forms a super quantum Airy structure as a representation of the subalgebra of the super Virasoro algebra in the NS sector:
\begin{align}
[H_m, H_n] =& (m-n)H_{m+n+N-1}, \nonumber\\
[H_m, F_r] =& \left( \frac{m-N-1}{2} - r  \right) F_{m+r+N-1}, \nonumber\\
\{ F_r, F_s \} =& 2 H_{r+s+N+1}.
\end{align}
Note that this is a super quantum Airy structure with an extra fermionic variable $\theta^0$.
\end{proposition}

\begin{proof}
We start with the $\rho$-twisted representation, see Lemma \ref{lem:twistedreps}. We represent the twisted bosonic modes and twisted fermionic modes as:
\begin{align}
\forall r \in \mathbb{Z}_{\geq 0} + \frac{1}{2},& \qquad b_r^\rho = \sqrt{\hbar} \frac{\partial}{\partial x^{r+\frac{1}{2}}}, \qquad b_{-r}^\rho = \frac{r}{\sqrt{\hbar}} x^{r + \frac{1}{2}}, \nonumber\\
\forall m \in \mathbb{Z}_{\geq 1},& \qquad \psi_m^\rho = \sqrt{\hbar} \frac{\partial}{\partial \theta^{m}}, \qquad \psi_{-m}^\rho = \frac{1}{\sqrt{\hbar}} \theta^{m}, \qquad \psi_0^\rho = \frac{1}{\sqrt{2 \hbar}} \left( \theta^0 + \hbar \frac{\partial}{\partial \theta^0} \right).
\end{align}
There is no bosonic zero mode here, but there is a fermionic one. Hence we expect to get only super quantum Airy structures with an extra fermionic coordinate, namely $\theta^0$.

From Lemma \ref{lem:twistedreps}, the super Virasoro generators (in the NS sector) take the form:
\begin{subequations}
\begin{align}
L_m^\rho =& \frac{1}{2} \sum_{r \in \mathbb{Z} +\frac{1}{2} } \altcolon b_r^\rho b_{m-r}^\rho \altcolon + \frac{1}{2}\sum_{i \in\mathbb{Z}}\left(i+\frac{m}{2}\right)\altcolon\psi^\rho_{-i}\psi^\rho_{m+i}\altcolon+ \frac{1}{8} \delta_{m,0} , \\
G_r^\rho =& \sum_{s \in \mathbb{Z} + \frac{1}{2}} b_{s}^\rho \psi_{r-s}^\rho,
\end{align}
\end{subequations}
We consider the closed subalgebras $\{L^\rho_m, G^\rho_r \}$ with $m \geq N$ and $r \geq N+\frac{1}{2}$, for any $N \geq -1$, which have no central term. We shift indices to index the bosonic generators with $1,2,\ldots$ and the fermionic ones with $\frac{1}{2}, \frac{3}{2}, \ldots$. We define the quadratic hamiltonians, for $i \in \mathbb{Z}_{\geq 1}$ and $r \in \mathbb{Z}_{\geq 0} + \frac{1}{2}$:
\begin{align}
H_i^2 :=& \hbar L^\rho_{i+N-1} =  \frac{\hbar}{2} \sum_{r \in \mathbb{Z} +\frac{1}{2} } \altcolon b_r^\rho b_{i+N-1-r}^\rho \altcolon + \frac{\hbar}{2}\sum_{j \in\mathbb{Z}}\left(j+\frac{i+N-1}{2}\right)\altcolon\psi^\rho_{-j}\psi^\rho_{i+N-1+j}\altcolon+ \frac{\hbar}{8} \delta_{i,1-N} , \nonumber\\
F_r^2 :=& \hbar G^\rho_{r+N} =\sum_{s \in \mathbb{Z} + \frac{1}{2}} b_{s}^\rho \psi_{r+N-s}^\rho,
\end{align}
which satisfy the commutation relations:
\begin{align}
[H_m^2, H_n^2] =& (m-n)H^2_{m+n+N-1}, \nonumber\\
[H_m^2, F_r^2] =& \left( \frac{m-N-1}{2} - r  \right) F^2_{m+r+N-1}, \nonumber\\
\{ F^2_r, F^2_s \} =& 2 H^2_{r+s+N+1}.
\end{align}
To create the appropriate linear terms, we shift the bosonic modes $b^\rho_{N-\frac{1}{2}} \mapsto b^\rho_{N-\frac{1}{2}} + \frac{1}{\sqrt{\hbar}}$, which creates terms $H_i^1 = \sqrt{\hbar} b^\rho_{i - \frac{1}{2}} = \hbar \frac{\partial}{\partial x^i}$ and $F_r^1 = \sqrt{\hbar} \psi^\rho_{r + \frac{1}{2}} = \hbar \frac{\partial}{\partial \theta^{r + \frac{1}{2}}}$. It also creates a constant term $H_N^0 = \frac{1}{2}$ which we get rid of as usual without changing the algebra. We can add $D$-terms $\hbar D_i$ to $H_i$ for $i=1,\ldots,N+1$, and we get our class of super quantum Airy structures. Note that for all of those, $\theta^0$ appears as an extra fermionic variable.
\end{proof}

\begin{remark}
The cases with $N=-1$ and $N=0$ of Proposition \ref{prop:SVrho} are interesting. In the case $N=-1$, the bosonic part of the Virasoro generators $L_i$ is almost the same as the $N=-1$ case of Proposition \ref{prop:twisted}, which reproduces the Virasoro constraints satisfied by the Kontsevich-Witten tau function of KdV. The only difference is in the term $\frac{\hbar}{8}$, which is twice that of the bosonic case.

For $N=0$ and $D_1=0$, the bosonic part of $L_i$ also recovers almost exactly the $N=0$ case of Proposition \ref{prop:twisted}, which gives the Virasoro constraints satisfied by the BGW tau function of KdV. Again, the only difference is in the term $\frac{\hbar}{8}$ which is twice its bosonic counterpart, but given that there is an arbitrary constant $D_1$ it could be adjusted accordingly.

Thus both cases may be thought of as supersymmetric generalizations of Kontsevich-Witten and BGW, which deserves further investigation.
\end{remark}


\section{Conclusion and Open Questions}   \label{sec-conclusion}

In this paper we have defined super quantum Airy structures, as a natural supersymmetric generalization of quantum Airy structures. We showed existence and uniqueness of the associated free energy, which is computed by a topological recursion that can be understood as a supersymmetric generalization of the Chekhov-Eynard-Orantin (CEO) topological recursion. We constructed a number of examples of finite-dimensional and infinite-dimensional quadratic super quantum Airy structures.

There are many open questions that deserve further investigation. Here we propose a few, in random order:

\begin{itemize}
\item
For quantum Airy structures that come from the CEO topological recursion, it is well known that the free energies are related to intersection numbers over the moduli space of curves \cite{DBOSS, Einter}. Is there an analogous statement for some classes of super quantum Airy structures? Do they provide invariants of moduli spaces of supersymmetric algebraic curves or super-Riemann surfaces?
\item Super-Virasoro constraints have appeared in enumerative geometry in the context of invariants that involve odd cohomology classes, such as Gromov-Witten invariants for higher genus target curves \cite{Okounkov:2003rf}. Perhaps the appearance of odd cohomology classes is key to uncovering the enumerative meaning of  super quantum Airy structures?
\item Stanford and Witten very recently proposed a supersymmetric generalization of Mirzakhani's recursion relation, in the context of JT gravity \cite{Stanford:2019vob}. Since Mirzakhani's recursion relation can be formulated as an example of a quantum Airy structure \cite{Eynard:2007fi}, are super quantum Airy structures related to the work of Stanford and Witten?
\item
The CEO topological recursion was originally discovered as a solution of the loop equations for matrix models \cite{CE,EO,EO2}. Is the supersymmetric generalization that we propose in this paper related to supereigenvalue models?
\item
In the case of the Airy and Bessel spectral curves, the partition functions calculated by the associated quantum Airy structures construct the Kontsevich-Witten and Br\'ezin-Gross-Witten tau functions for the KdV hierarchy \cite{DYZ}. Can an analogous statement be made for some of the examples of section \ref{sec:VOSA}, perhaps with respect to the super KdV hierarchy?
\item Super quantum Airy structures are $\mathbb{Z}_2$-graded quantum Airy structures. Can $G$-graded quantum Airy structures be defined for more general finite groups $G$?
\item While the CEO topological recursion can be understood as an example of a quantum Airy structure, it was originally formulated in terms of complex analysis of a spectral curve. Is there a similar class of super quantum Airy structures that can be understood as coming from ``super spectral curves''? If so, what is the formulation of the corresponding topological recursion in terms of the geometry of super spectral curves?
\item In this spirit, for a large class of spectral curves, the CEO topological recursion can be used to reconstruct the quantum curve and its associated wave-function \cite{Bouchard:2016obz,Gukov:2011qp}. Can our supersymmetric generalization of topological recursion be used to study super quantum curves \cite{Ciosmak:2017ofd,Ciosmak:2016wpx,Ciosmak:2017omd} ?
\item Can examples of infinite-dimensional super quantum Airy structures be constructed as representations of $N=2$ vertex operators super algebras? Or as representations of VOSAs more generally?
\item In this paper we focused on constructing examples of quadratic super quantum Airy structures. Are there interesting examples of higher order, finite-dimensional, super quantum Airy structures? In the infinite-dimensional case, it would certainly be interesting to generalize the higher order construction of \cite{BBCCN} in terms of representations of $W$-algebras to the supersymmetric realm.
\item In section \ref{s:classification} we proposed a classification scheme for quadratic, finite-dimensional super quantum Airy structures. Can this classification be carried out?
\end{itemize}

This list is certainly not exhaustive. But what is clear is that super quantum Airy structures are interesting conceptually, and that many questions remain. The time is ripe to understand their properties and applications
\footnote{During the review process of this manuscript, (parts of) the third, fourth, fifth, and seventh open questions in the list above were subsequently addressed in \cite{STR}. In particular, it was realized that (suitably polarized) super quantum Airy structures could be used to compute (parts of the) correlation functions for a variety of examples related to 2d supergravity. \cite{STR} also showed a relation between certain families of super quantum Airy structures and non-super quantum Airy structures, which provides an interpretation of the work of Stanford and Witten \cite{Stanford:2019vob} in terms of super quantum Airy structures. We are hoping to solve other open questions in
the near future.}.


\appendix
\section{Computational Proof of Existence}\label{proof of existence}

To prove existence of the free energy associated to a super quantum Airy structure, we need to show that the recursive formulae \eqref{F(i)} and \eqref{F(0)} produce $F_{g,n}[a_1,...,a_n]$ that are $\mathbb{Z}_2$-symmetric under permutations of indices.

Since $\mathbb{Z}_2$-symmetry for permutations among indices in $\Phi$ for $F_{g,n+1}[i,\Phi]$ is obvious from \eqref{F(i)}, the only nontrivial symmetry is for the permutation of $i$ and any other index in $\Phi$. We thus have to show that
\begin{equation}
F_{g,n+2}(i,j,\Phi)=(-1)^{|i||j|}F_{g,n+2}(j,i,\Phi).\label{existence1}
\end{equation}

Let us prove \eqref{existence1} by induction on $2g+n\geq1$. For $2g+n=1$, we have $F_{0,3}(i,j,a)=A_{ija}$, hence \eqref{existence1} holds thanks to \eqref{A}. For $2g+n=2$ there are two cases: $F_{0,4}(i,j,a_1,a_2)$ and $F_{1,2}(i,j)$. It follows from \eqref{F(i)} that $F_{0,4}(i,j,a_1,a_2)$ becomes
\begin{align}
&F_{0,4}(i,j,a_1,a_2) = \nonumber\\
&=B_{ij}^0F_{0,3}(0,a_1,a_2)+(-1)^{|i||j|}\left((-1)^{|i||j|}B_{ij}^kA_{ka_1a_2}+B_{ia_1}^bA_{jba_2}+(-1)^{|a_1||a_2|}B_{ia_2}^bA_{jba_1}\right)\nonumber\\
&=(-1)^{|i||j|}B_{ji}^0F_{0,3}(0,a_1,a_2)+\left((-1)^{|i||j|}B_{ji}^kA_{ka_1a_2}+B_{ja_1}^bA_{iba_2}+(-1)^{|a_1||a_2|}B_{ja_2}^bA_{iba_1}\right)\nonumber\\
&=(-1)^{|i||j|}F_{0,4}(j,i,a_1,a_2),
\end{align}
where we used \eqref{f} and \eqref{BA} for the third equality. Similarly, for $F_{1,2}(i,j)$ we have:
\begin{align}
F_{1,2}(i,j)=&(-1)^{|i||j|}\left(\frac{1}{2}C_i^{bc}A_{jcb}+(-1)^{|i||j|}B_{ij}^cD_c\right)\nonumber\\
=&(-1)^{|i||j|}(-1)^{|i||j|}\left(\frac{1}{2}C_{j}^{bc}A_{icb}+(-1)^{|i||j|}B_{ji}^cD_c\right)\nonumber\\
=&(-1)^{|i||j|}F_{1,2}(j,i),
\end{align}
where we used $F_{1,1}(0)=0$ and \eqref{f} for the second equality and \eqref{CA-BD} for the fourth equality. Therefore, the $F_{g,n+2}(i,j,\Phi)$ are $\mathbb{Z}_2$-symmetric for $2g+n=2$ as well.

Now let us assume $\mathbb{Z}_2$-symmetry for $F_{h,m+2}(i,j,\Phi)$ up to $1\leq2h+m<2g+n$. \eqref{F(i)} can be rewritten as:
\begin{align}
&\hspace{-0.3cm} F_{g,n+2}(i,j,\Phi) =\sum_{c\geq 0}B_{ij}^cF_{g,n+1}(c,\Phi)+\sum_{k=1}^n\sigma_{a_k\subset\{j,\Phi\}}\sum_{c\geq 0}B_{ia_k}^cF_{g,n+1}(c,j,\Phi\backslash a_k)\nonumber\\
&+\frac{1}{2}\sum_{b,c\geq 0}C_i^{bc}F_{g-1,n+3}(c,b,j,\Phi) +\sum_{b,c\geq 0}C_i^{bc}\sum_{g_1+g_2=g}\sum_{\Phi_1\cup \Phi_2=\Phi}\sigma_{\Phi_1\subset\Phi}F_{g_1,n_1+1}(b,j,\Phi_1)F_{g_2,n_2+1}(c,\Phi_2)\nonumber\\
&=B_{ij}^0F_{g,n+1}(0,\Phi)+\sum_{q=1}^dB_{ij}^qF_{g,n+1}(q,\Phi)
+\sum_{k=1}^n(-1)^{|i||j|}\sigma_{a_k\subset\Phi}\sum_{c\geq 0}B_{ia_k}^cF_{g,n+1}(j,c,\Phi\backslash a_k)\nonumber\\
&+\frac{1}{2}(-1)^{|i||j|}\sum_{b,c\geq 0}C_i^{bc}F_{g-1,n+3}(j,c,b,\Phi)\nonumber\\
&+\sum_{b,c\geq 0}C_i^{bc}(-1)^{|j||b|}\sum_{g_1+g_2=g}\sum_{\Phi_1\cup \Phi_2=\Phi}\sigma_{\Phi_1\subset\Phi}F_{g_1,n_1+1}(j,b,\Phi_1)F_{g_2,n_2+1}(c,\Phi_2).\label{existence2}
\end{align}
The first term in \eqref{existence2} is $\mathbb{Z}_2$-symmetric in $(i,j)$ thanks to \eqref{f}. For the second term, we apply \eqref{F(i)} to $F_{g,n+1}(q,\Phi)$. For the other terms, we substitute \eqref{F(i)} into $F_{h,m'+1}(j,\Phi')$ for any $h,\Phi'$ whenever $j$ is the first index. The computation becomes rather tedious; the final result after simplification is summarized below. The terms highlighted in red are $\mathbb{Z}_2$-symmetric in $(i,j)$ thanks to Lemma~\ref{lem:SAS}, while the other terms are manifestly $\mathbb{Z}_2$-symmetric. Therefore, $\mathbb{Z}_2$-symmetry of the $F_{g,n}[a_1,...,a_n]$ produced by  \eqref{F(i)} and \eqref{F(0)}  is proved by induction, and hence the unique free energy associated to a super quantum Airy structure exists.

\begin{align}
&\hspace{-0.3cm} F_{g,n+2}(i,j,\Phi) = \nonumber\\
&=\textcolor{red}{B_{ij}^0} F_{g,n+1}(0,\Phi)+(-1)^{|i||j|}\textcolor{red}{\left(B_{ia_1}^bA_{jba_2}+(-1)^{|a_1||a_2|}B_{ia_2}^bA_{jba_1}+(-1)^{|i||j|}B_{ij}^lA_{la_1a_2}\right)}\delta_{n,2}\delta_{g,0}\nonumber\\
&+(-1)^{|i||j|}\textcolor{red}{\left(\frac{1}{2}C_i^{bc}A_{jcb}+(-1)^{|i||j|}B_{ij}^lD_l\right)}\delta_{n,0}\delta_{g,1}\nonumber\\
&+(-1)^{|i||j|}\sum_{k=1}^n\sum_{c\geq 0}\sigma_{a_k\subset\Phi}F_{g,n}(c,\Phi\backslash a_k)\color{red}{\left(B_{ia_k}^bB_{jb}^c+(-1)^{|c||a_k|}C_i^{cb}A_{jba_k}+(-1)^{|i||j|}B_{ij}^pB_{pa_k}^c \right)} \nonumber\\
&+\frac{1}{2}(-1)^{|i||j|}\sum_{d,c\geq 0}F_{g-1,n+2}(d,c,\Phi)\color{red}{\Bigl(C_i^{cb}B_{jb}^d+(-1)^{|c||d|}C_i^{db}B_{jb}^c+(-1)^{|i||j|}B_{ij}^pC_p^{cd}\Bigr)}\nonumber\\
&+\frac{1}{2}(-1)^{|i||j|}\sum_{g_1+g_2=g}\sum_{\Phi_1\cup \Phi_2=\Phi}\sum_{c,d\geq 0}\sigma_{\Phi_1\subset\Phi}F_{g_1,n_1+1}(c,\Phi_1)F_{g_2,n_2+1}(d,\Phi_2)\nonumber\\
&\;\;\;\;\times\color{red}{\Bigl(C_i^{cb}B_{jb}^d+(-1)^{|c||d|}C_i^{db}B_{jb}^c+(-1)^{|i||j|}B_{ij}^pC_p^{cd}\Bigr)}\nonumber \\
&+\frac{1}{2}\sum_{k,l=1}^n\sum_{b,c\geq 0}\sigma_{\{a_k,a_l\}\subset\Phi}(-1)^{|b||a_l|}F_{g,n}(c,b,\Phi\backslash \{a_k,a_l\})\left((-1)^{|i||j|}B_{ia_k}^bB_{ja_l}^c+B_{ja_k}^bB_{ia_l}^c\right)\nonumber\\
&+\frac{1}{2}\sum_{k=1}^n\sum_{g_1+g_2=g}\sum_{\Phi_1\cup \Phi_2=\Phi\backslash a_k}\sum_{b,c,d\geq 0}\sigma_{\{a_k,\Phi_1\}\subset\Phi}F_{g_1,n_1+2}(c,b,\Phi_1)F_{g_2,n_2+1}(d,\Phi_2)\nonumber\\
&\;\;\;\;\times\left(B_{ja_k}^bC_i^{cd}+(-1)^{|j||i|}B_{ia_k}^bC_j^{cd}\right)\nonumber\\
&+\frac{1}{2}\sum_{k=1}^n\sum_{b,c,d\geq 0}\sigma_{a_k\subset\Phi}F_{g-1,n+2}(b,d,c,\Phi\backslash a_k)\left(C_i^{cd}B_{ja_k}^b(-1)^{|i||b|}+(-1)^{|j||i|}B_{ia_k}^bC_j^{cd}(-1)^{|j||b|}\right)\nonumber\\
&+\frac{1}{2}\sum_{g_1+g_2+g_3=g}\sum_{\Phi_1\cup \Phi_2\cup \Phi_3=\Phi}\sum_{b,c,d,e\geq 0}\sigma_{\{\Phi_1,\Phi_2\}\subset\Phi}F_{g_1,n_1+2}(b,d,\Phi_1)F_{g_2,n_2+1}(c,\Phi_2)F_{g_3,n_3+1}(e,\Phi_3)\nonumber\\
&\;\;\;\;\times \left(C_i^{de}C_j^{bc}(-1)^{|j||e|}+(-1)^{|i||j|}C_j^{de}C_i^{bc}(-1)^{|i||e|}\right)\nonumber\\
&+\sum_{b,c,d,e\geq 0}\Bigl((-1)^{|i||j|}C_i^{de}C_j^{bc}+C_j^{de}C_i^{bc}\Bigr)\times\Biggl(\frac{1}{8}F_{g-2,n+4}(c,b,e,d,\Phi)\nonumber\\
&\hspace{10mm}+\frac{1}{2}\sum_{g_1+g_2=g-1}\sum_{\Phi_1\cup \Phi_2=\Phi}\sigma_{\Phi_1\subset\Phi}F_{g_1,n_1+3}(b,e,d,\Phi_1)F_{g_2,n_2+1}(c,\Phi_2)\nonumber\\
&\hspace{10mm}+\frac{1}{4}\sum_{g_1+g_2=g-1}\sum_{\Phi_1\cup \Phi_2=\Phi}(-1)^{|e||b|}\sigma_{\Phi_1\subset\Phi}F_{g_1,n_1+2}(b,d,\Phi_1)F_{g_2,n_2+2}(c,e,\Phi_2)\Biggr).
\end{align}

\end{document}